\documentclass[fleqn]{article}
\usepackage{verbatim}
\usepackage{subfiles}
\usepackage{setspace}
\usepackage{amsmath}
\usepackage{amssymb}
\usepackage{amsthm}
\usepackage{bm}
\usepackage{color}
\usepackage{graphicx,caption}
\usepackage{subcaption}
\usepackage{cleveref}
\usepackage{epsfig}
\usepackage{epstopdf}
\usepackage{listings}
\usepackage{url}
\usepackage{multirow}
\usepackage{hyperref}
\usepackage{mathtools}
\usepackage[round]{natbib}
\usepackage{tikz}
\usepackage{adjustbox}
\usepackage{pgfplots}

\hypersetup{
    colorlinks = true,
    linkbordercolor={white},
    citecolor = {blue},
}

\newcommand{\ba}{\begin{eqnarray}}
\newcommand{\ea}{\end{eqnarray}}
\newcommand{\bas}{\begin{eqnarray*}}
\newcommand{\eas}{\end{eqnarray*}}
\newtheorem{theorem}{Theorem}
\newtheorem{coro}{Corollary}
\newtheorem{lemma}{Lemma}
\newtheorem{remark}{Remark}
\newtheorem{assumption}{Assumption}
\newtheorem{proposition}{Proposition}
\newtheorem{example}{Example}
\newtheorem{defn}{Definition}
\newcommand{\e}{{\mathbb{E}}}
\newcommand{\E}{{\textup{E}}}
\newcommand{\R}{{\mathbb{R}}}
\newcommand{\N}{{\mathbb{N}}}
\newcommand{\p}{{\mathbb{P}}}
\newcommand{\T}{{\intercal}}
\newcommand{\col}{{\textup{cl}}}
\newcommand{\norm}[1]{\left\lVert#1\right\rVert}
\newcommand{\var}{{\mathbb{V}\rm ar }}

\newcommand{\range}{{\textup{ran}}}
\newcommand{\I}{\mathbb{I}}

\newcommand{\blue}{\color{blue}}
\newcommand{\red}{\color{red}}
\newcommand{\interior}[1]
{%
  {\kern0pt#1}^{\mathrm{o}}%
}

\title{A Backward Simulation Method for\\ Stochastic Optimal Control Problems
\footnote{The authors are thankful to Alexander Schied, Pengfei Li, Degui Li, and Zhaoxing Gao for helpful discussions.}
}
\author{Zhiyi Shen
\footnote{Department of Statistics and Actuarial Science, University of Waterloo. Email Address: zhiyi.shen@uwaterloo.ca}, 
\, and \, Chengguo Weng
\footnote{Department of Statistics and Actuarial Science, University of Waterloo. Email Address: c2weng@uwaterloo.ca}
 }
\date{\today}

\begin{document}

\maketitle
\begin{abstract} 
A number of optimal decision problems with uncertainty can be formulated into a stochastic optimal control framework.
The Least-Squares Monte Carlo (LSMC) algorithm 
is a popular numerical method to approach solutions
of such stochastic control problems as analytical solutions are not tractable in general. 
This paper generalizes the LSMC algorithm proposed in \cite{Shen2017} 
to solve a wide class of stochastic optimal control models.
Our algorithm has three pillars: 
a construction of auxiliary stochastic control model,
an artificial simulation of post-action value of state process, 
and a shape-preserving sieve estimation method
which equip the algorithm with a number of merits including 
bypassing forward simulation and control randomization,
evading extrapolating the value function,
and alleviating computational burden of the tuning parameter selection.
The efficacy of the algorithm is corroborated by an application to pricing equity-linked insurance products.
\end{abstract}

\section{Introduction}
The stochastic optimal control model is a prevalent paradigm for 
solving optimal decision problems with uncertainty 
in a variety of fields, particularly, financial engineering.
In solving discrete-time stochastic optimal control problems,
the Dynamic Programming Principle (DPP) is a prevailing tool 
which characterizes the optimal value function as the solution to
a backward recursive equation system, often known as the \textit{Bellman equation}.
This reduces the stochastic optimization problem into two separate problems:
1) solving a sequence of deterministic optimization problems 
and 2) evaluating the conditional expectation terms in the Bellman equation.
In spite of the theoretical appealingness of the DPP,
there generally does not exist closed-form solution of the Bellman equation,
which impedes the application of stochastic optimal control models to complicated real-world problems.
Recently, a number of numerical methods have been proposed in the literature
to approach the optimal or suboptimal solutions to various stochastic optimal control problems
by combining Monte Carlo Simulation with nonparametric regression methods.

In a statistical setting,
the typical goal of nonparametric regression methods is to estimate the functional form
of the expectation of a response variable conditioning on a covariate variable.
This naturally motivates one to use certain nonparametric regression methods to evaluate the conditional expectation
(also known as the continuation value in the context of pricing Bermudan option)
involved in the Bellman equation where
the value function at the next time step and the state variable at the current time step
are taken as the response and covariate variables, respectively.
Such a ground-breaking idea was incubated in a series of papers including
\cite{Carriere1996}, \cite{Longstaff2001}, and \cite{Tsitsiklis2001},
and the corresponding numerical algorithms are often referred to as the Least-Squares Monte Carlo (LSMC) algorithms.
Since then, the LSMC algorithm has witnessed remarkable popularity in solving optimal stopping problems,
a special class of stochastic control problems;
see, e.g., \cite{Clement2002}, \cite{Stentoft2004}, \cite{Egloff2005}, \cite{Glasserman2004}, \cite{Glasserman2004-2}, 
\cite{Egloff2007}, \cite{Zanger2009}, \cite{Zanger2013}, \cite{Belomestny2009}, \cite{Belomestny2011},
and the references therein.

The problem of solving general stochastic control problems 
by resorting to the LSMC algorithm is considerably more involved.
To understand the crux, let us note that
the LSMC method has two building blocks: 1) a forward simulation of the state process
and 2) a backward updating procedure which employs the nonparametric regression to estimate the continuation value.
In an optimal stopping problem, the evolution of the state process is independent of decision maker's (DM's) action
and therefore, the forward simulation of the sample paths of the state process is relatively straightforward.
In stark contrast to this, in a general stochastic optimal control setting,
the state process is influenced by the DM's action
and accordingly, its simulation is unattainable without specifying the DM's action.
Ideally, one may expect to simulate the state process driven by the optimal action of the DM.
However, the optimal action should be determined by solving the Bellman equation in a backward recursion manner,
which is incongruous with the need of forward simulation in an LSMC algorithm.
To circumvent this, 
\cite{Kharroubi2014} proposes to first draw the DM's action from a random distribution
and then simulate the sample paths of the state process based on the initialized action.
This method has been referred to as the control randomization method in the literature and 
applied in the LSMC algorithm to solve many specific stochastic control problems;
see, e.g., \cite{Cong2016}, \cite{Zhang2018}, and \cite{Huang2016}, among others.
Despite the wide usage of the control randomization method, 
the accuracy of the numerical estimate is impaired over the region with sparse sample points 
\citep[Section 3.3]{Zhang2018}
and the spread of the sample paths is sensitive to the specific way of initializing the action.
In an extreme case, the LSMC algorithm might even miss the optimal solution under a dismal choice of random distribution 
from which the PH's action are drawn; see \cite{Shen2017}, for instance.
It is also notable that 
most literature bind together the control randomization and the forward simulation of the state process.
However, this paper will show that the forward simulation is not imperative in an LSMC algorithm
and the merits of abjuring the forward simulation are extant in several aspects.
The limitations of the control randomization and the consequential forward simulation
will be elaborated in Section \ref{sec:LSMC_review} of this paper.

Besides the simulation of the state process,
the approximation of the conditional expectation term in a Bellman equation is also taxing for several reasons.
Firstly, the prevalent regression methods only warrant the accuracy of the regression estimate over a compact support,
see, e.g., \cite{Newey1997}, \cite{Stentoft2004}, and \cite{Zanger2013},
whereas the state variable generally takes value in an unbounded set.
Some literature compromise to first truncate the domain of the continuation function
and then use extrapolation techniques when the knowledge of the function outside the truncated region is required.
It is worth noting that this problem is not acute in the context of optimal stopping problem
but is severe in a general stochastic control setting.
This is because, in the latter case, one has to traverse all admissible actions,
which calls for the values of the continuation function over a domain that 
is wider than spreading range of sample paths.
Secondly, in order to avoid overfitting or underfitting, 
most nonparametric regression methods
thirst for an appropriate choice of the tuning parameter, 
e.g., the number of basis functions in a linear sieve estimation method (see the sequel Section \ref{sec:sieve_estimation}).
This is often resolved by computationally expensive cross-validation methods, see, e.g., \cite{Li1987}.
However, in view of the extraordinarily large number of simulated paths, 
such a tuning parameter selection procedure is intolerable in implementing the LSMC algorithm.
The aforementioned challenges will be investigated in details in the sequel section.

The contribution of this paper is summarized as follows.
Firstly, we restrain the value set of the state process into a compact set,
which evades the undesirable extrapolating value function estimate 
during the backward recursion of the LSMC algorithm.
The value function accompanying the truncated state process 
is shown to be a legitimate approximation for the primal value function 
under a suitable choice of the truncation parameter.
Secondly, we generalize the idea of \cite{Shen2017} 
to simulate the post-action value of the state process
from an artificial probability distribution. 
This eliminates the need for the forward simulation 
and is consistent with the backward induction nature of the Bellman equation.
The memory as well as time costs of the artificial simulation method
are considerably less than those of the control-randomization-based forward simulation method.
Thirdly, we introduce a shape-preserving sieve estimation method to approximate
the conditional expectation term involved in the Bellman equation.
By exploiting certain shape information of the continuation function,
the sieve estimate is insensitive to the tuning parameter
and accordingly reduces the computational cost of the tuning parameter selection.
We refer to the proposed LSMC algorithm as the Backward Simulation and Backward Updating (BSBU) algorithm.
Finally, we establish the convergence result of BSBU algorithm
which sheds light on how the numerical error propagates over the backward recursion procedure.

This paper is organized as follows.
Section \ref{sec:framework} gives a tour through the LSMC algorithm and shows its challenges
in solving general stochastic optimal control problems.
Section \ref{sec:results} gives the main results of the paper:
a construction of auxiliary stochastic optimal control model,
the BSBU algorithm, and the associated convergence analysis.
Section \ref{sec:VA} applies the BSBU algorithm to the pricing problem
of an equity-linked insurance product and Section \ref{sec:num_experiment}
conducts the corresponding numerical experiments.
Finally, Section \ref{sec:conclusion_BSBU} concludes the paper.

\section{Basic Framework and Motivations}
\label{sec:framework}
\subsection{Stochastic Optimal Control Model}
\label{sec:model_setup}
We restrict our attention to a collection of consecutive time points labeled by $\mathcal{T}:=\{0,1,\dots,T\}$
on which a decision maker (DM) may take action.
The uncertainty faced by the DM is formulated by a probability space $\left(\Omega,\mathcal{F}, \mathbb{P}\right)$ 
equipped with a filtration $\mathbb{F}=\big\{\mathcal{F}_{t}\big\}_{t \in \mathcal{T}}$.
The DM's action is described by a discrete-time stochastic process ${\sf a}=\{a_t\}_{t \in \mathcal{T}_0}$
with $\mathcal{T}_0=\mathcal{T}\backslash \{T\}$.
Let $X=\{X_t\}_{t \in \mathcal{T}}$ be a certain state process valued in $\mathcal{X} \subseteq \mathbb{R}^d$ with $d \in \mathbb{N}$.
Starting from an initial state $X_0 \in \R^d$, it evolves recursively according to the following transition equation:
\ba
\label{transition_eq}
X_{t+1}=S\left(X_{t},a_{t},\varepsilon_{t+1}\right),\ \ 
\text{for}\ \ t=0,1,\dots,T-1,
\ea
where $\varepsilon:=\{\varepsilon_{t+1}\}_{t\in \mathcal{T}_{0}}$ 
is a sequence of independent random variables valued in $\mathbb{R}^{q}$ with $q \in \mathbb{N}$. 
$\varepsilon_{t+1}$ reflects the uncertainty faced by the DM at time step $t$
and is referred to as \textit{random innovation} in what follows. 
For brevity of notation, in what follows, we compress the dependency of the state process on the action
and the readers should always bear in mind that $X_t$ implicitly depends on the DM's action up to time $t-1$.
We give the formal definition of the action ${\sf a}$ as follows.
\begin{defn}[Admissible Action]
\label{def:control}
We call a discrete-time process
${\sf a}=\{a_t\}_{t\in \mathcal{T}_{0}}$ an admissible action if it satisfies:
\begin{description}
\vspace{-1ex}\item[(i)]
$a_t$ is $\mathcal{F}_{t}$-measurable for $t=0,1,\dots,T-1$;

\vspace{-1ex}\item[(ii)]
$a_t \in A_t(X_t)$ for $t=0,1,\dots,T-1$, 
where $A_t(\cdot)$ is some function valued as a subset of $\mathbb{R}^{p}$ with $p\in \mathbb{N}$.
\end{description}
\end{defn}
The function $A_t(\cdot)$ in the preceding definition corresponds to a certain state constraint on the action taken at time $t$.
Denote by $\mathcal{A}$ the set of the DM's admissible actions.
Consider a discrete-time stochastic optimal control problem in the following form:
\ba
\label{stochastic_control_problem}
V_{0}(X_{0})=\sup_{{\sf a}\in \mathcal{A}} \e \left[\sum_{t=0}^{T-1} \varphi^{t} f_{t}(X_t,a_{t}) + \varphi^{T}f_T\left(X_{T}\right)\right],
\ea
where $\varphi \in (0,1)$ is a certain discounting factor,
$f_t(\cdot,\cdot)$ and $f_T(\cdot)$ are the intermediate and terminal reward functions, respectively.
In order to ensure the well-posedness of the stochastic control problem \eqref{stochastic_control_problem}, 
we impose the following assumption which is conventional in literature, 
see \cite{Rogers2007} and \cite{Belomestny2010} for instance.
\begin{assumption}
\label{assum:moment_condition}
\bas
\sup_{{\sf a}\in \mathcal{A}}  \e \left[\sum_{t=0}^{T-1} f_{t}(X_t,a_{t})\right] < \infty,
\ \ \textup{and}\ \ 
\sup_{{\sf a}\in \mathcal{A}} \e \left[f_T(X_T)\right] < \infty.
\eas
\end{assumption}
The Dynamic Programming Principle states that the value function $V_0(\cdot)$ can be solved recursively:
\ba
\label{Bellman_eq}
\begin{cases}
V_{T}(x)&=f_T(x),\\
V_t(x) &=	  
\sup \limits_{a \in A_{t}(x)} \Big[
        f_{t}(x, a)+ \varphi 
         \bar{C}_{t}(x, a)\Big],
\ \ \text{for}\ \ t=0,1,\dots,T-1,
\end{cases} 
\ea
where 
\ba
\label{continuation_fun}
\bar{C}_{t}(x, a)=\e \left[V_{t+1}\left(X_{t+1}\right)\Big| X_{t}=x, a_{t}=a\right].
\ea

We proceed by rewriting the transition equation \eqref{transition_eq} into the following form:
\ba
\label{transition_eq-2}
S(X_{t},a_{t},\varepsilon_{t+1})= H\big(K(X_{t},a_{t}),\varepsilon_{t+1}\big),
\ea
where $H(\cdot,\cdot): \R^{r+q} \longrightarrow \R^{d}$ and $K(\cdot,\cdot): \R^{d+p} \longrightarrow \R^{r}$
are some measurable functions with $r \in \N$. It is worth stressing that any transition function $S(\cdot,\cdot,\cdot)$ can be rewritten into the above form
since one may choose $K(\cdot,\cdot)$ as identity function (i.e., $K(x,a)=(x,a)^{\T}$) and the above equation holds trivially.
Nevertheless, it is instructive to introduce the function $K(\cdot,\cdot)$ as it brings the benefit of dimension reduction.
We will explain this more in the sequel.
Combing Eqs. \eqref{continuation_fun} and \eqref{transition_eq-2}, we get
\bas
\bar{C}_{t}\left(X_{t}, a_{t}\right)
=\e \left[V_{t+1}\big(H\left(X_{t^{+}},\varepsilon_{t+1}\right)\big)\Big| X_{t^{+}}=K\left(X_{t}, a_{t}\right)\right].
\eas
Hereafter, we call $X_{t^{+}}$ the \textit{post-action} value of the state process $X_t$ at time $t$. 
It constitutes an essential component in the LSMC algorithm proposed in Section \ref{sec:algorithm}.
Define function 
\ba
\label{aux_continuation_fun}
C_{t}(k):=
\e \Big[V_{t+1}\left(X_{t+1}\right)\Big| X_{t^{+}}=k\Big]
=\e \Big[V_{t+1}\big(H\left(k,\varepsilon_{t+1}\right)\big)\Big].
\ea
We observe the following relationship between $\bar{C}_t(\cdot,\cdot)$ and $C_t(\cdot)$:
\ba
\label{relation}
\bar{C}_{t}(x, a)=C_{t}\big(K(x,a)\big).
\ea
The crucial implication of the above relation is that
it suffices to recover the functional form of $C_t(\cdot)$ in order to evaluate $\bar{C}_t(\cdot,\cdot)$ 
since $K(\cdot,\cdot)$ is known at the first hand.
The motivation of rewriting the transition equation into Eq. \eqref{transition_eq-2} is now clear:
$K(\cdot,\cdot)$ maps a $(d+p)$-dimensional vector into a $r$-dimensional vector, which compresses the dimension if $r<d+p$
and it is more efficient to recover the function $C_t(\cdot)$ than $\bar{C}_t(\cdot,\cdot)$ due to such a dimension reduction.
It is also worth noting that $C_t(\cdot)$ is solely determined by the probability distribution of $\varepsilon_{t+1}$
according to Eq. \eqref{aux_continuation_fun}.
This implies that \textit{it is not necessary to know the exact distribution of $X_{t^{+}}$ in the evaluation of the function $C_t(\cdot)$.}
In view of the relation \eqref{relation}, the Bellman equation \eqref{Bellman_eq} can be equivalently written as
\ba
\label{Bellman_eq-2}
\begin{cases}
V_{T}(x)&=f_T(x),\\
V_t(x) &=	  
\sup \limits_{a \in A_{t}(x)} \Big[
        f_{t}(x, a)+ \varphi 
         C_{t}\big(K(x,a)\big)\Big],
\ \ \text{for}\ \ t=0,1,\dots,T-1.
\end{cases} 
\ea
The above equation system states that, given the value function at time step $t+1$, one may first evaluate continuation function according to Eqs. \eqref{aux_continuation_fun} and \eqref{relation}
and then obtain the value function at time step $t$ via solving an optimization problem in the second line of Eq. \eqref{Bellman_eq-2}.
The information propagation behind the above recursive procedure is illustrated in Figure \ref{fig:propagation-0}.

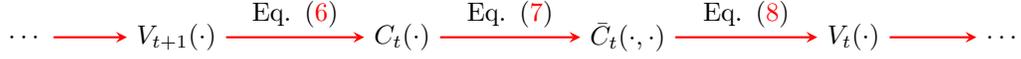
\begin{figure}
\centering
\begin{tikzpicture}[>=stealth,thick]
\draw (-2,0) node (e) {$\cdots$};
\draw node (a) {$V_{t+1}(\cdot)$};
\draw[->, thick, red] (e) to (a);
\draw (3,0) node (b) {$C_t(\cdot)$};
\draw[->, thick, red] (a) to node [above, align=center, text width=2cm,black] {Eq. \eqref{aux_continuation_fun}} (b);
\draw (6,0) node (c) {$\bar{C}_{t}(\cdot,\cdot)$};
\draw[->, thick, red] (b) to node [above, align=center, text width=2cm,black] {Eq. \eqref{relation}} (c);
\draw (9,0) node (d) {$V_{t}(\cdot)$};
\draw[->, thick, red] (c) to node [above, align=center, text width=2cm,black] {Eq. \eqref{Bellman_eq-2}} (d);
\draw (11,0) node (f) {$\dots$};
\draw[->, thick, red] (d) to (f);
\end{tikzpicture}
\captionsetup{width=0.9\textwidth}
\caption{A diagram for backward information propagation in solving the Bellman equation.}
\label{fig:propagation-0}
\end{figure}

\subsection{A Tour Through LSMC Algorithm}
\label{sec:LSMC_review}
We proceed by briefly reviewing the Least-squares Monte Carlo (LSMC) algorithm. 
We will show its limitations in several aspects which motivate the algorithm we will propose in the subsequent sections.
\subsubsection{``Forward simulation and backward updating" (FSBU) algorithm}
There has been voluminous literature on the LSMC for optimal stopping problem,
while the literature on the LSMC for general stochastic optimal control problem is thin.
Most literature addresses the LSMC for the stochastic control problems arising in some specific applications,
see, e.g.,
\cite{Carmona2010}, \cite{Barrera2006}, 
\cite{Huang2016}, \cite{Shen2017}, 
\cite{Cong2016}, and \cite{Zhang2018}, among others.
An LSMC algorithm for a class of stochastic control problem is developed in \cite{Belomestny2010}.

For most variants of the LSMC algorithm, they can be decomposed into two pillars: 
(i) a forward simulation of the state process and (ii) a backward updating of control policies.
We review these algorithms in a unified paradigm as follows.
\begin{enumerate}
\item
\textbf{Initiation:}\quad
Set $V_{T}^{\E}(x)=f_T(x)$.
For $t=T-1,T-2,\dots,0$, do the two steps below.

\item
\textbf{Forward Simulation:}
\begin{enumerate}
\item[2.1] \textbf{Control randomization} \quad
Generate a random sample of the DM's action up to time step $t$: 
\bas
{\sf a}_{0:t}^{M}:=\left\{\left(a_{0}^{(m)},\dots, a_{t}^{(m)}\right),\ m=1,2,\dots,M\right\}
\eas
with each $a_{t}^{(m)}$ generated by a certain heuristic rule.

\item[2.2] \textbf{Simulation of state process} \quad
Simulate a random sample of the random innovations:
\bas
\left\{\left(\varepsilon_{1}^{(m)},\dots, \varepsilon_{t+1}^{(m)}\right),\ m=1,2,\dots,M\right\}.
\eas
The sample of the state process up to time step $t+1$ is given by
\bas
\mathbf{X}_{1:t+1}^{M}:=\left\{X_{1:t+1}^{(m)}:=\left(X_{1}^{(m)},\dots, X_{t+1}^{(m)}\right),\ m=1,2,\dots, M\right\},
\eas
where $X_{n}^{(m)}=S\left(X_{n-1}^{(m)}, a_{n-1}^{(m)}, \varepsilon_{n}^{(m)}\right)$ for $n=1,2,\dots,t+1$.
\end{enumerate}
\item \textbf{Backward Updating}
\begin{enumerate}
\item[3.1] \textbf{Regression} \quad
Given a numerical estimate of value function at time step $t+1$, denoted by $V_{t+1}^{\textup{E}}(\cdot)$, construct the random sample
\bas
\mathbf{Y}_{t+1}^{M}:=\left\{V_{t+1}^{\textup{E}}\left(X_{t+1}^{(m)}\right),m=1,2,\dots, M\right\}.
\eas
Further construct a random sample of post-action value of the state process as follows:
\ba
\label{sample_post_action}
\mathbf{X}_{t^{+}}^{M}:=\left\{X_{t^{+}}^{(m)}:=K\left(X_{t}^{(m)},a_{t}^{(m)}\right),m=1,2,\dots, M\right\}.
\ea
Take $\mathbf{Y}_{t+1}^{M}$ and $\mathbf{X}_{t^{+}}^{M}$ as the samples of response variable and regressor, respectively, 
and employ a certain non-parametric regression to obtain a regression estimate
$C_{t}^{\textup{E}}(\cdot)$ for $C_{t}(\cdot)$.

\item[3.2] \textbf{Optimization} \quad
An estimate for the value function at time step $t$ is given by 
\ba
\label{local_opt}
V_{t}^{\textup{E}}(x) =	  
\sup \limits_{a \in A_{t}(x)} \Big[
        f_{t}(x, a)+ \varphi 
         C_{t}^{\textup{E}}\big(K(x,a)\big)\Big].
\ea
\end{enumerate}

\end{enumerate}

We henceforth call the above algorithm as the Forward Simulation and Backward Updating (FSBU) algorithm.
%
%

\begin{remark}[Randomness of $V_{t}^{\textup{E}}(\cdot)$ and $C_{t}^{\textup{E}}(\cdot)$]
\label{rem:notation_numerical_estimate}
The superscript $\textup{E}$ in $V_{t}^{\textup{E}}(\cdot)$ and $C_{t}^{\textup{E}}(\cdot)$ stresses that they are numerical estimates of the true value function and continuation function, respectively.
Since a certain regression technique is employed to get such numerical estimates,
they essentially depend on the random samples $\mathbf{Y}_{t+1}^{M}$ and $\mathbf{X}_{t^{+}}^{M}$
and hence on all previously generated random samples going from step $T-1$ down to step $t$, i.e., $\mathbf{Y}_{n+1}^{M}$ and $\mathbf{X}_{n^{+}}^{M}$ for $n=t,\dots,T-1$.
Such dependency is suppressed in notation for brevity, but the readers should keep in mind that
both $V_{t}^{\textup{E}}(\cdot)$ and $C_{t}^{\textup{E}}(\cdot)$ are random functions.
\end{remark}

\begin{figure}
\centering
\begin{tikzpicture}[>=stealth,thick]
\tikzset{
    position label/.style={
       below = 3pt,
       text height = 1.5ex,
       text depth = 1ex
    },
   brace/.style={
     decoration={brace, mirror},
     decorate
   }
}
\draw[-latex,thick,blue!80] (0,0) to [out=60,in=120] node[above,midway,black]{$K(\cdot,\cdot)$} (6,0);
\draw[thick] (-1.5,0) to node[below,midway,black] {$(x, a)$} (1.5,0);
\draw[thick] (3.5,0) to node[below,midway,black] {} (6.5,0);
\draw (3.5,-0.1) to (3.5,0.1);
\draw (6.5,-0.1) to (6.5,0.1);
\draw (-1.5,-0.1) to (-1.5,0.1);
\draw (1.5,-0.1) to (1.5,0.1);
\draw [brace,decoration={raise=2ex}] (3.5,0) -- node [position label,yshift=-3ex] {Range of post-action value} (6.5,0);
\draw [decoration={brace}, decorate] (3.5,0.2) -- node [above=3pt] {$\mathcal{D}$} (5,0.2);
\end{tikzpicture}
\captionsetup{width=0.9\textwidth}
\caption{A diagram illustrating the map $K(\cdot,\cdot)$ relating pre-action value to the post-action value.}
\label{fig:extrapolation}
\end{figure}
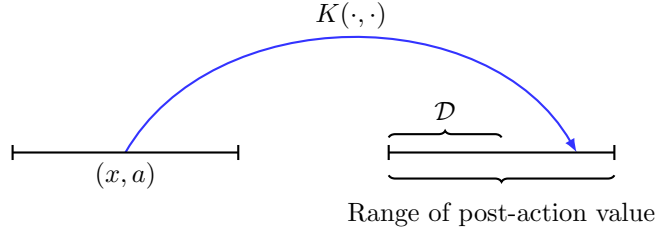

\subsubsection{Challenges}
There are several challenges in implementing the previously introduced FSBU algorithm to solve 
a stochastic control problem. 
Below, we make some comments on the challenges from three aspects.
\begin{description}
\item[(I)] \textbf{Limitation of control randomization}\quad 
As the DM's optimal action is not tractable at the very first but should be solved in the backward updating stage of the algorithm,
Step 1.1 randomly generates a feasible action,
which is referred to as \textit{control randomization} method; see \cite{Kharroubi2014}.
For some selected action ${\sf a}_{0:n}^{M}$, the accuracy of the regression estimate $C_{t}^{\textup{E}}(\cdot)$ can be warranted only over the support of the resulting sampling points $\mathbf{X}_{t^{+}}^{M}$, say $\mathcal{D}$, which might be smaller than those for other actions.
However, in order to solve the optimization problem in Step 2.2 (see Eq. \eqref{local_opt}), 
one requires the knowledge of $C_{t}^{\textup{E}}(\cdot)$ 
over the range of post-action value $X_{t^{+}}=K(X_{t},a_{t})$ for all feasible actions $a_{t}$
because all possible values of the action should be traversed and taken as the input of the function $C_{t}^{\textup{E}}(K(x,a))$
in evaluating $V_{t}^{\textup{E}}(x)$;
see Figure \ref{fig:extrapolation} for a graphical illustration.
As a compromise, one may use certain extrapolation methods 
to infer the value of $C_{t}^{\textup{E}}(\cdot)$ outside the region $\mathcal{D}$,
which incurs extra error and is hard to justify its legitimacy.

\item[(II)] \textbf{Cost of forward simulation}\quad
It is notable that, at time step $t$ of the above FSBU algorithm, 
a new random sample of the state process that is \textit{independent} of the sample at time step $t+1$ is simulated;
see Figure \ref{fig:forward_simulation} for a graphical illustration.
This is required in order to apply the nonparametric regression theory to establish the convergence result,
see Section 2.3 of \cite[p.~511]{Zanger2013} for instance.
On the contrary, using a single sample causes in-sample bias
because the numerical estimate of value function obtained at time step $t+1$ $V_{t+1}^{\textup{E}}(\cdot)$
is correlated with $\mathbf{X}_{t^{+}}^{M}$;
see, e.g., Section 3.1 of \cite{Choi2018} and the earlier Remark \ref{rem:notation_numerical_estimate}.
The total time cost in a forward simulation procedure of the above LSMC algorithm is of $O(T^2)$\footnote{Suppose the time cost of simulating a path over each time interval $[t,t+1]$ is $\mathcal{C}$. Then the time cost of simulating a whole path up to time step $n$ is about $n\times \mathcal{C}$, and the forward simulation in the whole LSMC algorithm has approximate time cost of $\mathcal{C}(1+2+\dots+T)=T(T+1)\mathcal{C}/2$, accordingly.}.
Simulating the whole path of state process can be time-consuming especially when one uses some approximation schemes to simulate general stochastic differential equations\footnote{It is the authors' experience that a single simulation of $10^5$ paths of the Heston model over a $10$-year period takes $365$ seconds by using the R package ``\texttt{msde}" on a MacBook Pro (2.8 GHz Intel Core i7).}.
Besides the issue of time cost, the memory cost in a single simulation is of $O\left(dT\right)$ with $T$ and $d$ being the number of time steps and dimensionality of the state process, respectively, which is sizable for a large $T$.

\item[(III)] \textbf{Choice of regression technique}\quad
Despite the voluminous literature on nonparametric regression,
the choice of the nonparametric regression method in Step 2.1 should be meticulous. 
In the above FSBU algorithm,
the sample size in the regression problem corresponds to the number of simulated paths and is generally recommended in the literature be chosen larger than one hundred thousand, which makes most regression methods computationally \textit{prohibitive}. To be specific, the local methods such as local-polynomial regression are clearly not wise choices as they require running a regression at each sample point.
It is worthy to point out that even computing a \textit{single} point in the sample $\mathbf{Y}_{t}^{M}$ is fairly time-consuming as it involves a local optimization problem (see Eq. \eqref{local_opt}).
Furthermore, the nuisance of high memory cost also burdens most nonparametric regression methods.
For example,
the kernel regression and Isotonic regression methods require storing all sample points
in order to recover the functional form of the regression function over some support,
and the memory cost is \textit{extraordinarily} large， accordingly.
The above two thorny issues are escalated by the fact that almost all nonparametric regression techniques involve a computationally-intensive cross-validation procedure to determine the tuning parameter (e.g., the bandwidth in local regression methods and the number of basis functions in global regression methods) in order to avoid overfitting or underfitting. 
\end{description}

\begin{figure}
\centering
\begin{tikzpicture}[>=stealth,thick]
\draw[->, thick] (-1,0) to (10,0);
\draw (-1,1pt) -- (-1,-3pt)
node[anchor=north] {0};
\draw (0,1pt) -- (0,-3pt)
node[anchor=north] {1};
\draw (8,1pt) -- (8,-3pt);
\draw (8,-0.3) node {$t+1$};
\draw (7,1pt) -- (7,-3pt)
node[anchor=north] {$t$};

\draw (10,-0.5) node {Time step};

\draw[->, thick,red] (0,1) to node [above, align=center, text width=2cm,black] {$\red \mathbf{X}_{1:t+1}^{M}$} (8,1);
\draw[->, dashed,thick,blue] (0,-1) to node [below, align=center, text width=2cm,black] {$\blue \mathbf{X}_{1:t}^{M}$} (7,-1);
\end{tikzpicture}
\captionsetup{width=0.9\textwidth}
\caption{A diagram for illustrating the forward simulation of the state process.
The solid line corresponds to a simulated sample of state process at time step $t$ of the LSMC algorithm. The dashed line corresponds to another \textit{independent} simulated sample of the state process at time step $t-1$.
}
\label{fig:forward_simulation}
\end{figure}
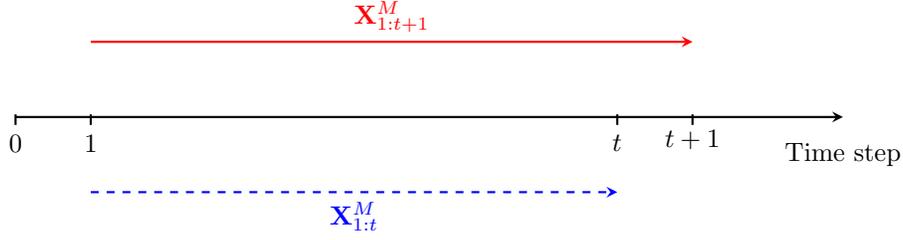

\subsubsection{Motivations}
In view of the previous items (I)--(III), 
the thrust behind this paper is to explore possible answers to the following questions:
\begin{description}
\item[(Q1)] \textit{How to avoid theoretically shaky extrapolation?}


\item[(Q2)] \textit{Is it possible to bypass the forward simulation in an LSMC algorithm?}

\item[(Q3)] \textit{Is there a regression method that is insensitive to tuning parameter?}

\end{description}
In terms of \textbf{(Q1)}, in the sequel section, 
we will construct an auxiliary stochastic control problem 
where the accompanying state process only takes values in a bounded set.
This construction sidesteps extrapolating the regression function
outside the region where the sample distributes.
In response to \textbf{(Q2)},
we will propose to directly simulate the post-action value of state process.
For \textbf{(Q3)},
we will introduce a shape-preserving sieve estimation method to infer the continuation function.
The resulting sieve estimate, on the one hand, is insensitive to the tuning parameter,
and on the other hand, preserves certain shape properties of the continuation function.

\section{Main Results}
\label{sec:results}
\subsection{Localization and Error Bound}
\label{sec:truncation}
As commented in the item ``Limitation of control randomization" in the last section,
it is necessary to know the value of the continuation function
over the whole range of post-action value of the state process 
which is wider than the set where the regression sample suffuse.
It is notable that the range of post-action value is unbounded
if the state process takes value in an unbounded set,
which is particularly the case in many finance applications.
Therefore, it is generally inevitable to infer the continuation function 
outside the support of the sample and
the error incurred by extrapolating the regression estimate is hard to quantify.

The aim of this subsection is to find a certain way to circumvent the unsound extrapolation in the implementation of an LSMC algorithm.
The pivotal idea is to first construct an auxiliary stochastic optimal control problem 
where the accompanying state process only takes values in a bounded set 
and then show the discrepancy between the auxiliary problem and the primal one is marginal in a certain sense.
To formalize the idea, we let $\mathcal{X}_{R}$ be a bounded subset of the set $\mathcal{X}$
where the subscript $R$ denotes a certain truncation parameter.
Further denote $\mathring{\mathcal{X}}_{R}$ (resp. $\partial \mathcal{X}_{R}$) as the interior (resp. boundary) of $\mathcal{X}_{R}$.
Given the initial state $X_0 \in \mathring{\mathcal{X}}_{R}$,
define the following stopping time:
\ba
\label{first_exit_time}
\tau^{R}:=\inf \left\{t\in \mathcal{T} \ \Big|\ X_t \notin \mathring{\mathcal{X}}_{R} \right\},
\ea
with the convention: $\tau^{R}=\infty$ if $X_t \in \mathring{\mathcal{X}}_{R}$ for all $t \in \mathcal{T}$.
Let $\col\left(\mathcal{X}_{R}\right)$ be the closure of the set $\mathcal{X}_{R}$
and assume it to be strictly convex.

We recursively define an auxiliary state process $X^{R}:=\left\{X_{t}^{R}\right\}_{t\in \mathcal{T}}$ as follows: 
\ba
\label{aux_state_process}
\begin{cases}
X_{0}^{R}&=X_{0},\\
X_{t}^{R}&=X_t\I_{\left\{\tau^{R}>t\right\}} + \mathcal{Q}\left(X_{\tau^{R} \wedge t}\right)\I_{\left\{\tau^{R}\leq t\right\}}, 
\ \ \textup{for}\ \ t=1,2,\dots,T,\\
\end{cases}
\ea
where $\mathcal{Q}(x)= \arg \inf_{y \in \textup{col}\left(\mathcal{X}_{R}\right)}\norm{y-x}$ 
with $\norm{\cdot}$ denoting the Euclidean $\ell_2$-norm.
Since $\col\left(\mathcal{X}_{R}\right)$ is a compact and strictly convex set, $\mathcal{Q}(x)$ is unique and lies 
on the boundary set $\partial \mathcal{X}_{R}$ for $x \notin \mathring{\mathcal{X}}_{R}$.

Below we give some interpretations regarding the  auxiliary state process defined in the above Eq. \eqref{aux_state_process}.
The primal state process $X$ coincides with the auxiliary state process $X^{R}$ until the stopping time $\tau^{R}$.
Once the primal state process passes through the interior of the truncated domain,
the auxiliary state process freezes at a certain point in the boundary set $\partial \mathcal{X}_{R}$ thereafter.
The evolution mechanisms of the primal and auxiliary state processes are illustrated in Figure \ref{fig:aux_state_process}.
The following proposition gives the transition equation of $X^{R}$.

\begin{proposition}
\label{prop:transition_eq_aux_process}
The auxiliary state process $X^{R}$ defined by Eq. \eqref{aux_state_process} admits the following transition equation across each time point: $X_{0}^{R}=X_{0}$ and 
\ba
\label{transition_aux_state_process}
X_{t+1}^{R}=X_{t}^{R}\I_{\left\{X_{t}^{R}\in \partial \mathcal{X}_{R}\right\}}+
\tilde{H}\Big(K\left(X_{t}^{R},a_{t}\right),\varepsilon_{t+1}\Big)
\I_{\left\{X_{t}^{R} \in  \mathring{\mathcal{X}}_{R} \right\}}, 
\ \ \textup{for}\ \ t=0,1,\dots,T-1,
\ea
where
\ba
\label{H}
\tilde{H}(k,e)=
\begin{cases}
\mathcal{Q}(H(k,e)), & \text{if}\ \ H(k,e) \notin \mathring{\mathcal{X}}_{R},\\
H(k,e), & \text{otherwise},
\end{cases}
\ea
and 
$K(\cdot,\cdot)$ is the transition equation relating the pre-action and post-action values of the primal state process 
defined in Eq. \eqref{transition_eq-2}.
\end{proposition}
The proof of the above proposition is relegated to Appendix \ref{app:proof_prop-1} for clarity of presentation.
The preceding Eq. \eqref{transition_aux_state_process} essentially states that 
$X^{R}$ is a Markov chain by itself, and accordingly, 
it is the sole state process of the auxiliary stochastic control model defined in the sequel.

\newcommand{\Lathrop}[7]{
\pgfmathsetseed{#7}%
\draw[#4] (0,0)
\foreach \x in {1,...,#6}
{   -- ++(#2,0.02+rand*#3)
}
coordinate (tempcoord) {};
\pgfmathsetmacro{\remaininglength}{(#1-#6)*#2}
\draw[#4] (tempcoord) -- ++ (\remaininglength,0) node[right] {#5};
\pgfmathsetmacro{\remaininglength}{#6*#2}
\draw[dashed] (tempcoord) --++ (-\remaininglength-0.5,0) node (end) {};
}

\newcommand{\Emmett}[7]{
\pgfmathsetseed{#7}%
\draw[#4] (0,0)
\foreach \x in {1,...,#1}
{   -- ++(#2,0.02+rand*#3)
}
\foreach \x in {1,...,#6}
{   -- ++(#2,-0.01+rand*#3)
}
node[right] {#5};
}

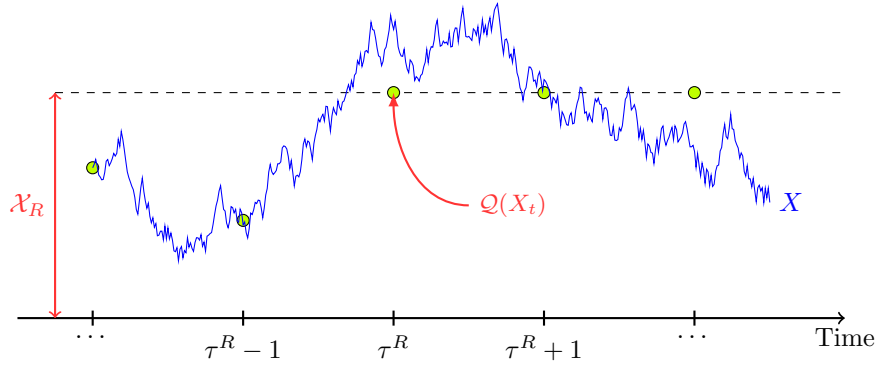
\begin{figure}
\centering
\begin{tikzpicture}
\draw[->,thick] (-1,-2) -- (10,-2) node [below] {Time};
\draw[dashed] (-0.5,1) to (10,1);
\draw[red!80,<->,thick] (-0.5,1) to node [left] {$\mathcal{X}_{R}$} (-0.5,-2);

\draw[thick] (0,-1.9) -- (0,-2.1) node [below] {$\dots$};
\draw[thick] (2,-1.9) -- (2,-2.1) node [below] {$\tau^{R}-1$};
\draw[thick] (4,-1.9) -- (4,-2.1) node [below] {$\tau^{R}$};
\draw[thick] (6,-1.9) -- (6,-2.1) node [below] {$\tau^{R}+1$};
\draw[thick] (8,-1.9) -- (8,-2.1) node [below] {$\dots$};
\draw [fill=lime](0,0) circle (.08);
\draw [fill=lime](2,-.7) circle (.08);
\draw [fill=lime](4,1) circle (.08);
\draw [fill=lime](6,1) circle (.08);
\draw [fill=lime](8,1) circle (.08);
\Emmett{200}{0.02}{0.2}{blue}{$X$}{250}{1};
\draw[-latex,thick,red!80](5,-.5)node[right]{\small $\mathcal{Q}(X_t)$}
	  to[out=180,in=270] (4,1);
\end{tikzpicture}
\captionsetup{width=0.9\textwidth}
\caption{Graphical illustration of the evolution mechanisms of $X$ and $X^{R}$. It is notable that $X$ might evolve continuously between two discrete time points $t$ and $t+1$. The stopping time $\tau^{R}$ corresponds to the first time point upon which $X_t$ stays outside of $\mathcal{X}_{R}$ among all discrete time points $\{0,1,\dots,T\}$. The circles correspond to a path of $X^{R}$.}
\label{fig:aux_state_process}
\end{figure}

Let $\mathcal{A}^{R}$ be the set of all admissible actions for the auxiliary state process which is defined as:
\bas
\mathcal{A}^{R}:=
\left\{{\sf a}=\{a_t\}_{t\in \mathcal{T}_{0}}\ \big|
\ a_{t}\ \textup{is}\ \mathcal{F}_{t}-\textup{measurable},
\ a_{t} \in A_{t}\left(X_{t}^{R}\right),
\ \textup{for} \ t \in \mathcal{T}_{0}\right\}.
\eas
Relative to the primal stochastic optimal control problem \eqref{stochastic_control_problem}, we consider the following auxiliary problem:
\ba
\label{aux_stochastic_control_problem}
\tilde{V}_{0}(X_{0})=\sup_{{\sf a}\in \mathcal{A}^{R}} 
\e \left[\sum_{t=0}^{T-1} \varphi^{t} 
f_{t}\left(X_t^{R},a_{t}\right)
+ \varphi^{T}f_T\left(X_{T}^{R}\right)\right],
\ea
where $X^{R}=\left\{X_{t}^{R}\right\}_{t\in \mathcal{T}}$ is defined recursively by Eq. \eqref{transition_aux_state_process}
for any given action ${\sf a}$.
Since the state process $X^{R}$ freezes once it reaches the boundary set $\partial \mathcal{X}_{R}$,
the value function in Eq. \eqref{aux_stochastic_control_problem} is given by
\ba
\label{aux_value_fun_boundary}
\tilde{V}_{t}(x) =\sum_{n=t}^{T-1}\varphi^{n-t} f_{n}\big(x;a_{n}^{*}(x)\big) + \varphi^{T-t}f_T(x),
\ \ \text{for}\ \ x \in \partial \mathcal{X}_{R},\ t\in \mathcal{T},
\ea 
with $a_{n}^{*}(x)\in \arg \max_{a \in A_n(x)}f_{n}(x;a)$.
Over the interior of the truncated domain,
the above value function $\tilde{V}_{0}(\cdot)$ can be solved in a similar backward recursion way as $V_{0}(\cdot)$ does, that is, 
\ba
\label{Bellman_eq-3}
\begin{cases}
\tilde{V}_{T}(x)&=f_T(x),\\
\tilde{V}_{t}(x) &=	  
\sup \limits_{a \in A_{t}(x)} \Big[
        f_{t}(x, a)+ \varphi 
         \tilde{C}_{t}\big(K(x,a)\big)\Big],
\ \ \text{for}\ x \in \mathring{\mathcal{X}}_{R},\ t=0,1,\dots,T-1,
\end{cases} 
\ea
where $\tilde{C}_{t}(\cdot)$ is defined in line with Eq. \eqref{aux_continuation_fun}
with $H(\cdot,\cdot)$ replaced by $\tilde{H}(\cdot,\cdot)$.
It is worth noting that, in evaluating $\tilde{C}_{t}\big(K(x,a)\big)$,
the knowledge of $\tilde{V}_{t+1}(\cdot)$ over $\partial \mathcal{X}_{R}$ might be in need,
and in such a situation, Eq. \eqref{aux_value_fun_boundary} is invoked.

We make some comparisons between Eq. \eqref{Bellman_eq-3}
and the Bellman equation \eqref{Bellman_eq-2} associated with the primal stochastic control model.
Firstly, in both equations, the state constraint $A_t(\cdot)$, transition equation between pre-action and post-action values $K(\cdot,\cdot)$, and reward functions are exactly the same.
Secondly, the value function $\tilde{V}_{t}(\cdot)$ is solely defined on a bounded set $\col \left(\mathcal{X}_{R}\right)$,
whilst $V_{t}(\cdot)$ is defined on the set $\mathcal{X}$ which might be unbounded in many financial applications
as the primal state process $X$ may correspond to a certain risky asset valued 
on the whole positive real line.

We will characterize the discrepancy between the value functions $\tilde{V}_{t}(\cdot)$ and $V_{t}(\cdot)$ in a certain sense.
To this end, it is necessary to impose some assumptions on the state process and the reward functions.
\begin{assumption}
\label{assum:tail_prob}
Let $X_{0}^{R}=X_{0} \in \mathring{\mathcal{X}}_{R}$.
There exists a measurable function $\mathcal{E}(\cdot,\cdot): \mathring{\mathcal{X}}_{R} \times \R_{>0} \longrightarrow [0,1]$ satisfying 
\ba
\label{tail_prob}
\inf_{{\sf a} \in \mathcal{A}}
\p \Big[ X_{t}=X_{t}^{R},\ \ \textup{for all}\ \ 1\leq t \leq T \Big]
\geq 1- \mathcal{E}(X_0,R).
\ea
\end{assumption}
$\mathcal{E}(X_0,R)$ in Eq. \eqref{tail_prob}
gives an upper bound for the probability that the auxiliary state process disagrees with the primal at some time 
before maturity regardless of the DM's action.
Since the primary difference between the auxiliary and primal value functions stems from the disparity between the associated state processes,
it is not surprising that the above inequality \eqref{tail_prob} plays an important role in characterizing the approximation error of $\tilde{V}_{t}(\cdot)$
as we will see later in the proof of Theorem \ref{thm:truncation_error_bound}.
The expression of $\mathcal{E}(X_0,R)$ should be specified for each specific application.

\begin{assumption}
\label{assum:growth_condition}
\begin{description}
\item[(i)]
There exists a measurable function $B(\cdot): \R^{d} \longrightarrow \R_{>0}$ and a generic constant $\zeta$ independent of $t$ and $R$ such that 
$\big|f_T(x)\big|^2 \leq B(x)$, $\sup_{a \in A_{t}(x)} \big|f_{t}(x,a)\big|^2 \leq B(x)$, 
\ba
\label{growth_condition}
\sup_{{\sf a} \in \mathcal{A}}
\e\left[B(X_{t+1})\right]
\leq \zeta,
\ \ \textup{and}\ \  
\sup_{{\sf a} \in \mathcal{A}}
\e\left[f_T(X_{T})\right]
\leq \zeta,
\ \ \textup{for all}\ \ t \in \mathcal{T}_{0}.
\ea

\vspace{-1ex}\item[(ii)]
There exists a measurable function $\xi(\cdot): \R_{>0} \longrightarrow \R_{>0}$ satisfying
\bas
\sup_{x \in \mathcal{X}_{R}} \left( \sup_{a \in A_{t}(x)}\big|f_{t}(x,a)\big|^2 \right)\leq \xi(R),
\ \ \textup{for all}\ \ t \in \mathcal{T}_{0},
\ \ \textup{and}\ \ 
\sup_{x \in \mathcal{X}_{R}} \big|f_T(x)\big|^2 \leq \xi(R).
\eas
\end{description}
\end{assumption}
In many applications $B(x)$ has a polynomial form 
and in such a situation,
the above assumption states that the reward functions are bounded by a certain polynomial from the above uniformly in $t$.
In the context of pricing financial products, 
this assumption says that the policy payoffs have a polynomial growth rate.

The following theorem quantifies the error stemming from using the auxiliary problem \eqref{aux_stochastic_control_problem} as a proxy for the primal stochastic control model \eqref{stochastic_control_problem}.
\begin{theorem}[Truncation Error Estimate]
\label{thm:truncation_error_bound}
Suppose Assumptions \ref{assum:moment_condition}, \ref{assum:tail_prob}, and \ref{assum:growth_condition} hold. 
Then
\ba
\label{truncation_error}
\left|V_{0}(X_0)-\tilde{V}_{0}(X_0)\right| \leq  T 
\sqrt{2\big(\xi(R) + \zeta \big) \mathcal{E}(X_0,R)}.
\ea
\end{theorem}
The proof of the above theorem is relegated to Appendix \ref{app:proof_thm1}.
The inequality \eqref{truncation_error} can be understood as follows.
The term $\big(\xi(R) + \zeta \big) \mathcal{E}(X_0,R)$ corresponds to an upper bound for the discrepancy between the reward functions
of the two stochastic control models \eqref{stochastic_control_problem} and \eqref{aux_stochastic_control_problem} at each time step.
Since such a difference primarily stems from replacing the primal state process $X$ by $X^{R}$,
it is not surprising that the term $\mathcal{E}(X_0,R)$ appears in the error estimate.
Furthermore, the two terms $\xi(R)$ and $\zeta$ correspond to certain upper bounds of
the magnitudes of the reward terms $f_{t}^2\left(X_{t}^{R},a_{t}\right)$ and $f_{t}^2\left(X_{t},a_{t}\right)$, respectively,
and therefore a square root arises in the inequality \eqref{truncation_error}. 
Finally, the discrepancy between the two value functions is amplified as the time horizon is prolonged, 
which is reflected by the existence of factor $T$ in the above error estimate.

\subsection{A Backward Simulation and Backward Updating Algorithm}
\label{sec:algorithm}
\subsubsection{Simulation of post-action value}
In this subsection we propose an LSMC algorithm which simulates
the state process without referring to the optimal action.
Recall from Step 2.1 of the FSBU algorithm in Section \ref{sec:LSMC_review} that the ultimate goal of simulating
the state process is generating a random sample of the post-action value of the state process 
which acts as a crucial input for the regression step.
This naturally inspires us to directly simulate the post-action value $X_{t^{+}}$
from an artificial probability distribution. 
The term ``artificial" stresses the fact that
such a distribution might not coincide with the distribution of $X_{t^{+}}$ under the optimal action process.

\begin{figure}
\begin{adjustbox}{max totalsize={0.7\textwidth}{.5\textheight},center}
\begin{tikzpicture}
\draw [fill=black!20!red!20!white]  plot[smooth, tension=.7] coordinates {(-4,2.5) (-3,3) (-2,2.8) (-0.8,2.5) (-0.5,1.5) (0.5,0) (0,-2)(-1.5,-2.5) (-4,-2) (-3.5,-0.5) (-5,1) (-4,2.5)};
\draw [fill=black!20!blue!20!white]  plot[smooth, tension=.7] coordinates {(-3,2.2) (-3.5,2.2) (-4,1.5) (-3,0) (-3,-1) (-2.5,-1.7) (-1.5,-2) (-0.5,-1) (-0.5,0) (-1,1.5)  (-3,2.2)};
\node at (-2,0.5) {$\mathcal{K}_{t,R}$};
\node at (-2,2.5) {$\widehat{\mathcal{K}}_{t,R}$};
\draw [thick, red, fill=black!20!blue!20!white]  plot[smooth, tension=.7] coordinates {(3,2.5) (2.5,2.2) (2,1.5) (3,0) (3,-1) (3.5,-1.7)  (5.5,-1) (5.5,0) (5,1.5) (3.5,2.5) (3,2.5)};
\draw[-latex,thick,blue!60] (-1,.5) to node[above,midway,black]{$\tilde{H}(\cdot,\varepsilon_{t+1})$} (3.5,.5);
\draw[-latex,thick,red] (-1,2) to [out=60,in=120]  node[above,midway,black]{$\tilde{H}(\cdot,\varepsilon_{t+1})$} (2.5,2.2);
\node at (4,.5) {$\mathcal{X}_{R}$};
\draw[-latex](5,2.5)node[right]{\small $\partial \mathcal{X}_{R}$}
	  to[out=120,in=30] (3.5,2.5);
\end{tikzpicture}
\end{adjustbox}
\captionsetup{width=0.9\textwidth}
\caption{A graphical illustration for the relationships between $\mathcal{K}_{t,R}, \widehat{\mathcal{K}}_{t,R}$, 
and $\mathcal{X}_{R}$.}
\label{fig:simu_range_post_action}
\end{figure}
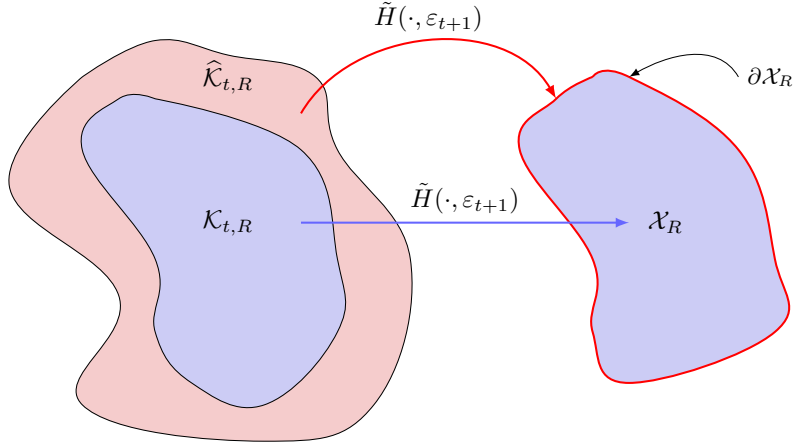

Since the value function $\tilde{V}_{t}(\cdot)$ is explicitly given by Eq. \eqref{aux_value_fun_boundary} over $\partial \mathcal{X}_{R}$,
the primary goal of our proposed LSMC algorithm is to get a numerical estimate for the value function
over the open set $\mathring{\mathcal{X}}_{R}$.
In view of this, we may circumscribe the support of the artificial probability distribution 
that the post-action values are simulated from.
First note that the range of post-action value of the auxiliary state process denoted by
$\widetilde{\mathcal{K}}_{t,R}$ is given by
\bas
\widetilde{\mathcal{K}}_{t,R}:= 
\bigcup_{x \in \mathring{\mathcal{X}}_{R}}
\left(\bigcup_{a \in A_{t}(x) }\big\{ K(x,a)\big\}\right),
\ \ \textup{for}\ \ t\in \mathcal{T}_{0}.
\eas
Consider the following subset:
\ba
\label{simu_range_post_action}
\widehat{\mathcal{K}}_{t,R}:=
\left\{k \in \widetilde{\mathcal{K}}_{t,R}\ \Big|
\ \tilde{H}(k,e_1)=\tilde{H}(k,e_2) \in  \partial \mathcal{X}_{R},
\ \forall e_{1}\ \textup{and}\ e_2 \in \range\left(\varepsilon_{t+1}\right)
\right\},
\ea
for $t\in \mathcal{T}_{0}$,
where $\range\left(\varepsilon_{t+1}\right)$ is the set of all values the random innovation $\varepsilon_{t+1}$ might take
and $\tilde{H}(\cdot,\cdot)$ is the transition equation relating the post-action value at time step $t$ to the state variable at time step $t+1$ which is given in Eq. \eqref{H}.

The preceding equation states that
$X_{t+1}$ will stop at a certain point in the boundary set $\partial \mathcal{X}_{R}$ 
if $X_{t^{+}}^{R}:=K\left(X_{t}^{R},a_{t}\right)$ lies in the set $\widehat{\mathcal{K}}_{t,R}$;
see Figure \ref{fig:simu_range_post_action} for a graphical illustration.
To make the matter more concrete, let us consider the example of pricing variable annuities
(see, e.g. \cite{Huang2016} and \cite{Shen2017})
where $X_{t^{+}}^{R}$ corresponds the post-withdrawal value of the investment account.
If the investment account is depleted after the policyholder's withdrawal 
(i.e., $X_{t^{+}}^{R}=0$),
it remains exhausted forever (i.e., $X_{n}^{R}=0$ for $n=t+1,\dots,T$).
In such an example, $\widehat{\mathcal{K}}_{t,R}$ is a singleton $\{0\}$.
In view of the above discussion and Eq. \eqref{aux_value_fun_boundary},
for any $k \in \widehat{\mathcal{K}}_{t,R}$,
we observe
\ba
\label{continuation_fun_after_absorbing}
\tilde{C}_{t}(k)= \e \left[\tilde{V}_{t+1}\big(\tilde{H}(k,\varepsilon_{t+1})\big)\right] 
= \tilde{V}_{t+1}\big(\tilde{H}(k,e)\big)
\ea
which has a value independent of $e \in \range(\varepsilon_{t+1})$
and is given by Eq. \eqref{aux_value_fun_boundary}.
Therefore, at time step $t$,
it suffices to get a regression estimate for the continuation function $\tilde{C}_{t}(\cdot)$ on the set 
$\mathcal{K}_{t,R}:=\widetilde{\mathcal{K}}_{t,R} \backslash \widehat{\mathcal{K}}_{t,R}$.

\subsubsection{The algorithm}
Now we present the Backward Simulation and Backward Updating (BSBU) algorithm as follows.
\begin{enumerate}
\item
\textbf{Initiation:}\quad 
Set $\tilde{V}_{T}^{\textup{E}}(x)=f_T(x)$ for $x \in \col\left(\mathcal{X}_{R}\right)$.
For $t=T-1,T-2,\dots,0$, do the two steps below.

\item
\textbf{Backward Simulation:}
\begin{enumerate}
\item[2.1] \textbf{Simulation of post-action value} \quad
Generate a sample of the post-action values denoted by 
\bas
\mathbf{X}_{t^{+}}^{M}:=\left\{X_{t^{+}}^{(m)},\ m=1,2,\dots,M\right\}
\eas
from a probability distribution ${\sf Q}_{t,R}$ with support $\mathcal{K}_{t,R}$.

\item[2.2] \textbf{Simulation of the state process} \quad
Construct the sample of the state process at time step $n+1$ according to
\ba
\label{sample_post_action-2}
\mathbf{X}_{t+1}^{M}:=\left\{X_{t+1}^{(m)}=\tilde{H}\left(X_{t^{+}}^{(m)},\varepsilon_{t+1}^{(m)}\right),\ m=1,2,\dots, M\right\}.
\ea
with $\left\{\varepsilon_{t+1}^{(m)}, m=1,2,\dots,M\right\}$ being a sample of the random innovations.
\end{enumerate}

\item \textbf{Backward Updating:}
\begin{enumerate}
\item[3.1] \textbf{Data preparation} \quad
Given a numerical estimate of value function at time step $t+1$, denoted by $\tilde{V}_{t+1}^{\textup{E}}(\cdot)$, construct the sample
\ba
\label{response_variable}
\mathbf{Y}_{t+1}^{M}:=\left\{\tilde{V}_{t+1}^{\textup{E}}\left(X_{t+1}^{(m)}\right),m=1,2,\dots, M\right\}.
\ea

\item[3.2] \textbf{Regression} \quad
Take $\mathbf{Y}_{t+1}^{M}$ and $\mathbf{X}_{t^{+}}^{M}$ as the samples of response variable and regresssor, respectively, and employ a certain non-parametric regression to obtain a regression estimate 
$\tilde{C}_{t}^{\textup{E}}(\cdot)$ over the set $\mathcal{K}_{t,R}$.
For $k\in \widehat{\mathcal{K}}_{t,R}$, we set $\tilde{C}_{t}^{\textup{E}}(k)=\tilde{C}_{t}(k)$ 
with $\tilde{C}_{t}(\cdot)$ given by Eq. \eqref{continuation_fun_after_absorbing}. 

\item[3.3] \textbf{Optimization} \quad
An estimate for the value function at time step $t$ is given by:
\ba
\label{Bellman_eq-4}
\tilde{V}_{t}^{\textup{E}}(x) =
\sup \limits_{a \in A_{t}(x)} 
\Big[f_{t}(x, a)+ \varphi \tilde{C}_{t}^{\textup{E}}\big(K(x,a)\big)\Big], 
\ \ \textup{for}\ \ x\in \mathring{\mathcal{X}}_{R}.
\ea
For $x \in \partial \mathcal{X}_{R}$, we set $\tilde{V}_{t}^{\textup{E}}\left(x\right)=\tilde{V}_{t}(x)$ 
with $\tilde{V}_{t}(\cdot)$ given by Eq. \eqref{aux_value_fun_boundary}.
\end{enumerate}

\end{enumerate}
In Step 3.2, we prescribe $\tilde{C}_{t}(k)$ for the value of $\tilde{C}_{t}^{\textup{E}}(k)$ when $k \in \widehat{\mathcal{K}}_{t,R}$
because $K\left(X_{t}^{(m)},a\right)$ might fall in the set $\widehat{\mathcal{K}}_{t,R}$.
Similarly, in Step 3.3, Eq. \eqref{aux_value_fun_boundary} is invoked to evaluate $\tilde{V}_{t}^{\textup{E}}\left(x\right)$
for $x \in \partial \mathcal{X}_{R}$ as $X_{t}^{(m)}$ generated by Eq. \eqref{sample_post_action-2} 
may lie on $\partial \mathcal{X}_{R}$, the boundary set of the truncated domain.
The backward information propagation in the above BSBU algorithm 
is illustrated in Figure \ref{fig:propagation-1}.

\begin{figure}
\centering
\begin{tikzpicture}[>=stealth,thick]
\draw (-2,0) node (e) {$\cdots$};
\draw (0,0) node (a) {$\tilde{V}_{t+1}^{\textup{E}}(\cdot)$};
\draw[->, thick, red] (e) to (a);
\draw (4,0) node (b) {$\tilde{C}_{t}^{\textup{E}}(\cdot)$};
\draw[->, thick, red] (a) to node [above, align=center, text width=2cm,black] {Regression} (b);
\draw (8,0) node (c) {$\tilde{V}_{t}^{\textup{E}}(\cdot)$};
\draw[->, thick, red] (b) to node [above, align=center, text width=2cm,black] {Eq. \eqref{Bellman_eq-4}} (c);
\draw (10,0) node (d) {$\dots$};
\draw[->, thick, red] (c) to (d);
\end{tikzpicture}
\captionsetup{width=0.9\textwidth}
\caption{A diagram for backward information propagation in the BSBU algorithm.}
\label{fig:propagation-1}
\end{figure}
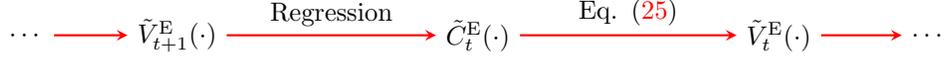

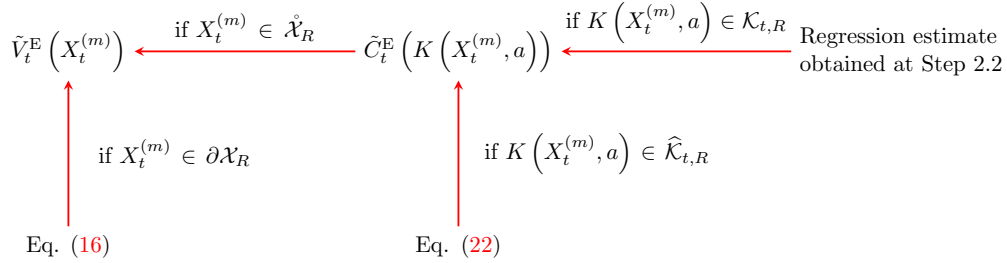
\begin{figure}
\centering
\begin{adjustbox}{max totalsize={0.9\textwidth}{.7\textheight},center}
\begin{tikzpicture}[>=stealth,thick]
\draw (0,0) node (a) {$\tilde{V}_{t}^{\textup{E}}\left(X_{t}^{(m)}\right)$};
\draw (6,0) node (b) {$\tilde{C}_{t}^{\textup{E}}\left(K\left(X_{t}^{(m)},a\right)\right)$};
\draw[<-, thick, red] (a) to node [above, align=center, text width=3cm,black] {if $X_{t}^{(m)} \in \mathring{\mathcal{X}}_{R}$} (b);
\draw (0,-3) node (c) {Eq. \eqref{aux_value_fun_boundary}};
\draw[<-, thick, red] (a) to node [right, align=center, text width=3cm,black] {if $X_{t}^{(m)} \in \partial \mathcal{X}_{R}$} (c);
\draw (13,0) node (d) [text width= 3.5cm]{Regression estimate obtained at Step 2.2};
\draw[->, thick, red] (d) to node [above, align=center, text width=3.5cm,black] {if $K\left(X_{t}^{(m)},a\right) \in \mathcal{K}_{t,R}$} (b);
\draw (6,-3) node [align=center, text width=3cm] (e) {Eq. \eqref{continuation_fun_after_absorbing}};
\draw[->, thick, red] (e) to node [right, align=center, text width=4cm,black] 
{if $K\left(X_{t}^{(m)},a\right) \in \widehat{\mathcal{K}}_{t,R}$} (b);
\end{tikzpicture}
\end{adjustbox}
\captionsetup{width=0.9\textwidth}
\caption{A diagram for the information propagation in evaluating {\small $ \tilde{V}_{t}^{\textup{E}}\left(X_{t}^{(m)}\right)$}.}
\label{fig:information_propagation-2}
\end{figure}

\subsubsection{Discussions}
Comparing the above BSBU algorithm and the FSBU counterpart in Section \ref{sec:LSMC_review}, 
we have the following observations.
\begin{enumerate}
\item 
Firstly, the primary difference of the two algorithms lies in how to generate the post-action values of state process, i.e., $\mathbf{X}_{t^{+}}^{M}$.
The control randomization method is a forward simulation scheme
while the BSBU algorithm directly generates post-action value from a certain prior distribution.
Indeed, the FSBU algorithm can be viewed as a special BSBU algorithm
if ${\sf Q}_{t,R}$ is chosen as the probability distribution of the post-action value from a control randomization procedure.
In general, both methods do not yield the distribution of $X_{t^{+}}$ driven by the optimal action,
and thus, there is no loss to directly generate $\mathbf{X}_{t^{+}}^{M}$ from a prior distribution ${\sf Q}_{t,R}$.

\item
Secondly, the BSBU method has the advantage of reducing memory and time costs.
On the one hand, one does not need to store the sample of whole trajectories at each time step in the BSBU algorithm.
On the other hand, 
the total time cost of simulating the state process is of $O(T)$ in the BSBU algorithm,
while it is of $O\left(T^2\right)$ in the FSBU counterpart; 
see the item ``Cost of forward simulation" in Section \ref{sec:LSMC_review}.

\item
Thirdly, the BSBU algorithm circumvents extrapolating the numerical estimates of the continuation function and the value function.
It is notable that $\tilde{C}_{t}^{\textup{E}}(\cdot)$ and $\tilde{V}_{t}^{\textup{E}}(\cdot)$ are obtained over the sets
$\widetilde{\mathcal{K}}_{t,R}$ and $\textup{col}\left(\mathcal{X}_{R}\right)$, respectively, at time step $t$;
see Steps 2.2-2.3 of the above BSBU algorithm.
At the time step $t-1$, the BSBU algorithm does not require the knowledge of the value function (resp. continuation function) 
outside $\textup{col}\left(\mathcal{X}_{R}\right)$ (resp. $\widetilde{\mathcal{K}}_{t,R}$)
in computing the regression data $\mathbf{Y}_{t}^{M}$; see Figure \ref{fig:information_propagation-2} for a graphical illustration. 
This nice property inherits from the construction of the auxiliary state process $X^{R}$
whose values are confined to a bounded set.
In the FSBU algorithm, however,
the state process is not restrained and $K\left(X_{t}^{(m)},a\right)$ might lie outside 
the regression domain for $\tilde{C}_{t}^{\textup{E}}(\cdot)$.
In such a situation, extrapolating the numerical solution causes extra error which is hard to control.
\end{enumerate}

\subsection{Sieve Estimation Method}
\label{sec:sieve_estimation}
In this subsection, we discuss the details of the regression method
used to estimate the continuation function in our BSBU algorithm.
\subsubsection{Selection criteria for regression method}
In Section \ref{sec:LSMC_review} we have discussed the potential issue which may be associated with a regression method
when used for the estimation of the continuation function of a stochastic control problem;
see the item ``Choice of regression method".
Based on the discussion, we propose the following criteria for the choice of regression method 
in estimating continuation function.
\begin{description}
\vspace{-1ex}\item[\bf (C1)]
{\bf Small memory cost}\quad 
The regression problem embedded in an LSMC algorithm usually exhibits extraordinarily large sample size.
Thus, an appropriate regression method should have small memory requirement.
This criterion excludes the kernel method (\cite{Nadaraya1964} and \cite{Watson1964}), 
local-polynomial regression method (\cite{Fan1996}),
and Isotonic regression method (\cite{Robertson1988}) 
which require storing all sample points in the memory
in order to compute the regression function at any point in the domain.

\vspace{-1ex}\item[\bf (C2)]
{\bf Computationally cheap}\quad 
In almost all nonparametric regression methods, a certain parameter 
(referred to as \textit{tuning} parameter in statistical literature)
is used to avoid undesirable overfitting or underfitting of the regression model.
Determining the optimal value of such a tuning parameter is usually computationally intensive.
Therefore, an ideal regression method should be insensitive to the tuning parameter.
\end{description}
In view of the above two criteria,
we have a limited number of suitable choices 
despite the voluminous nonparametric regression methods in the literature.
In the following, we discuss a class of regression methods referred to as the \textit{sieve estimation method}
which include the least-squares method of \cite{Longstaff2001} as a special case.

\subsubsection{Shape-preserving sieve estimation}
We give a brief introduction to the sieve estimation method; 
refer to \cite{Chen2007} for a comprehensive review.
Suppose we have a sample of independent and identically distributed (i.i.d.) random pairs $\left\{\left(U^{(m)},Z^{(m)}\right)\right\}_{m=1}^{M}$
where $Z^{(m)}$ is a $\R^{r}$-valued random vector with compact support $\mathcal{Z}$ 
and $U^{(m)}$ is a univariate random variable.
Define the function $g(\cdot): \mathcal{Z} \longrightarrow \R$ as 
\ba
\label{regression_fun}
g\left(z\right) = \e \left[\left. U^{(m)} \right|Z^{(m)}=z\right]
\ea
which is independent of $m$.
In the context of our BSBU algorithm,
$U^{(m)}$ and $Z^{(m)}$ correspond to $\tilde{V}_{t+1}\left(X_{t+1}^{(m)}\right)$
and  $X_{t^{+}}^{(m)}$, respectively,
and the parallel function $g(\cdot)$ is the continuation function $\tilde{C}_{t}(\cdot)$.

The sieve estimation method strives to estimate the functional form of $g(\cdot)$
by solving the following optimization problem:
\ba
\label{optimization_sieve}
\hat{g}(\cdot):= \arg \min_{h(\cdot) \in \mathcal{H}_{J}} 
\frac{1}{M} \sum_{m=1}^{M}\left[U^{(m)}- h \left(Z^{(m)}\right)\right]^2,
\ea
where $\mathcal{H}_{J}$ is a finite-dimensional functional space depending on a certain parameter $J$
and is called as \textit{sieve space}.
Intuitively, the ampler the sieve space is,
the smaller the ``gap" between the $\mathcal{H}_{J}$ and the function $g(\cdot)$ would be.
The price to pay is that larger estimation error is incurred for a richer sieve space due to limited sample size $M$.
Therefore, one has to balance such a trade-off by controlling the complexity of the sieve space
and this is achieved by tuning the parameter $J$. 
To make the matter more concrete, we consider two examples of the sieve space in the sequel.
\begin{example}[Linear Sieve Space]
Let $\left\{\phi_{j}(\cdot): \mathcal{Z} \longrightarrow \R\right\}_{j\in \N}$ 
be a sequence of basis functions indexed by $j \in \N$.
Consider the sieve space defined by
\ba
\label{linear_sieve_space}
\mathcal{H}_{J}= \left\{h(\cdot):\quad  
h(z)=\sum_{j=0}^{J} \beta_{j} \phi_{j}(z), \ \ \beta_{j} \in \R \right\}.
\ea
\end{example}
The above set $\mathcal{H}_{J}$ is essentially a linear span of finitely many basis functions
and is referred to as \textit{linear sieve space} in the statistical literature.

In the present context of stochastic control,  
the regression function $g(\cdot)$ corresponds to the continuation function
and it exhibits some shape properties such as monotonicity in many applications;
see \cite{Del2012} for pricing American option
and \cite{Huang2016} for valuing equity-linked insurance product, among others.
In view of this, it is natural to expect the element in the sieve space satisfies such shape constraints,
which in turn preserves the financial interpretations of the numerical result.
This can be achieved by considering a special linear sieve space in the following example.

\begin{example}[Shape-Preserving Sieve Space]
Let $\left\{\phi_{j}(\cdot): \mathcal{Z} \longrightarrow \R\right\}_{j\in \N}$ 
be a sequence of basis functions indexed by $j \in \N$.
Denote $\bm{\beta}_{J}=\left(\beta_{0},\dots,\beta_{J}\right)^{\T}$ with 
$\beta_{j} \in \R,\ j=0,1,\dots,J$.
Consider the sieve space defined by
\ba
\label{shape_sieve_space}
\mathcal{H}_{J}= \left\{h(\cdot):\quad  
h(z)=\sum_{j=0}^{J} \beta_{j} \phi_{j}(z), 
\ \ \mathbf{A}_{J} \bm{\beta}_J \geq \bm{0}_{b(J)}\right\},
\ea
where $b(\cdot): \N \longrightarrow \N$ is some integer-valued function,
$\mathbf{A}_{J}$ is a $b(J)$-by-$(J+1)$ matrix, and $\bm{0}_{b(J)}$ is a $b(J)$-by-1 null vector.
\end{example}
\cite{Wang2012-1,Wang2012-2} show that 
each element in the sieve space in Eq. \eqref{shape_sieve_space} 
is a convex, concave, or monotone function (with respect to each coordinate)
with a special the choice of matrix $\mathbf{A}_{J}$
given that $\phi_{j}(\cdot),\ j=0,1,\dots,J,$ are Bernstein polynomials.
Some forms of $\mathbf{A}_{J}$ are relegated to Appendix \ref{app:constraint_matrix}.

For a linear sieve space $\mathcal{H}_{J}$ defined either in Eq. \eqref{linear_sieve_space} or Eq. \eqref{shape_sieve_space},
the solution of the preceding optimization problem \eqref{optimization_sieve}
is given by the following form:
\ba
\label{regression_estimate}
\hat{g}(z)
= \hat{\bm{\beta}}^{\T} \bm{\phi}(z),
\ \ \textup{for}\ \ z \in \mathcal{Z},
\ea
where $\bm{\phi}(z):=\left(\phi_{1}(z),\dots,\phi_{J}(z)\right)^{\T}$ 
and $\hat{\bm{\beta}}$ is the optimizer of the following optimization problem: 
\ba
\label{sieve_estimator}
\min\limits_{\bm{\beta} \in \R^{J}}
\frac{1}{M} \sum_{m=1}^{M}\left[U^{(m)}- \bm{\beta}^{\T} \bm{\phi} \left(Z^{(m)}\right)\right]^2,
\ \ \textup{subject to}\ \ \bm{\beta}^{\T} \bm{\phi}(\cdot) \in \mathcal{H}_{J}.
\ea
The dependency of $\hat{\bm{\beta}}$ and $\bm{\phi}(\cdot)$ on $J$ is suppressed for brevity.
In general, one has to solve a constrained quadratic programming problem to obtain $\hat{\bm{\beta}}$.

\subsubsection{Discussions}
One clear merit of the above linear sieve estimation is that one only needs to store the vector 
$\hat{\bm{\beta}}$ for future evaluation of the regression function $\hat{g}(\cdot)$ at any point in the domain
because basis functions $\bm{\phi}(\cdot)$ are explicitly known at the first hand.
This makes the linear sieve estimation method tailored to our present problem in terms of the criterion {\bf (C1)}.

For the criterion {\bf (C2)}, it is well documented in statistical literature that
when the true regression function $g(\cdot)$ satisfies certain shape constraints,
the shape-preserving estimate $\hat{g}(\cdot)$ obtained by \eqref{sieve_estimator} 
with $\mathcal{H}_{J}$ given by Eq. \eqref{shape_sieve_space}
is \textit{insensitive} to the tuning parameter $J$; 
see, e.g., \cite{Meyer2008} and \cite{Wang2012-1,Wang2012-2}. 
However, this is legitimate only when the true conditional mean function $g(\cdot)$ 
exhibits such convexity, concavity or monotonicity property.
For the general case when there is no prior shape information of $g(\cdot)$, 
one has to use the sieve space \eqref{linear_sieve_space}
and the regression estimate is sensitive to the choice of $J$.
Under such a situation, $J$ should be determined in a data-driven manner.
In Appendix \ref{app:CV},
we present some common methods of selecting $J$ discussed in the statistical literature.

Finally, the convergence of the sieve estimate $\hat{g}(\cdot)$ to the conditional mean function $g(\cdot)$
is ensured under some technical conditions.
These conditions are summarized in Assumption \ref{assum:sieve_estimation}
which is relegated to Appendix \ref{app:assum_sieve} 
for the clarity of presentation.

\subsection{Convergence Analysis of BSBU Algorithm}
\label{sec:conv_BSBU}
Now we are ready to conduct convergence analysis of the BSBU algorithm proposed in Section \ref{sec:algorithm}.
For the regression method employed in the algorithm,
we restrict our attention to the linear sieve estimator given by Eqs. \eqref{regression_estimate}
and \eqref{sieve_estimator} in the last subsection.

A complete convergence analysis of the BSBU algorithm should take account of three types of errors:
\begin{description}
\item[\bf (E1)] \textbf{Truncation Error}\quad 
The truncation error is caused by taking $\tilde{V}_{0}(X_{0})$ as a proxy for $V_{0}(X_{0})$.

\item[\bf (E2)] \textbf{Sieve Estimation Error}\quad
At each step of the BSBU algorithm,
the sieve estimation method is employed to get an estimate for the continuation function.
The associated sieve estimation error stems from two resources: 
(a) the bias caused by using a finite-dimensional sieve space $\mathcal{H}_{J}$ to approximate continuation function;
and (b) the statistical error in estimating coefficients of basis functions
under a limited sample size of $M$. 

\item[\bf (E3)] \textbf{Accumulation Error}\quad
The primal goal of the regression step in the BSBU algorithm is to estimate the continuation function 
of the auxiliary stochastic control problem, i.e., $\tilde{C}_{t}(\cdot)=\e \left[\tilde{V}_{t+1}(X_{t+1})\big| X_{t^{+}}=\cdot\ \right]$.
Thus, in principle, one should generate a random sample
\bas
\left\{
\left(\tilde{V}_{t+1}\left(X_{t+1}^{(m)}\right), X_{t^{+}}^{(m)}\right)
\right\}_{m=1}^{M}
\eas
based on which the sieve estimation method can be employed to get a regression estimate.
However, $\tilde{V}_{t+1}(\cdot)$ is not known exact at each time step of the algorithm
and is replaced by its numerical estimate $\tilde{V}_{t+1}^{\textup{E}}(\cdot)$; see Step 2.2 of the BSBU algorithm in Section \ref{sec:algorithm}.
Therefore, the the algorithm error accumulated from time step $T-1$ down to $t+1$
triggers a new type of error in addition to {\bf (E1)} and {\bf (E2)}.
\end{description}

{\bf (E1)} has been investigated in Theorem \ref{thm:truncation_error_bound}.
The discrepancy between $\tilde{V}_{0}(X_0)$ and $\tilde{V}_{0}^{\textup{E}}(X_{0})$
is contributed by {\bf (E2)} and {\bf (E3)}.
Distinguishing these two types of error plays a crucial role in our convergence analysis and 
this is inspired by \cite{Belomestny2010}.
Our main convergence result is summarized in the following theorem.

\begin{theorem}[BSBU Algorithm Error]
\label{thm:LSMC_error}
Suppose that 
\begin{description}
\item[(i)]
Assumptions \ref{assum:moment_condition}--\ref{assum:growth_condition} 
and Assumption \ref{assum:eigen_psi} in Appendix \ref{app:proof_thm2} hold;

\vspace{-1ex}\item[(ii)]
Assumption \ref{assum:sieve_estimation} in Appendix \ref{app:sieve} holds for
$U^{(m)}=V_{t+1}\left(X_{t+1}^{(m)}\right)$ 
and $Z^{(m)}=X_{t^{+}}^{(m)}$ uniformly in $t\in \mathcal{T}_{0}$,
where $X_{t^{+}}^{(m)}$ and $X_{t+1}^{(m)}$ are given in Steps 2.1 and 2.2 of the BSBU algorithm.
\end{description}
Then, there exists a constant $\psi$ such that
\ba
\label{LSMC_error}
    \left|\tilde{V}_{0}(X_{0})-\tilde{V}_{0}^{\textup{E}}(X_{0})\right|
=   O_{\p}\left(\sqrt{\psi^{T-1}\left(J/M + \rho_{J}^{2}\right)} \right),
\ \ \textup{as}\ \ M \longrightarrow \infty,
\ea
with ``Big O p" notation $O_{\p}(\cdot)$ defined in Definition \ref{def:big_O_P} of Appendix \ref{app:proof_thm2}.
\end{theorem}
The above theorem basically states that the numerical solution $\tilde{V}_{0}^{\textup{E}}(X_0)$
converges to $\tilde{V}_{0}(X_0)$ in probability as both the number of basis functions $J$ 
and number of simulated paths $M$ approach infinity at the rate specified by Condition (v) in Assumption \ref{assum:sieve_estimation}.
Since Theorem \ref{thm:truncation_error_bound} shows that 
the discrepancy between $\tilde{V}_{0}(X_0)$ and $V_{0}(X_{0})$ shrinks as $R$ increases,
the numerical estimate $\tilde{V}_{0}^{\textup{E}}(X_0)$ is a legitimate approximation for $V_{0}(X_0)$
when $R$, $J$, and $M$ are considerable.
The R.H.S. of Eq. \eqref{LSMC_error} reveals that the overall BSBU algorithm error arises from 
the two resources discussed in the previous item {\bf (E2)},
which are indicated by the terms $\rho_{J}$ and $J/M$, respectively.
Furthermore, Eq. \eqref{LSMC_error} also shows that 
such a regression error is escalated by a factor $\psi$ at each time step, 
which reflects the error accumulation from time step $T-1$ down to time step $0$
and is in line with the earlier discussion in the item {\bf (E3)}.


\section{Application: Pricing Equity-linked Insurance Products}
\label{sec:VA}
In this section, we apply the BSBU algorithm to the pricing of equity-linked insurance products.
This pricing problem is an appropriate example to show the limitations of the FSBU algorithm commented in Section \ref{sec:LSMC_review}.
For the convenience of illustration, 
the contract we study here is a simplified version of variable annuities (VAs);
for the discussions on more generic policies, we refer to
\cite{Azimzadeh2015}, \cite{Huang2016}, \cite{Huang2017}, and \cite{Shen2017}, among others.

\subsection{Contract Description}
We give a brief introduction to the VA.
VAs are equity-linked insurance products issued by insurance companies.
At the inception of the contract, 
the policyholder (PH) pays a lump sum $W_0$ to the insurer which is invested into a certain risky asset.
The PH is entitled to withdraw any portion of the investment before the maturity.
She also enjoys certain guaranteed payments regardless of the performance of the investment account.
Therefore, the insurer provides downside protection for a potential market decline.
As a compensation, the insurer deducts insurance fees from the investment account
and trade available securities to hedge his risk exposure.
Thus, no-arbitrage pricing has been the dominating paradigm for pricing VAs in the literature.
The primary challenge of this pricing problem stems from the uncertainty of the PH's withdrawal behavior.
This is conventionally resolved by studying the optimal
withdrawal strategy of the PH, 
which naturally leads to a stochastic control problem;
see \cite{Dai2008}, \cite{Chen2008}, \cite{Huang2016}, and many others.

\subsection{Model Setup}
\label{sec:model_setup_VA}
In the following, we exemplify the model setup of Section \ref{sec:framework} in the present pricing problem.
The lattice $\mathcal{T}$ corresponds to the collection of all available withdrawal dates.
The first decision variable $\tau_t$ represents
the PH's decision to initialize the withdrawal or not
by taking values $1$ and $0$, respectively.
As we will see later, the payoff functions depend on the timing of the first withdrawal of the PH.
Therefore, a state variable $\{I_t\}_{t\in \mathcal{T}}$ is introduced to record the first-withdrawal-time, 
and its evolution mechanism is prescribed as follows: $I_{0}=0$, and
\ba
\label{trans_I}
I_{t+1}=S_{t}^{I}(I_{t},\tau_{t}):=
\begin{cases}
t, \ \, &\mbox{if}\ \, I_{t}=0\ \,\mbox{and}\ \, \tau_{t}=1,\\
I_{t}, \ \, &\mbox{otherwise},
\end{cases}
\ea
for $t\in \mathcal{T}_{0}$.
The feasible set of $\tau_t$ is a singleton $\{1\}$ if the withdrawal has been initialized, i.e., $I_t>0$;
otherwise, it is $\{0,1\}$.

Denote $(a)_{+}:=\max\{a,0\}$ and $a \vee b := \max\{a,b\}$.
The second state variable corresponds to the investment account and it evolves according to
\ba
\label{trans_W}
\begin{cases}
W_{0}&=W_{0},\\
W_{t+1}&=\underbrace{\big(W_{t}-\gamma_{t}\big)_{+}}_{\textup{post-withdrawal value}} \cdot
\varepsilon_{t+1},\ \ \gamma_{t}\in \big[0,W_{t} \vee G(I_{t})P_{0}\big],
\ \ t\in \mathcal{T}_{0},
\end{cases}
\ea
where $\gamma_{t}$ is the withdrawal amount of the PH at time $t$, 
$\varepsilon_{t+1}$ is the absolute return of the underlying asset over $[t,t+1]$,
and $G(I_t)$ is a certain percentage depending on the first-withdrawal-time $I_t$.
The above equation implies that
the PH can withdraw up to the amount of $G(I_{t})P_{0}$ 
even if the investment account is depleted, i.e., $W_{t}=0$.
The jump mechanism of the investment account across each withdrawal date is illustrated in Figure \ref{fig:investment_account}.

\newcommand{\nEmmett}[7]{
\pgfmathsetseed{#7}%
\draw[#4] (4.4,0)
\foreach \x in {1,...,#1}
{   -- ++(#2,0.02+rand*#3)
}
\foreach \x in {1,...,#6}
{   -- ++(#2,-0.01+rand*#3)
}
node[right] {#5};
}

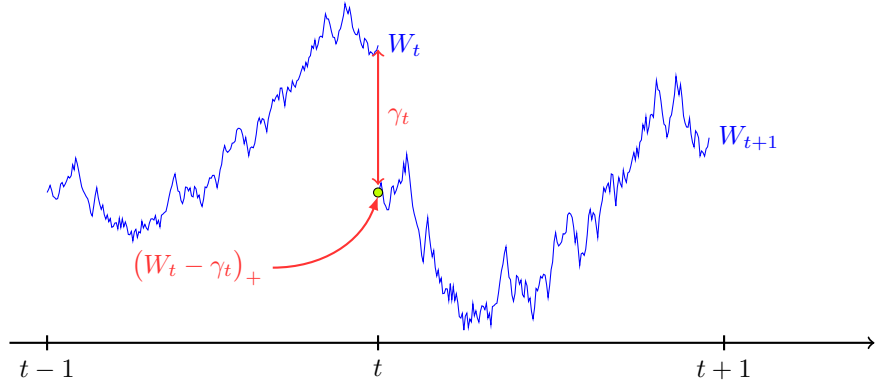
\begin{figure}
\centering
\begin{tikzpicture}
\draw[->,thick] (-0.5,-2) -- (11,-2);
\Emmett{200}{0.02}{0.15}{blue}{$W_{t}$}{20}{1};
\nEmmett{200}{0.02}{0.25}{blue}{$W_{t+1}$}{20}{1};
\draw[red!80,<->,thick] (4.4,0.1) to node [right] {$\gamma_{t}$} (4.4,1.9);
\draw[thick] (0,-1.9) -- (0,-2.1) node [below] {$t-1$};
\draw[thick] (4.4,-1.9) -- (4.4,-2.1) node [below] {$t$};
\draw[thick] (9,-1.9) -- (9,-2.1) node [below] {$t+1$};
\draw [fill=lime](4.4,0) circle (.06);
\draw[-latex,thick,red!80](3,-1)node[left]{$\big(W_{t}-\gamma_{t}\big)_{+}$}
	  to[out=0,in=250] (4.4,-0.05);
\end{tikzpicture}
\captionsetup{width=0.9\textwidth}
\caption{Jump mechanism of the investment account across a withdrawal date.}
\label{fig:investment_account}
\end{figure}

Now, the state process and the DM's action are
$X=\left\{X_{t}=(W_{t},I_{t})^{\T}\right\}_{t\in\mathcal{T}}$ and 
${\sf a}=\left\{a_{t}=(\gamma_t,\tau_t)^{\T}\right\}_{t \in \mathcal{T}}$, respectively,
with the superscript ``$\T$" denoting vector transpose.
In accordance with Eqs. \eqref{trans_I} and \eqref{trans_W},
the accompanying transition equation is
$X_{t+1}=H\big(K(X_{t},a_{t}),\varepsilon_{t+1}\big)$, where 
\ba
\label{post-action_value}
K(X_{t},a_{t})=\left(\big(W_{t}-\gamma_{t}\big)_{+},\
S_{t}^{I}(I_{t},\tau_{t})
\right)^{\T},\ \ 
H\big(k,\varepsilon_{t+1}\big)=\big(k_1\varepsilon_{t+1}, k_2\big)^{\T}，
\ea
with $k=(k_1,k_2)^{\T}\in [0,\infty)\times \mathcal{T}_{0}$.
The dependency of $K(\cdot,\cdot)$ on $t$ is suppressed for notational brevity.

Next, we discuss the feasible set of the PH's action. 
In principle, the withdrawal amount $\gamma_{t}$ takes values in a continum 
$\big[0,W_{t} \vee G(I_{t})P_{0}\big]$ (see Eq. \eqref{trans_W}).
However, it can be shown that the optimal withdrawal amount is limited to three choices:
1) $\gamma_{t}=0$, 
2) $\gamma_{t}=G(I_{t})P_{0}$,
and 3) $\gamma_{t}=W_{t}$
under certain contract specifications; 
see \cite{Azimzadeh2015}, \cite{Huang2016}, \cite{Huang2017}, and \cite{Shen2017}.
Via a similar argument adopted by the above references,
one may show that this conclusion still holds for the contract considered here.
Therefore, we restrict the feasible set of action $a_{t}$ into the following discrete set:
\ba
\label{feasible_set}
A_{t}(X_{t})=
\begin{cases}
\big\{(0,0)^{\T}, \left(G(I_{t})P_{0},1\right)^{\T},  \left(W_{t},1\right)^{\T}\big\},
&\textup{if}\ I=0,
\ \ \textup{(withdrawal has not been initialized)}\\
\big\{(0,1)^{\T}, \left(G(I_{t})P_{0},1\right)^{\T},  \left(W_{t},1\right)^{\T}\big\},
&\textup{if}\ I>0,
\ \ \ \ \textup{(withdrawal has been initialized)}\\
\end{cases}
\ea
for $t=1,2,\dots,T-1$. As a convention, the PH is not allowed to withdraw at inception,
and thus $A_{0}(X_0)=\emptyset$.

We proceed by specifying the reward functions which corresponds to the policy payoffs in the present context.
Before maturity, the cash inflow of the PH is her withdrawal amount subject to some penalty:
\bas
f_{t}(X_{t},a_{t})=\gamma_{t}-\kappa\big(\gamma_{t}-G(I_{t})P_{0}\big)_{+},
\ \ \gamma_{t} \in \big[0,W_{t} \vee G(I_{t})P_{0}\big],\ \ 
t=1,2,\dots,T-1,
\eas
with $\kappa \in [0,1]$ being the penalty rate. 
In other words, the withdrawal amount in excess of
the guaranteed amount is subject to a proportional penalty.
Conventionally, $f_{0}(\cdot,\cdot)\equiv 0$.
At maturity, the policy payoff is the remaining value of the investment account, 
i.e., $f_{T}(X_{T})=W_{T}$.

Finally, we give the interpretation of value function in the present context.
$V_{t}(x)=V_{t}\left((W,I)^{\T}\right)$ with $I>0$ (resp., $I=0$) 
corresponds to the no-arbitrage price of the contract at withdrawal date $t$
given that the investment account has a value of $W$ and
the first withdrawal is triggered at withdrawal date $I$
(resp., no withdrawal has been taken).

\subsection{A BSBU Algorithm for the Pricing Problem}
The state process $X$ generally takes value in the unbounded set $\mathcal{X}=[0,\infty)\times \mathcal{T}_{0}$.
We consider a truncated domain:
$\mathcal{X}_{R}=[0,R)\times \mathcal{T}_{0}$ with $R>0$.
Consequently, we may define the auxiliary state process $X^{R}$ as in Eq. \eqref{transition_aux_state_process}.
The range of the post-action value is given by 
$\widetilde{\mathcal{K}}_{t,R}=\widehat{\mathcal{K}}_{t,R}\cup \mathcal{K}_{t,R}$,
where $\widehat{\mathcal{K}}_{t,R}=\{0,R\}\times \{0,1,\dots, t\}$ 
and $\mathcal{K}_{t,R}=(0,R)\times \{0,1,\dots, t\}$, respectively.
This is in line with Eq. \eqref{simu_range_post_action}.

Now we are almost ready to employ the BSBU algorithm developed in Section \ref{sec:algorithm}
to solve the present pricing problem.
It is worth noting that a discrete state variable $I_t$ appears in the present context
and the continuation function, in general, is not continuous with respect to the post-action value accompanying this state variable,
i.e., $k_2$; see Eq. \eqref{post-action_value}.
Consequently, Condition (iii) of Assumption \ref{assum:sieve_estimation} might not hold here; see Appendix \ref{app:assum_sieve}.
However, for each given value of $k_2$, 
the continuation function is still continuous with respect to $k_1$,
the post-action value associated with the investment account value.
And therefore, one may repeat Step 3.2 of the BSBU algorithm for every distinct value of $k_2$.
It is easy to see the convergence of the resulting BSBU algorithm is not influenced by this modification.

Finally, it remains to specify how to simulate the post-action value of the state process
in order to pave the way to implementing the BSBU algorithm.
In the sequel section, we will address this issue in details and, 
in particular, we will compare the control randomization method with our artificial simulation method.

\section{Numerical Experiments}
\label{sec:num_experiment}
This section devotes to conducting numerical experiments to show the merits of the BSBU algorithm
in the context of pricing the variable annuity product addressed in the last section.
\subsection{Parameter Setting}
We first present the parameter setting for our numerical experiments.
We consider $T=12$ time steps and
the time interval between two consecutive withdrawal date is assumed to be $\delta =1/12$.
This corresponds to a contract with one-year maturity and monthly withdrawal frequency.
The PH's initial investment is assumed to be one unit, i.e., $W_{0}=1$.
The guaranteed payment percentage is prescribed as follows:
\bas
G(I_{t})=
\begin{cases}
0.03,&\text{if}\ \ 0\leq I_{t} \leq 3,\\
0.05,&\text{if}\ \ 4\leq I_{t} \leq 7,\\
0.07,&\text{if}\ \ 8\leq I_{t} \leq 11.
\end{cases}
\eas
In other words, the PH enjoys a larger amount of guaranteed payment
if she postpones the initiation of the withdrawal.
As a result, the value function/continuation is not continuous with respect to the state variable $I_{t}$.

Let $r$ and $q$ be the annualized risk-free rate and insurance fee rate, respectively.
We assume the absolutely return $\varepsilon_{t+1}$ of the underlying fund follows
a log-normal distribution with $\e[\log\varepsilon_{t+1}]=\big(r-q-\sigma^2/2\big)\delta$
and $\var[\log\varepsilon_{t+1}]=\sigma^2\delta$ under a risk-neutral pricing measure.
This implicitly assumes the underlying fund evolves according to a Geometric Brownian Motion 
with annualized volatility rate $\sigma$.
Finally, the discounting rate is given by $\varphi = e^{-r \delta}$.
All of the above market and contract parameters are summarized in Table \ref{tab:para}.

\begin{table}[ht]
\begin{center}
\caption{Parameters used for numerical experiments.}
\label{tab:para}
\vspace{.5ex}
\renewcommand{\arraystretch}{1.2}
\begin{tabular*}{0.9\textwidth}{l @{\extracolsep{\fill}} l}
\hline\\[-8pt]
Parameter &   Value\\[5pt]
\hline\\[-8pt]
Volatility rate $\sigma$  &   $0.15$ \\
Risk-free rate $r$     &   $0.03$ \\
Insurance fee rate $q$     &   $0.01$ \\
Number of time steps $T$    &     12\\    
Length of time interval $\delta$    &    1/12\\  
Discouting factor $\varphi=e^{-r\delta}$    & 0.9975    \\ 
Initial purchase payment $W_0$ & 1\\ 
Withdrawal penalty $\kappa$ & 0.8 \\
Guaranteed payment percentage $G(I)$ & $0\leq I \leq 3: 3\%,\ 4\leq I \leq 7: 5\%$  \\
&$8\leq I \leq 11: 7\%$\\[5pt]
\hline
\end{tabular*}
\end{center}
\end{table}

Finally, we discuss the choice of truncation parameter $R$.
Under the present context,
it is easy to see that the function $\mathcal{E}(X_0,R)$ in Assumption \ref{assum:tail_prob} 
is bounded from above by the tail probability of the continuous running maximum of a geometric Brownian Motion.
To be specific, we have
\bas
\mathcal{E}(X_0,R) 
&\leq& \p \left(W_0 \max_{t\in[0,\delta T]} \left(e^{(r-q-0.5\sigma^2)t + \sigma \mathcal{B}_{t}}\right)\geq R\right),\\
&=& \p \left(\max_{t\in[0,\delta T]} \left((r-q-0.5\sigma^2)t/\sigma + \mathcal{B}_{t}\right)\geq (1/\sigma)\log \big(R/W_0\big)\right)\\
&=& 1-\mathcal{N}\left(\frac{(1/\sigma)\log \big(R/W_0\big)-\alpha \delta T}{\sqrt{\delta T}}\right)
+ \left(\frac{R}{W_0}\right)^{2(\alpha/\sigma)}
\mathcal{N}\left(\frac{-(1/\sigma)\log \big(R/W_0\big)-\alpha \delta T}{\sqrt{\delta T}}\right)
\eas
with $\alpha:=(r-q-0.5\sigma^2)/\sigma$,
where $\mathcal{B}_{t}$ is a standard Brownian Motion,
$\mathcal{N}(\cdot)$ is the cumulative distribution function of a standard normal distribution,
and the last equality follows by the Reflection Principle (see, e.g., \citet[pp. 297]{Shreve2004}).
Let $R=4$. 
Then the R.H.S. of the above inequality approximately equals to $2 \times 10^{-20}$
under the parameter setting in Table \ref{tab:para}.
It is also easy to see that $\xi(R)$ is quadratic in $R$ in the present example,
and therefore, the truncation error is marginal according to the error bound in Eq. \eqref{truncation_error}. 
In view of this, we fix $R=4$ in all subsequent numerical experiments.

\subsection{Forward Simulation v.s. Artificial Simulation}
Next, we would like to show the limitations of the forward simulation based on control randomization
in generating random samples of the state process. 
Below, we present some prevalent control randomization methods.
\begin{description}
    \item[\bf (CR0)] Given a simulated $X_{t}^{(m)}:=\left(W_{t}^{(m)},I_{t}^{(m)}\right)^{\T}$,
    the PH's action $a_{t}^{(m)}$ is simulated from a degenerated distribution 
    with one single point mass at $\left(G(I_{t})P_{0},1\right)^{\T}$.

    \item[\bf (CR1)] Given $X_{t}^{(m)}$,
    the DM's action $a_{t}^{(m)}$ is simulated from a discrete uniform distribution 
    with support set $\big\{(0,0)^{\T}, \left(G(I_{t})P_{0},1\right)^{\T},  \left(W_{t},1\right)^{\T}\big\}$
    if $I_{t}^{(m)}=0$; 
    and 
    $\big\{(0,0)^{\T}, \left(G(I_{t})P_{0},1\right)^{\T},  \left(W_{t},1\right)^{\T}\big\}$, 
    otherwise. 
    
    \item[\bf (CR2)] Given $X_{t}^{(m)}$,
    the DM's action $a_{t}^{(m)}$ is simulated from a discrete uniform distribution 
    with support set $\big\{(0,0)^{\T}, \left(G(I_{t})P_{0},1\right)^{\T}\big\}$
    if $I_{t}^{(m)}=0$; 
    and $\big\{(0,0)^{\T}, \left(G(I_{t})P_{0},1\right)^{\T}\big\}$, 
    otherwise.
\end{description}
Given the above rules of generating the PH's action,
one may simulate the state process in a forward manner in accordance with Steps 2.1 and 2.2 of the FSBU algorithm;
see Section \ref{sec:LSMC_review}.

{\bf (CR0)} is first proposed by \cite{Huang2016} in the context of pricing Guaranteed Lifelong Withdrawal Benefit,
a particular type of variable annuity policy.
It initializes the withdrawal at $t=1$
and the resulting simulated the state variable $I_{t}^{(m)}$ 
(resp., its accompanying post-action value $S_{t}^{I}\left(I_{t}^{(m)},a_t^{(m)}\right)$) 
equals a fixed value for all $t=1,2,\dots,T-1$
although $I_{t}$ (resp., $S_{t}^{I}\left(I_{t},a_t\right)$), in principle, can take any value in 
$\{0,1,\dots,t-1\}$ (resp., $\{0,1,\dots,t\}$).
A consequential annoying issue is that the obtained estimate for the value function/continuation function
is invariant to the first-withdrawal-time $I_{t}$,
which is not sensible since the later the PH initializes the withdrawal
the larger guaranteed amount $G(I_{t})$ she could enjoy in remaining contract life.

{\bf (CR1)} uniformly simulates the PH's action from its feasible set $A_t(X_{t})$; see Eq. \eqref{feasible_set}.
By virtue of this, there always exists some paths with $I_{t}^{(m)}=0$ 
which corresponds to the scenario that the withdrawal has not been initialized.
This in turn guarantees that, in principle, $I_{t}^{(m)}$ 
(resp., $S_{t}^{I}\left(I_{t}^{(m)},a_t^{(m)}\right)$)
can take any value in $\{0,1,\dots,t-1\}$ (resp., $\{0,1,\dots,t\}$).
However, this strategy is also not satisfactory: 
an overwhelming portion of paths are absorbed by the state $W_{t}=0$,
i.e., the depletion of investment account,
and very sparse sample points of the investment account are positive.
This is graphically illustrated in the top panel of Figure \ref{fig:sample_path_invest}
where $1000$ sample paths are plotted for the clarity of presentation. 
So it is not hard to expect that the accuracy of the regression estimate 
is severely impaired over $\mathcal{K}_{t,R}$.

To alleviate the serious problem mentioned above,
{\bf (CR2)} discards the strategy of depleting the investment account, i.e., $\left(W_{t},1\right)^{\T}$, in simulating the PH's action.
And therefore, the simulated investment account value $W_{t}^{(m)}$ 
can spread over a wider range than that accompanying {\bf (CR1)};
see the bottom panel of Figure \ref{fig:sample_path_invest}.
This phenomenon is more palpable from the histograms of $W_{T-1}^{(m)}$
which are collected by Figure \ref{fig:hist_invest}.
Nevertheless, {\bf (CR2)}'s performance in simulating the $I_{t}^{(m)}$ is undesirable:
Figure \ref{fig:hist_withdrawal_time} shows that a substantial portion of 
sample points of the first-withdrawal-time $I_{t}^{(m)}$ are concentrated in first few values that $I_{t}$ can take.
To understand the crux, we note that at the first possible withdrawal date, 
one-half of sample paths exhibit the initiation of the withdrawal; 
among the remaining paths, one-half of them witness the withdrawal in the consecutive withdrawal date.
Therefore, the portion of positive $I_{t}^{(m)}$ declines at an exponential rate as $t$ increases
which is in line with Figure \ref{fig:hist_withdrawal_time}.
In view of this, it can be expected that the consequential numerical estimate for the
value function sustains significant error at state $x=(W,I)^{\T}$ with a large $I$.

Overall, none of the above rules {\bf (CR0)}--{\bf (CR2)} gives agreeable performance.
It is hard to figure out an ideal way to randomize the PH's action 
which can sidestep the thorny issues mentioned above.
This shows one drawback of binding together control randomization and forward simulation
in addition to the issue of computational cost;
see also the item ``Limitation of control randomization" of Section \ref{sec:LSMC_review}.

To circumvent the annoying problems mentioned above,
in the sequel numerical experiments, we simulate the post-action value of the state process 
at each time step as follows: 
$X_{t^{+}}^{(m)}:=\left(W_{t^{+}}^{(m)}, I_{t^{+}}^{(m)}\right)$
where $W_{t^{+}}^{(m)}$ and $I_{t^{+}}^{(m)}$ are simulated from two independent uniform distributions
with support sets $(0,R)$ and $\{0,1,\dots,t\}$, respectively.
This ensures the post-action value evenly distributed over $\widehat{\mathcal{K}}_{t,R}$.

\begin{figure}[ht]
\centering
\includegraphics[width=.8\textwidth]{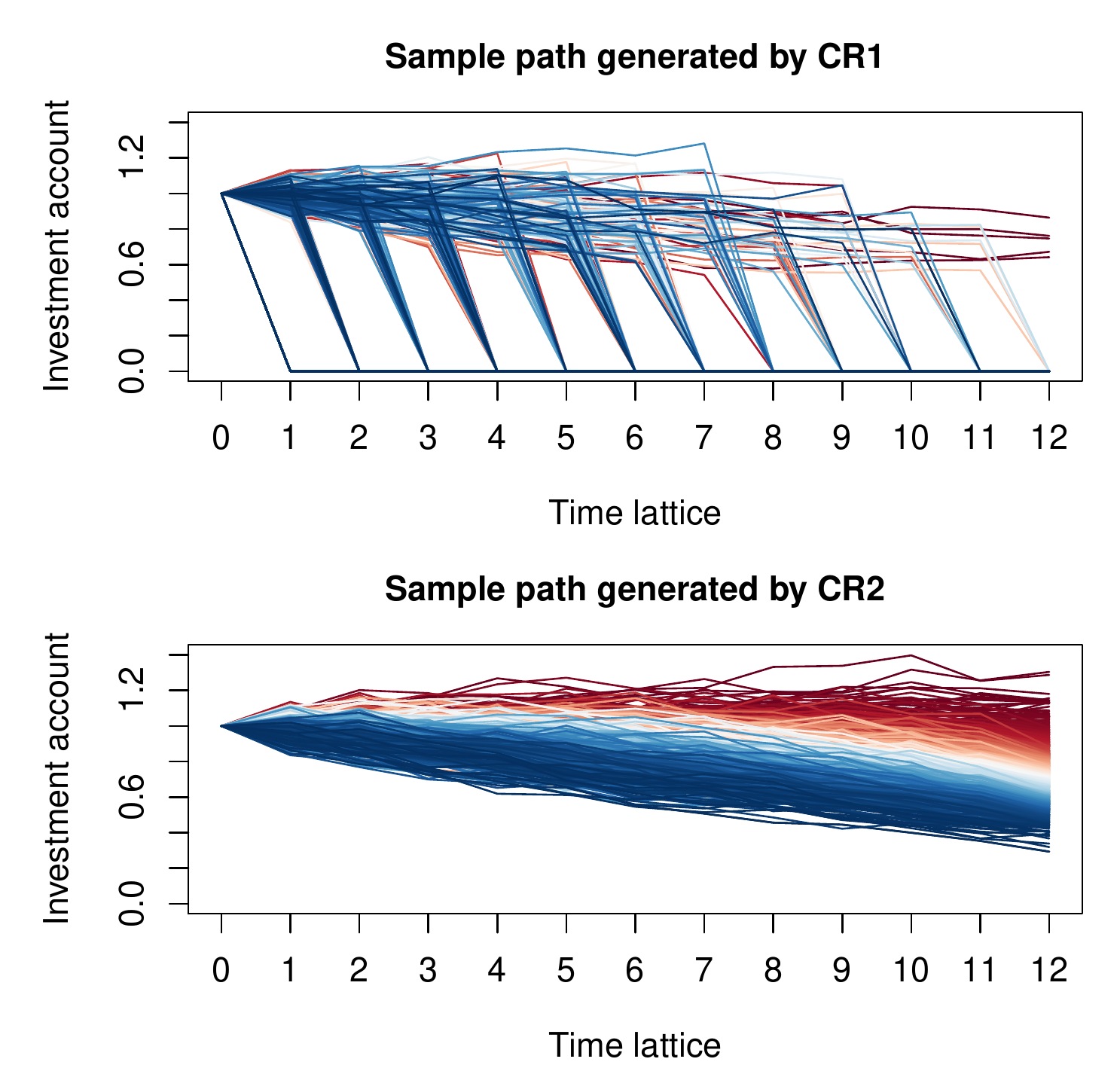}
\captionsetup{width=0.9\textwidth}
\caption{Sample paths of the investment account generated by control randomization methods {\bf (CR1)} and {\bf (CR2)}.}
\label{fig:sample_path_invest}
\end{figure}

\begin{figure}[hbt!]
\centering
\includegraphics[width=.8\textwidth]{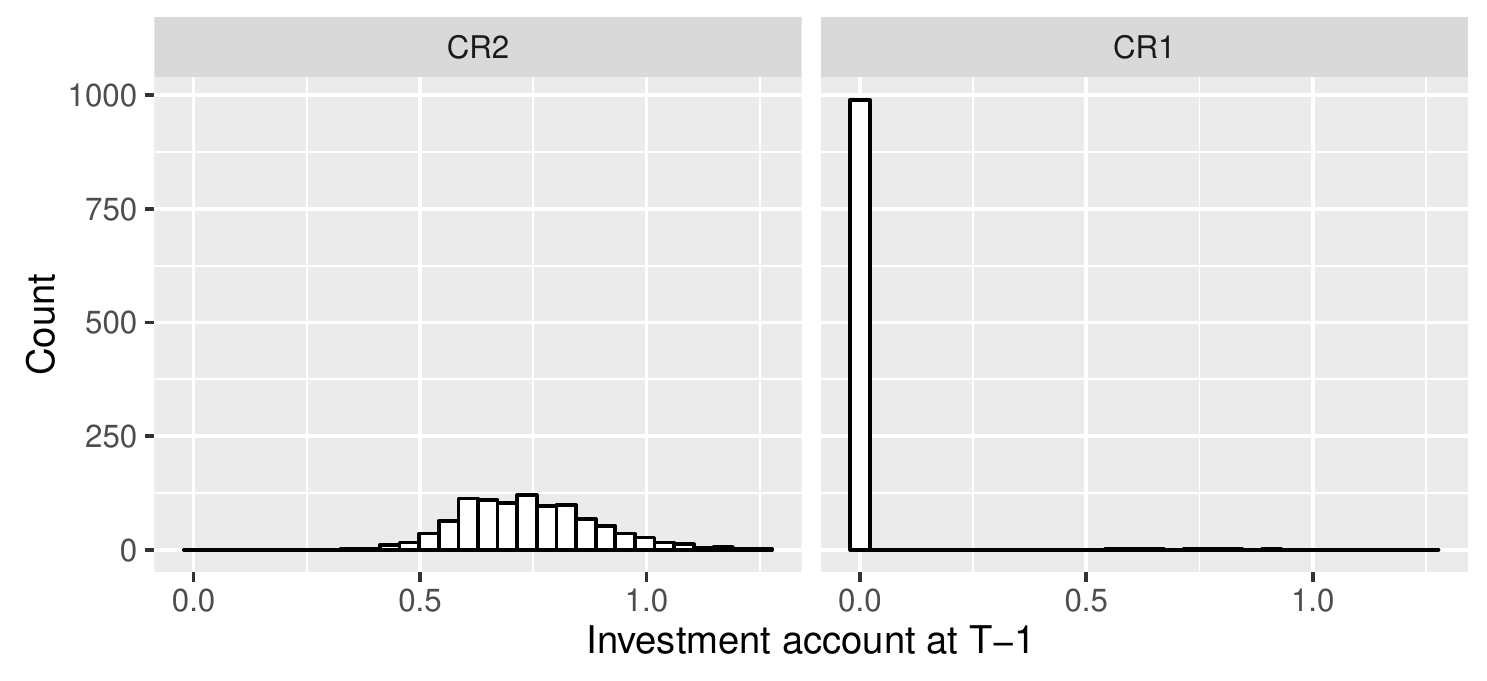}%
\captionsetup{width=0.9\textwidth}
\caption{Histograms of $W_{T-1}^{(m)}$ generated by control randomization methods {\bf (CR1)} and {\bf (CR2)}.}
\label{fig:hist_invest}
\end{figure}

\begin{figure}[hbt!]
    \centering
    \includegraphics[width=0.6\textwidth,height=.5\textwidth]{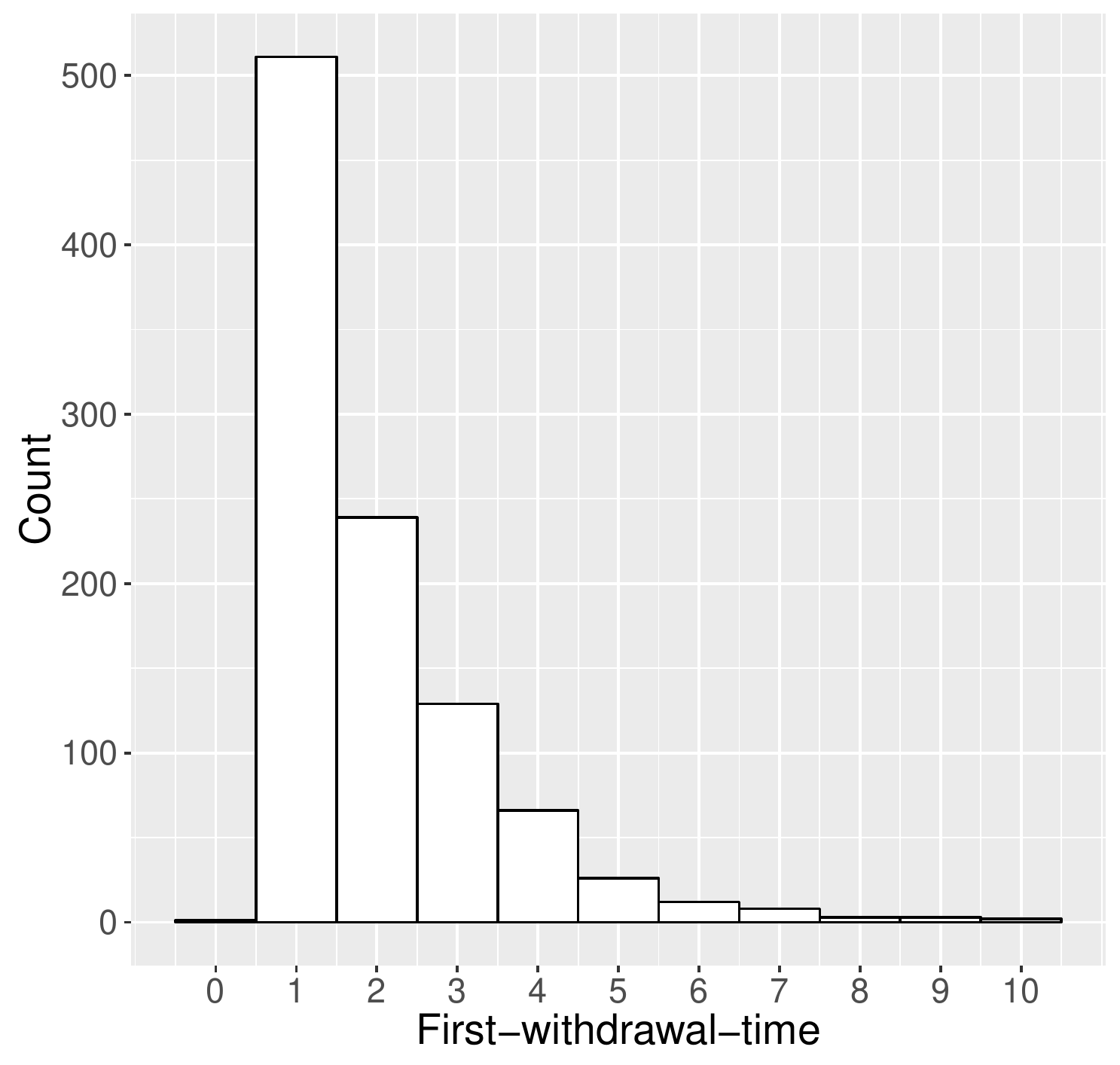}
    \captionsetup{width=0.9\textwidth}
    \caption{Histogram of $I_{11}^{(m)}$ generated by control randomization method {\bf (CR2)}.}
    \label{fig:hist_withdrawal_time}
\end{figure}

\subsection{Raw Sieve Estimation v.s. Shape-Preserving Sieve Estimation}
In the sequel, we conduct several numerical experiments to compare 
the regression estimates for the continuation function produced by two regression methods:
the raw sieve estimation (RSE) method and the shape-preserving sieve estimation (SPSE) method.
The RSE and SPSE are essentially the linear sieve estimation method discussed in 
Section \ref{sec:sieve_estimation} with sieve spaces 
\eqref{linear_sieve_space} and \eqref{shape_sieve_space}, respectively.
It is easy to show that the continuation function $k_1 \longmapsto \tilde{C}_{t}^{\textup{E}}(k)$ is monotone
and therefore we incorporate this shape constraint in the SPSE method.
The expression of the accompanying matrix $\mathbf{A}_{J}$ is given in Appendix \ref{app:constraint_matrix}.
The RSE method is equivalent to the least-squares method commonly adopted in the literature;
see, for instance, \cite{Longstaff2001}.
The aim of this subsection is to show the advantage of incorporating shape constraints 
in the regression step of an LSMC algorithm.

In the first numerical experiment, we compare the SPSE and RSE for the continuation function at each time step.
For the fairness of the comparison, for both methods,
we take $\bm{\phi}(\cdot)=\big(\phi_{0}(\cdot),\dots,\phi_{J}(\cdot)\big)^{\T}$ 
as a vector of univariate Bernstein polynomials up to order $J=20$ in both sieve estimation methods.
Figure \ref{fig:cont_fun-1} collects the plots of regression estimates as a function of $k_1$ at odd time steps with $k_2=0$.
To better show the subtle difference between the estimates produced by SPSE and RSE,
the plots are restricted on the interval $[0,1]$.
From Figure \ref{fig:cont_fun-1}, 
we can see that the discrepancy between the regression estimates accompanying SPSE and RSE is 
not conspicuous at large time step but becomes more significant as the time step goes down.
Despite the continuation function, in principle, is a monotone function with respect to $k_1$,
the dotted lines in Figure \ref{fig:cont_fun-1} show that
its regression estimate produced by the RSE does not inherit this monotonicity
and loses certain economic interpretations, accordingly.
This issue is more serious at smaller time steps as shown by the bottom panel of Figure \ref{fig:cont_fun-1}.
This is not surprising because once the monotonicity is lost at a certain time step,
the regression estimate obtained in the consecutive time step will be influenced,
which in turn exaggerates the problem.
In contrast, as depicted by the solid lines in Figure \ref{fig:cont_fun-1},
the SPSE method always preserves the monotonicity of the continuation function
and therefore the corresponding regression estimates are more economically sensible.
This shows the first advantage of the SPSE method in terms of preserving certain shape properties of the continuation function.


\begin{figure}[hbt!]
    \centering
    \includegraphics[width=1.0\textwidth]{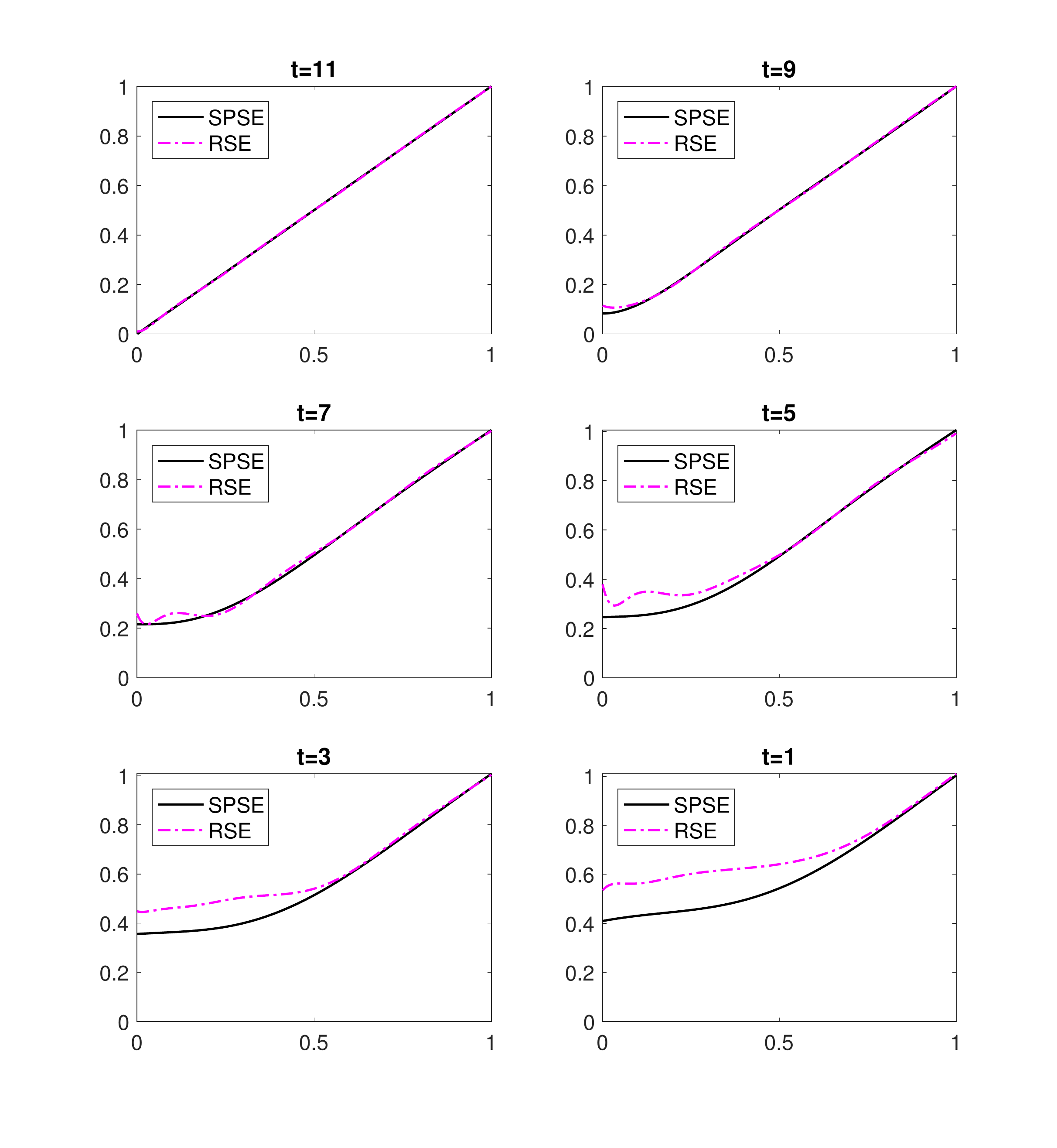}
    \captionsetup{width=0.9\textwidth}
    \caption{Regression estimates of SPSE and RSE for $k_1 \longmapsto \tilde{C}_{t}^{\textup{E}}(k_1,0)$ over $[0,1]$.
     $M=10^5$ sample paths are generated and $J=20$ basis functions are used.}
    \label{fig:cont_fun-1}
\end{figure}

Next, we implement the BSBU algorithm to compute $\tilde{V}_{0}^{\textup{E}}(X_0)$ with $X_{0}=(1,0)^{\T}$
where the SPSE and RSE are employed, respectively.
These estimates approximate the no-arbitrage price of the VA policy at the inception
and therefore are of most interest in the present context;
see the last paragraph of Section \ref{sec:model_setup_VA}.
It is worth noting that $\tilde{V}_{0}^{\textup{E}}(X_0)$ is random due to the randomness of the simulated sample;
see also Remark \ref{rem:notation_numerical_estimate}.
And therefore, we repeat the BSBU algorithm 40 times in order to
study the stability of $\tilde{V}_{0}^{\textup{E}}(X_0)$ under a finite sample size.
Table \ref{tab:price} summarizes the mean and standard deviation of $\tilde{V}_{0}^{\textup{E}}(X_0)$
under different pairs of $M$ and $J$.
The ``S.d." column of the table discloses that the standard deviation 
accompanying the SPSE is nearly one half of that 
associated with the RSE under all numerical settings.
For the numerical settings 0-2 of Table \ref{tab:price},
Figure \ref{fig:hist_price} delineates the corresponding density plots of the 40 estimates.
By comparing the left and right panels of Figure \ref{fig:hist_price},
we can easily perceive that the numerical estimates accompanying SPSE are less volatile
as reflected by the more spiked shape of the corresponding density plots.
This observation is consistent with Table \ref{tab:price}.
To sum up, the SPSE surpasses the RSE in terms of smaller standard deviation
of the resulting numerical estimate for the optimal value function at the initial state.

\begin{table}[ht]
\begin{center}
\captionsetup{width=0.9\textwidth}
\caption{Mean and standard deviation of $\tilde{V}_{0}^{\textup{E}}(X_0)$ produced by SPSE and RSE methods.
The results are obtained by repeating the BSBU algorithm 40 times.}
\label{tab:price}
\vspace{0.5ex}
\renewcommand{\arraystretch}{1.2}
\begin{tabular*}{0.9\textwidth}{c @{\extracolsep{\fill}} ccccc}
\multicolumn{1}{c}{\multirow{2}{*}{Setting}}&
\multicolumn{1}{c}{\multirow{2}{*}{$(M,J)$}}   &\multicolumn{2}{c}{SPSE} &\multicolumn{2}{c}{RSE}\\
\cline{3-6}\multicolumn{1}{c}{}& \multicolumn{1}{c}{}
&  \multicolumn{1}{c}{Mean}     & \multicolumn{1}{c}{S.d.}
&  \multicolumn{1}{c}{Mean}  & \multicolumn{1}{c}{S.d.}\\
\hline
0&$(1\times10^5, 15)$          &$0.9940$	&${\bf 0.0040}$	&$1.0045$	&$0.0091$	\\
1&$(1\times10^5, 20)$          &$0.9916$	&${\bf 0.0035}$	&$1.0028$	&$0.0070$	\\
2&$(1\times10^5, 25)$          &$0.9969$	&${\bf 0.0031}$	&$1.0029$	&$0.0056$	\\
3&$(2\times10^5, 20)$          &$0.9913$	&${\bf 0.0025}$	&$1.0012$	&$0.0058$	\\
4&$(4\times10^5, 20)$          &$0.9910$	&${\bf 0.0015}$	&$0.9983$	&$0.0034$	\\
\hline
\end{tabular*}
\end{center}
\end{table}

\begin{figure}[hbt!]
    \centering
    \includegraphics[width=0.8\textwidth]{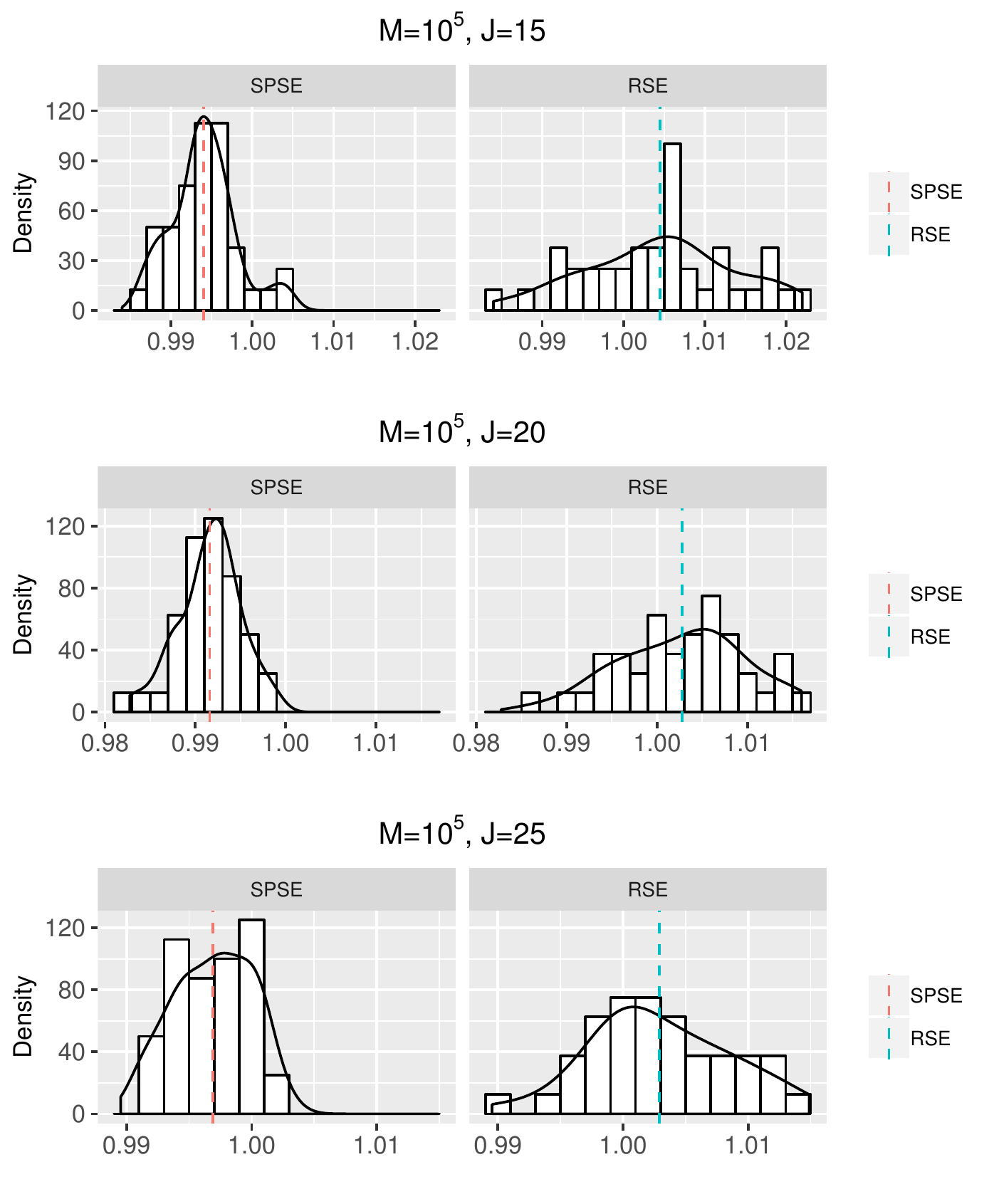}
    \captionsetup{width=0.9\textwidth}
    \caption{Density plots of $\tilde{V}_{0}^{\textup{E}}(X_0)$ produced by SPSE and RSE methods
    under 40 repeats of the BSBU algorithm.
    The dotted line corresponds to the sample mean.}
    \label{fig:hist_price}
\end{figure}

From the settings 0-2 of Table \ref{tab:price}, 
we also observe that for both methods,
the change of the mean of the numerical estimate is not substantial
as the number of basis functions $J$ hikes from $15$ to $25$.
In the settings 1, 3 and 4 of Table \ref{tab:price},
we fix $J=20$ and increase the number of simulated paths $M$
from $10^5$ to $4\times 10^5$.
We witness that standard deviation decreases as the number of simulated paths climbs.
This descending trend is also confirmed by the box plots depicted in Figure \ref{fig:box_price}:
the height of the box shrinks as the number of simulated paths hikes.
All of these show the convergence of the BSBU algorithm
which is in line with the convergence result established in Section \ref{sec:conv_BSBU}.

\begin{figure}[hbt!]
    \centering
    \includegraphics[width=0.8\textwidth]{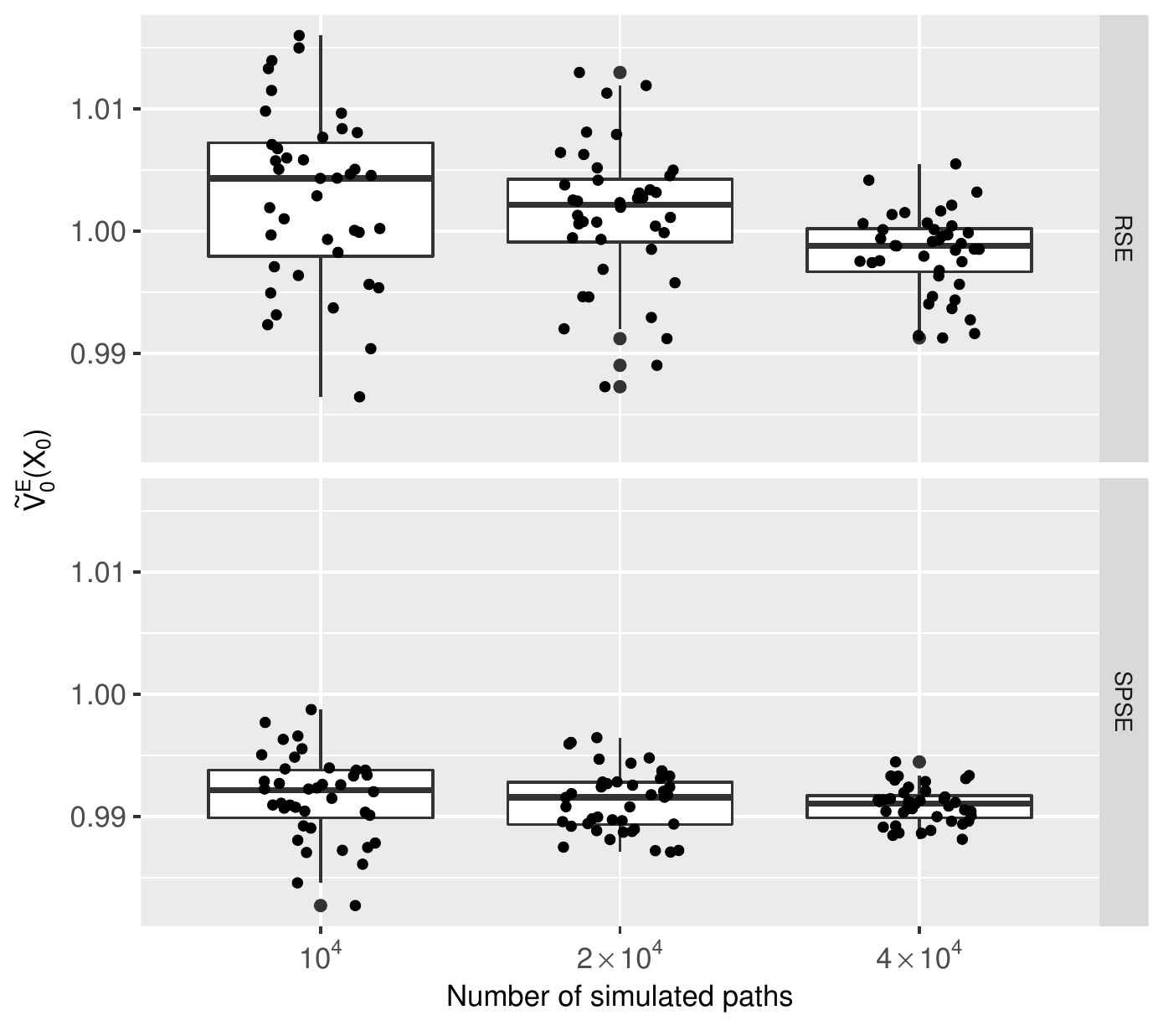}
    \captionsetup{width=0.9\textwidth}
    \caption{Box plots of $\tilde{V}_{0}^{\textup{E}}(X_0)$ produced by SPSE and RSE methods
    under 40 repeats of the BSBU algorithm. 
    $J$ is fixed as $20$ and $M$ varies from $10^5$ to $4\times 10^5$.}
    \label{fig:box_price}
\end{figure}

Overall, the advantages of the SPSE over the RSE are extant at least in two-fold.
Firstly, the SPSE produces economically sensible regression estimates by inheriting certain
shape properties of the true continuation function.
Secondly, the consequential estimate for the optimal value function
accompanying the SPSE method is less volatile than that produced by the RSE method
under a finite number of simulated sample paths.

\section{Conclusion}
\label{sec:conclusion_BSBU}
This paper develops a novel LSMC algorithm,
referred to as Backward Simulation and Backward Updating (BSBU) algorithm, 
to solve discrete-time stochastic optimal control problems.
We first introduce an auxiliary stochastic control problem where the state process only takes value in a compact set.
This enables the BSBU algorithm to successfully sidestep extrapolating value function estimate.
We further show the optimal value function of the auxiliary problem is a legitimate approximation for that of the original problem
with an appropriate choice of the truncation parameter.
To circumvent the drawbacks of forward simulation and control randomization,
we propose to directly simulate the post-action value of the state process from an artificial probability distribution.
The pivotal idea behind this artificial simulation method is that
the continuation function is solely determined by the distribution of random innovation term.
Moreover, motivated by the shape information of the continuation function,
we introduce a shape-preserving sieve estimation technique to alleviate the computational burden of tuning parameter selection involved in the regression step of an LSMC algorithm.
Furthermore, convergence result of the BSBU algorithm is established by resorting to 
the theory of nonparametric sieve estimation.
Finally, we confirm the merits of the BSBU algorithm through 
an application to pricing equity-linked insurance products 
and the corresponding numerical experiments.


\bibliographystyle{plainnat}
\bibliography{ref}

\newpage
\appendix
{
\numberwithin{equation}{section}
\section{Supplements for Sieve Estimation Method}
\label{app:sieve}
\subsection{Forms of Matrix $\mathbf{A}_{J}$}
\label{app:constraint_matrix}
To make the paper self-contained, we collect several forms of the constraint matrix $\mathbf{A}_{J}$ 
in \eqref{shape_sieve_space} which ensures monotonicity, convexity, or concavity
of the sieve estimate \eqref{regression_estimate};
see \cite{Wang2012-1,Wang2012-2} for a justification.
For the simplicity of notation, we only address the case $r=1$.

\begin{description}
\item {\bf Monotonicity}\quad 
Suppose the conditional mean $g(\cdot)$ defined in Eq. \eqref{regression_fun} is monotone.
Then the corresponding monotonicity-preserved sieve estimate $\hat{g}(\cdot)$ is obtained from 
\eqref{regression_estimate} with $\mathcal{H}_J$ given by Eq. \eqref{shape_sieve_space}
and 
\bas
\mathbf{A}_{J}=
\left(
\begin{matrix}
-1&1&0&\cdots&0\\
0&-1&1&1&\cdots\\
 &    &\ddots&&\\
0&\cdots&0&-1&1\\
\end{matrix}
\right)_{J\times(J+1)}.
\eas

\item {\bf Convexity/Concavity}\quad 
If $g(\cdot)$ is convex, we choose the matrix $\mathbf{A}_J$ as 
\bas
\mathbf{A}_{J}=
\left(
\begin{matrix}
1&-2&1&0&\cdots&0\\
0&1&-2&1&\cdots&0\\
 &    &\ddots&&&\\
0&\cdots&0&1&-2&1\\
\end{matrix}
\right)_{(J-1)\times(J+1)}.
\eas
Moreover, the matrix $\mathbf{A}_J$ accompanying a concave $g(\cdot)$
is obtained by taking negative of the above matrix.

\item {\bf Convexity and Monotonicity}\quad 
If $g(\cdot)$ is convex and monotone, the corresponding $\mathbf{A}_J$ is given by
\bas
\mathbf{A}_{J}=
\left(
\begin{matrix}
-1&1&0&\cdots&\cdots&0\\
1&-2&1&0&\cdots&0\\
0&1&-2&1&\cdots&0\\
 &&\ddots&&&\\
0&\cdots&0&1&-2&1\\
\end{matrix}
\right)_{J\times(J+1)}.
\eas
\end{description}   

\subsection{A Data-driven Choice of $J$}
\label{app:CV}
Below we present some common methods of choosing the number of basis functions $J$
in a sieve estimation method; see, e.g., \cite{Li1987}.
\begin{description}
\item[]
{\bf Mallows's $C_{p}$}\quad 
For a discrete set $\mathcal{J} \subseteq \mathbb{N}$,
$J$ is determined by solving the following minimization problem: 
\bas
\hat{J}=\arg \min_{J \in \mathcal{J}} 
\frac{1}{M} \sum_{m=1}^{M}\left[U^{(m)}- \hat{g} \left(Z^{(m)}\right)\right]^2
+ 2 \hat{\sigma}^2 \big(J/M\big),
\eas
where $\hat{g}(\cdot)$ is given in \eqref{regression_estimate} and
$
\hat{\sigma}^2:=M^{-1}\sum_{m=1}^{M} \left[U^{(m)}- \hat{g} \left(Z^{(m)}\right)\right]^2
$
which is an estimate for the variance of residual term.

\item[]
{\bf Generalized cross-validation}\quad 
$J$ is determined by
\bas
\hat{J}=\arg \min_{J \in \mathcal{J}} 
\frac{M^{-1} \sum_{m=1}^{M}\left[U^{(m)}- \hat{g} \left(Z^{(m)}\right)\right]^2}
{\left(1-\big(J/M\big)\right)^2},
\eas
with $\hat{g}(\cdot)$ given in \eqref{regression_estimate}.

\item[]
{\bf Leave-one-out cross-validation}\quad 
Select $J$ to minimize
\bas
\textup{CV}(J):=
\frac{1}{M} \sum_{m=1}^{M}\left[U^{(m)}- \hat{g}_{-m} \left(Z^{(m)}\right)\right]^2,
\eas
where $\hat{g}_{-m}(\cdot)$ is similarly obtained by Eq. \eqref{regression_estimate} 
with the sample point $\left(U^{(m)},Z^{(m)}\right)$ removed.
\end{description}
It is worth stressing that, among the above three selection methods,
the leave-one-out cross-validation method is most computationally expensive
as one has to compute the regression estimate $\hat{g}_{-m}(\cdot)$ $M$ times
in a single evaluation of $\textup{CV}(J)$.
This is clearly computationally prohibitive when $M$ is considerable,
which is particularly the case in the present context of the LSMC algorithm.
The Mallows's $C_{p}$ criterion and generalized cross-validation method
are relatively less cumbersome but still undesirable under a sizable $M$.
These show the importance of avoiding such a tuning parameter selection procedure
which is one important thrust behind proposing shape-preserved sieve estimation method.

\subsection{Technical Assumption of Sieve Estimation Method}
\label{app:assum_sieve}
We impose the following assumption accompanying the sieve estimation method discussed in Section \ref{sec:sieve_estimation} 
which follows from  \cite{Newey1997}.
\begin{assumption}
\label{assum:sieve_estimation}
\begin{description}
\item[(i)] $\left\{\left(U^{(m)}, Z^{(m)}\right)\right\}_{m=1}^{M}$ are i.i.d. 
and $Z^{(m)}$ has compact support $\mathcal{Z}$. 
Furthermore, $\var\left[\left.U^{(m)}\right|Z^{(m)}=\cdot\ \right]$ is bounded over $\mathcal{Z}$.

\item[(ii)]
There exists a sequence $\Upsilon(J)$ such that $\norm{\bm{\phi}}_{\infty} \leq \Upsilon(J)$
with $\norm{\cdot}$ denoting the supremum norm of a continuous function over $\mathcal{Z}$.

\item[(iii)]
For the sieve space $\mathcal{H}_{J}$ defined either in Eq. \eqref{linear_sieve_space} or Eq. \eqref{shape_sieve_space},
there exists a $(J+1)$-by-1 vector $\tilde{\bm{\beta}}$ and a sequence $\rho_{J}$ such that
$\rho_{J} \longrightarrow 0$ as $J \longrightarrow \infty$, and
\ba
\label{oracle}
\inf_{h(\cdot)\in \mathcal{H}_{J}}
\norm{h-g}_{\infty}
=\norm{\tilde{\bm{\beta}}^{\T}\bm{\phi}-g}_{\infty}
=O\left(\rho_{J}\right),
\ea
where we remind that $g(\cdot):=\e \left[\left.U^{(m)}\right|Z^{(m)}=\cdot \ \right]$.

\item[(iv)]  
Let $\Phi:=\e \left[\bm{\phi}\left(Z^{(m)}\right)\bm{\phi}^{\T}\left(Z^{(m)}\right)\right]$.
There exists a positive constant $\underline{c}_{\Phi}$ independent of $J$ such that
$
0< \underline{c}_{\Phi} \leq \lambda_{\min}\left(\Phi\right) 
\leq \lambda_{\max}\left(\Phi\right)\leq \bar{c}_{\Phi}<\infty,
$
with $\lambda_{\min}\left(\Phi\right)$ 
and $\lambda_{\max}\left(\Phi\right)$
denoting the smallest and largest eigenvalues of $\Phi$, respectively.

\item[(v)]
As $M \longrightarrow \infty$, $J \longrightarrow \infty$, and $\Upsilon^2(J)J/M \longrightarrow 0$.

\end{description}
\end{assumption}
We give some comments on the above technical conditions.
\begin{enumerate}
\item 
The i.i.d. condition in Part (i) of the above assumption
discloses the necessity of generating an independent sample at each time step in an LSMC algorithm;
see also the discussion in the earlier item ``Cost of forward simulation" of Section \ref{sec:LSMC_review}.
Part (i) further requires $Z^{(m)}$ has a compact support, 
which is conventional in literature see, e.g., \cite{Newey1997} and \cite{Chen2007}.
In the context of BSBU algorithm,
this shows that restraining the state process into a bounded domain
is not only beneficial in eliminating undesirable extrapolation 
but also indispensable in guaranteeing the convergence of the regression estimate to the continuation function.
This has also been pointed out in the literature,
see, e.g., \cite{Stentoft2004} and \cite{Zanger2013}.

\item
Part (ii) specifies how the magnitude of $\bm{\phi}(\cdot)$ is amplified
as the number of basis function functions grows up. 
In particular, \cite{Newey1997} shows that $\Upsilon(J)=O\left(\sqrt{J}\right)$ for B-splines
and $\Upsilon(J)=O(J)$ for power series;
for the cases of other types of basis functions, we refer to \cite{Chen2007}.

\item
Part (iii) states that there exists a function 
$\tilde{\bm{\beta}}^{\T}\bm{\phi}(\cdot)$ 
in the sieve space $\mathcal{H}_{J}$ 
that ``best" approximates the conditional mean function $g(\cdot)$
under the supremum norm; see Figure \ref{fig:sieve_estimation} for a graphical illustration.
The existence of vector $\tilde{\bm{\beta}}$ 
(referred to as \textit{oracle})
is guaranteed by the convexity of sieve space $\mathcal{H}_J$.
For the sieve space \eqref{shape_sieve_space}, 
the existence of $\rho_{J}$ relies on the convexity, concavity or monotonicity of the function $g(\cdot)$
which follows by the Property 3.2 of \cite{Wang2012-2}.
Figure \ref{fig:sieve_estimation} depicts the relationship between
$\tilde{\bm{\beta}}^{\T}\bm{\phi}(\cdot)$ and $g(\cdot)$:
their discrepancy vanishes as $J$ increases
and, for a fixed $J$, the sieve estimate $\hat{\bm{\beta}}^{\T}\bm{\phi}(\cdot)$
converges to $\tilde{\bm{\beta}}^{\T}\bm{\phi}(\cdot)$
as the sample size $M$ approaches infinity.
Therefore, one may view the sieve estimation as a two-stage approximation for 
the conditional mean function $g(\cdot)$.

\item
The condition in Part (iv) ensures the design matrix of the regression problem is nonsingular
with a high probability and does not blow up as $J$ approaches infinity.
Finally, Part (v) prescribes the growth rates of $J$ and $M$
in order to avoid overfitting or underfitting.
\end{enumerate}

\begin{figure}
\centering
\begin{adjustbox}{max totalsize={.8\textwidth}{.8\textheight},center}
\tikzset{set/.style={draw,circle,inner sep=0pt,align=center}}
\begin{tikzpicture}
\draw [fill=blue](7,0) circle (.08);
\node at (7.8,0) {$g(\cdot)$};
\node at (1,3) {{\bf Sieve Spaces $\mathcal{H}_{J}$}};
\node[set,text width=1cm] (J1) at (0,0) {$J=1$};
\draw [fill=lime](.5,0) circle (.08);
\draw[-latex,thick,red!80](6,-1)node[right]{$\tilde{\bm{\beta}}^{\T}\bm{\phi}(\cdot)$} to[out=180,in=270] (0.5,0);
\node[set,text width=3cm] (J2) at (0.5,0)  {};
\node at (1.2,0) {$J=2$};
\draw [fill=lime](2,0) circle (.08);
\draw[-latex,thick,red!80](6,-1) to[out=180,in=270] (2,0);
\node[set,text width=5cm] (J3) at (1,0)  {};
\node at (2.7,0) {$J=3$};
\draw [fill=lime](3.5,0) circle (.08);
\draw[-latex,thick,red!80](6,-1)
to[out=180,in=270] (3.5,0);
\draw[red,<->,thick] (3.6,0) to
node [above, align=center,black] {$\rho_{J} \rightarrow 0,\ \textup{as}\ J \rightarrow \infty$}
(6.9,0);
\end{tikzpicture}
\end{adjustbox}
\captionsetup{width=0.9\textwidth}
\caption{A diagram illustrating the relationship between $\tilde{\bm{\beta}}^{\T}\bm{\phi}(\cdot)$ and $g(\cdot)$.}
\label{fig:sieve_estimation}
\end{figure}
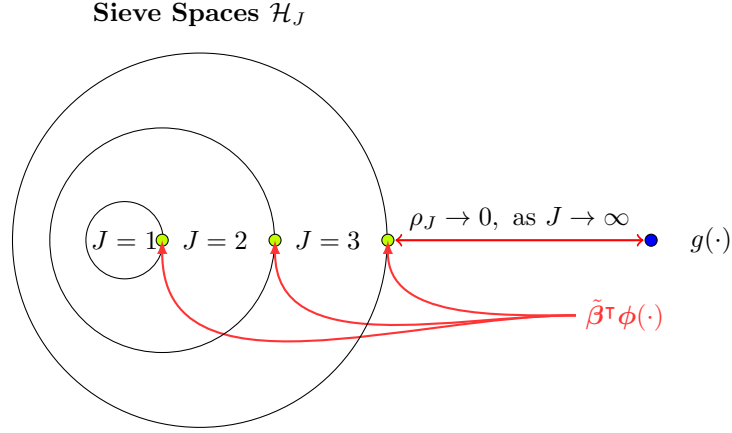

\section{Proofs of Statements}
\label{app:proofs}

\subsection{Proof of Proposition \ref{prop:transition_eq_aux_process}}
\label{app:proof_prop-1}
\subsubsection{Preliminary}
\begin{lemma}
\label{lemma:201809051430}
For any $\mathcal{F}$-adapted process ${\sf a}=\{a_{t}\}_{t\in \mathcal{T}_{0}}$, the following statements hold:
\begin{description}
\item[(i)]
$\left\{\tau^{R} \leq t\right\}=\left\{X_{t}^{R} \in \partial \mathcal{X}_{R}\right\}$
for $t=1,2,\dots,T$;

\item[(ii)]
$\left\{\tau^{R}=t+1\right\}=
\left\{X_{t}^{R} \in \mathring{\mathcal{X}}_{R},
\ S\left(X_{t}^{R}, a_{t}, \varepsilon_{t+1}\right)\notin \mathring{\mathcal{X}}_{R}\right\}$
for $t=0,1,\dots,T-1$,
\end{description}
where $\tau^{R}$ and $X_{t}^{R}$ are defined in Eqs. 
\eqref{first_exit_time} and \eqref{aux_state_process}, respectively.
\end{lemma}

\begin{proof}[Proof of Lemma \ref{lemma:201809051430}]
\begin{description}
\item[(i)]
According to Eq. \eqref{aux_state_process},
we observe 
\bas
    \left\{X_{t}^{R} \in \partial \mathcal{X}_{R}\right\}
&=&   \left\{X_{t} \in \partial \mathcal{X}_{R}, \tau^{R} > t \right\}
        \cup \left\{\mathcal{Q}\left(X_{\tau^{R}\wedge t}\right) \in \partial \mathcal{X}_{R}, \tau^{R} \leq t \right\}\\
&=&  \left\{\mathcal{Q}\left(X_{\tau^{R}\wedge t}\right) \in \partial \mathcal{X}_{R}, \tau^{R} \leq t \right\},
\eas
where the second identity is by the definition of the stopping time $\tau^{R}$
and the fact that $\partial \mathcal{X}_{R} \cap \mathring{\mathcal{X}}_{R}=\emptyset$.
To show the statement in Part (i) of Lemma \ref{lemma:201809051430}, it suffices to prove 
$\left\{\tau^{R} \leq t\right\} \subseteq \left\{\mathcal{Q}\left(X_{\tau^{R}\wedge t}\right)  \in \partial \mathcal{X}_{R}\right\}$.
Indeed, $\tau^{R} \leq t$ implies $X_{\tau^{R}\wedge t} \notin \mathring{\mathcal{X}}_{R}$,
and thus $\mathcal{Q}\left(X_{\tau^{R}\wedge t}\right)  \in \partial \mathcal{X}_{R}$.

\item[(ii)]
In view of Part (i) and Eq. \eqref{aux_state_process}, we obtain 
\bas
\left\{X_{t}^{R} \in \mathring{\mathcal{X}}_{R}\right\} 
=\left\{X_{t}^{R} \in \partial \mathcal{X}_{R}\right\}^{\textup{c}}
=\left\{\tau^{R} > t\right\}
\subseteq \left\{X_{t}^{R}=X_t\right\}.
\eas
Therefore, we obtain
\bas
\left\{X_{t}^{R} \in \mathring{\mathcal{X}}_{R},
\ S\left(X_{t}^{R}, a_{t}, \varepsilon_{t+1}\right)\notin \mathring{\mathcal{X}}_{R}\right\}
&=& \left\{X_{t}^{R} \in \mathring{\mathcal{X}}_{R},\ X_{t}^{R}=X_t,\
S\left(X_{t}^{R}, a_{t}, \varepsilon_{t+1}\right)\notin \mathring{\mathcal{X}}_{R},\ \tau^{R} > t\right\}\\
&=& \left\{X_{t} \in \mathring{\mathcal{X}}_{R},\ 
S\left(X_t, a_{t}, \varepsilon_{t+1}\right)\notin \mathring{\mathcal{X}}_{R},\ 
\tau^{R} > t\right\}\\
&=&\left\{X_{t} \in \mathring{\mathcal{X}}_{R},\ X_{t+1}\notin \mathring{\mathcal{X}}_{R},\ 
\tau^{R} > t\right\} 
= \left\{\tau^{R}=t+1\right\}.
\eas
This proves Part (ii) of Lemma \ref{lemma:201809051430}.
\end{description}
\end{proof}

\subsubsection{Proof of the Main Result}
\begin{proof}[Proof of Proposition \ref{prop:transition_eq_aux_process}]
By exploiting Lemma \ref{lemma:201809051430} and Eq. \eqref{aux_state_process}, we get
\bas
X_{t+1}^{R}&=&X_{t+1}\I_{\left\{\tau^{R}>t+1\right\}} + \mathcal{Q}\left(X_{\tau^{R}\wedge (t+1)}\right)\I_{\left\{\tau^{R}\leq t+1\right\}}\\
&=& X_{t+1}\I_{\left\{\tau^{R}>t+1\right\}} 
+ \mathcal{Q}\left(X_{\tau^{R}\wedge t}\right)\I_{\left\{\tau^{R}\leq t\right\}}
+ \mathcal{Q}\left(X_{t+1}\right)\I_{\left\{\tau^{R}= t+1\right\}}\\
&=& S\left(X_{t}, a_{t}, \varepsilon_{t+1}\right)\I_{\left\{\tau^{R}>t+1\right\}} 
+ \mathcal{Q}\left(X_{\tau^{R}\wedge t}\right)\I_{\left\{\tau^{R}\leq t\right\}}\\
&&+ \mathcal{Q}\left(S\left(X_{t}, a_{t}, \varepsilon_{t+1}\right)\right)
\I_{\left\{\tau^{R}= t+1\right\}}\\
&=& S\left(X_{t}^{R}, a_{t}, \varepsilon_{t+1}\right)\I_{\left\{\tau^{R}>t+1\right\}} 
+ X_{t}^{R}\I_{\left\{\tau^{R}\leq t\right\}}\\
&&+ \mathcal{Q}\left(S\left(X_{t}^{R}, a_{t}, \varepsilon_{t+1}\right)\right)
\I_{\left\{\tau^{R}= t+1\right\}}\\
&=& S\left(X_{t}^{R}, a_{t}, \varepsilon_{t+1}\right)\I_{\left\{\tau^{R}>t+1\right\}} 
+ X_{t}^{R}\I_{\left\{X_t \in \partial \mathcal{X}_{R} \right\}}\\
&&+ \mathcal{Q}\left(S\left(X_{t}^{R}, a_{t}, \varepsilon_{t+1}\right)\right)
\I_{\left\{
X_{t}^{R} \in \mathring{\mathcal{X}}_{R},
\ S\left(X_{t}^{R}, a_{t}, \varepsilon_{t+1}\right)\notin \mathring{\mathcal{X}}_{R}\right\}},
\eas
where the fourth equality follows from Eq. \eqref{aux_state_process} 
and the last equality follows from Lemma \ref{lemma:201809051430}.

The above equation together with Eqs. \eqref{transition_eq-2} and \eqref{H} yields Eq. \eqref{transition_aux_state_process}.
This completes the proof.
\end{proof}

\subsection{Proof of Theorem 1}
\label{app:proof_thm1}
\subsubsection{Preliminary}
Recall that $X$ and $X^{R}$ implicitly depend on certain actions ${\sf a}$; see Eqs. \eqref{transition_eq} and \eqref{transition_aux_state_process}, respectively.
In the sequel, we sometimes stress such dependency by writing $X_{t}({\sf a})$ (resp. $X_{t}^{R}({\sf a})$)
and $X({\sf a})$ (resp. $X^{R}({\sf a})$).

\begin{lemma}
\label{lemma:201809051441}
For the state process $X^{R}$ defined through Eq. \eqref{transition_aux_state_process},
the following statements hold.
\begin{description}
\item[(i)]
For any ${\sf a} \in \mathcal{A}$,
there exists 
$\tilde{{\sf a}} \in \mathcal{A}^{R}$ such that
$X_{t}^{R}\left({\sf a}\right)=X_{t}^{R}\left(\tilde{{\sf a}}\right)$
for all $t\in \mathcal{T}$ almost surely.

\item[(ii)]
For any $\tilde{{\sf a}} \in \mathcal{A}^{R}$, 
there exists 
${\sf a} \in \mathcal{A}$ such that
$X_{t}^{R}\left({\sf a}\right)=X_{t}^{R}\left(\tilde{{\sf a}}\right)$
for all $t\in \mathcal{T}$ almost surely.
\end{description}
\end{lemma}
\begin{proof}[Proof of Lemma \ref{lemma:201809051441}]
\begin{description}
\item[(i)]
Given ${\sf a} \in \mathcal{A}$ and $X^{R}({\sf a})$, we construct $\tilde{{\sf a}}$ as follows:
$\tilde{a}_{0}=a_{0}$, and 
\ba
\label{201809061627}
\tilde{a}_{t}=a_{t} \I_{\left\{X_{t}^{R}\left({\sf a}\right) \in \mathring{\mathcal{X}}_{R}\right\}}
+ a_{t}^{*}\left(X_{t}^{R}\left({\sf a}\right)\right)
\I_{\left\{X_{t}^{R}\left({\sf a}\right) \in \partial \mathcal{X}_{R}\right\}},
\ \ \textup{for}\ \ 
t=1,2,\dots,T-1,
\ea
where
$a_{t}^{*}(x):=\arg \sup_{a \in A_t(x)} f_{t}(x,a)$ for $x \in \partial \mathcal{X}_{R}$ and $t \in \mathcal{T}_{0}$.

It is easy to see from the above construction that $\tilde{\sf a}$ is $\mathcal{F}$-adapted.
It remains to show that 
\ba
\label{201809061421}
\tilde{a}_{t} \in A_{t}\left(X_{t}^{R}\left(\tilde{{\sf a}}\right)\right)
\ \ \textup{and}\ \
X_{t}^{R}({\sf a})=X_{t}^{R}\left(\tilde{{\sf a}}\right),
\ \ \textup{for}\ \
t \in \mathcal{T}.
\ea
Firstly, we observe $\tilde{a}_{0}=a_0 \in A_{0}(X_{0})$ and $X_{0}^{R}\left(\tilde{{\sf a}}\right)=X_0$.
As induction hypothesis, we assume the statement \eqref{201809061421} holds for time step $t$.
For time step $t+1$, we split the discussions into two cases.
\begin{enumerate}
\item
If $X_{t}^{R}({\sf a})=X_{t}^{R}\left(\tilde{{\sf a}}\right) \in \partial \mathcal{X}_{R}$,
then 
\bas
    X_{t+1}^{R}\left(\tilde{{\sf a}}\right)
=   X_{t}^{R}\left(\tilde{{\sf a}}\right)
=   X_{t}^{R}({\sf a})
=   X_{t+1}^{R}({\sf a}) ,
\eas
where the first and third equalities follow by Eq. \eqref{transition_aux_state_process}
and the second equality is due to the induction hypothesis.

\item 
In the second case that $X_{t}^{R}({\sf a})=X_{t}^{R}\left(\tilde{{\sf a}}\right) \in \mathring{\mathcal{X}}_{R}$,
we apply Eq. \eqref{transition_aux_state_process} to get
\bas
X_{t+1}^{R}\left(\tilde{{\sf a}}\right)
= \tilde{H}\Big(K\left(X_{t}^{R}\left(\tilde{{\sf a}}\right),\tilde{a}_{t}\right),\varepsilon_{t+1}\Big)
= \tilde{H}\Big(K\left(X_{t}^{R}\left({\sf a}\right),a_{t}\right),\varepsilon_{t+1}\Big)
= X_{t+1}^{R}\left({\sf a}\right),
\eas
where the second equality follows by Eq. \eqref{201809061627} and the induction hypothesis \eqref{201809061421}.
\end{enumerate}
In either of the above cases, we have $X_{t+1}^{R}\left(\tilde{{\sf a}}\right)=X_{t+1}^{R}\left({\sf a}\right)$.
This combined with Eq. \eqref{201809061627} implies
\bas
\tilde{a}_{t+1}=a_{t+1} \in 
A_{t+1}\left(X_{t+1}^{R}\left({\sf a}\right)\right)
=A_{t+1}\left(X_{t+1}^{R}\left(\tilde{{\sf a}}\right)\right),
\ \ \textup{if}\ \ 
X_{t+1}^{R}\left(\tilde{{\sf a}}\right) \in \mathring{\mathcal{X}}_{R}.
\eas
Otherwise, 
$\tilde{a}_{t+1}=a_{t+1}^{*}\left(X_{t+1}^{R}\left({\sf a}\right)\right)=
a_{t+1}^{*}\left(X_{t+1}^{R}\left(\tilde{{\sf a}}\right)\right)
\in A_{t+1}\left(X_{t+1}^{R}\left(\tilde{{\sf a}}\right)\right)$.
This proves the statement \eqref{201809061421} for time step $t+1$.
The proof of Part (i) is complete.

\item[(ii)]
Given $\tilde{{\sf a}} \in \mathcal{A}^{R}$ and $X^{R}\left(\tilde{{\sf a}}\right)$, 
we construct ${\sf a}$ as follows: $a_{0}=\tilde{a}_{0}$, and 
\ba
\label{201809081627}
a_{t}=\tilde{a}_{t} \I_{\left\{X_{t}^{R}\left(\tilde{{\sf a}}\right) \in \mathring{\mathcal{X}}_{R}\right\}}
+ \hat{a}_{t}\left(X_{t}\left({\sf a}\right)\right)
\I_{\left\{X_{t}^{R}\left(\tilde{{\sf a}}\right) \in \partial \mathcal{X}_{R}\right\}},
\ \ \textup{for}\ \ 
t=1,2,\dots,T-1,
\ea
where $\hat{a}_{t}(\cdot)$ is any measurable function satisfying
$\hat{a}_{t}(x) \in A_{t}(x)$ for $x \in \mathcal{X}$ and $t \in \mathcal{T}_{0}$.

It is easy to see that $\tilde{\sf a}$ is $\mathcal{F}$-adapted.
Next, we use a forward induction argument to show that
\ba
\label{201809081421}
a_{t} \in A_{t}\left(X_{t}\left({\sf a}\right)\right)
\ \ \textup{and}\ \
X_{t}^{R}({\sf a})=X_{t}^{R}\left(\tilde{{\sf a}}\right),
\ \ \textup{for}\ \
t \in \mathcal{T}.
\ea
The above statement holds trivially for $t=0$.
As induction hypothesis, we assume it holds for time step $t$.
For time step $t+1$, we consider two separate cases.
\begin{enumerate}
\item
If $X_{t}^{R}({\sf a})=X_{t}^{R}\left(\tilde{{\sf a}}\right) \in \partial \mathcal{X}_{R}$,
Eq. \eqref{transition_aux_state_process} in combined with \eqref{201809081421} implies
\bas
    X_{t+1}^{R}\left({\sf a}\right)
=   X_{t}^{R}\left({\sf a}\right)
=   X_{t}^{R}(\tilde{{\sf a}})
=   X_{t+1}^{R}(\tilde{{\sf a}}).
\eas

\item 
In the second case that $X_{t}^{R}({\sf a})=X_{t}^{R}\left(\tilde{{\sf a}}\right) \in \mathring{\mathcal{X}}_{R}$,
applying Eq. \eqref{transition_aux_state_process} gives
\bas
X_{t+1}^{R}\left({\sf a}\right)
= \tilde{H}\Big(K\left(X_{t}^{R}\left({\sf a}\right),a_{t}\right),\varepsilon_{t+1}\Big)
= \tilde{H}\Big(K\left(X_{t}^{R}\left(\tilde{{\sf a}}\right),\tilde{a}_{t}\right),\varepsilon_{t+1}\Big)
= X_{t+1}^{R}\left(\tilde{{\sf a}}\right),
\eas
where the second equality follows by Eq. \eqref{201809081627} and the induction hypothesis \eqref{201809081421}.
\end{enumerate}
Overall, we always observe $X_{t+1}^{R}\left({\sf a}\right)=X_{t+1}^{R}\left(\tilde{{\sf a}}\right)$.
To prove the statement \eqref{201809081421} holds for time step $t+1$,
it remains to show 
$
a_{t+1}
\in A_{t+1}\left(X_{t+1}\left({\sf a}\right)\right).
$
We split the discussion into two separate cases.
\begin{enumerate}
\item Firstly, suppose 
$
X_{t+1}^{R}\left({\sf a}\right)=X_{t+1}^{R}\left(\tilde{{\sf a}}\right) \in \partial \mathcal{X}_{R},
$
Eq. \eqref{201809081627} implies 
\bas
a_{t+1} =  \hat{a}_{t+1}\left(X_{t+1}\left({\sf a}\right)\right)
\in A_{t+1}\left(X_{t+1}\left({\sf a}\right)\right).
\eas

\item Secondly, suppose 
$
X_{t+1}^{R}\left({\sf a}\right)=X_{t+1}^{R}\left(\tilde{{\sf a}}\right) \in \mathring{\mathcal{X}}_{R}.
$
By Part (i) of Lemma 1, 
$\left\{X_{t+1}^{R}\left({\sf a}\right) \in \mathring{\mathcal{X}}_{R}\right\}=
\left\{\tau^{R}>t\right\}$ and thus, it follows from Eq. \eqref{aux_state_process} that
$\left\{X_{t+1}^{R}\left({\sf a}\right) \in \mathring{\mathcal{X}}_{R}\right\} 
 \subseteq
\left\{X_{t+1}^{R}\left({\sf a}\right)= X_{t+1}\left({\sf a}\right)\right\}$.
Consequently, we apply Eq. \eqref{201809081627} to get
\bas
a_{t+1}=\tilde{a}_{t+1}
\in A_{t+1}\left(X_{t+1}^{R}\left(\tilde{{\sf a}}\right)\right)
= A_{t+1}\left(X_{t+1}^{R}\left({\sf a}\right)\right)
= A_{t+1}\left(X_{t+1}\left({\sf a}\right)\right).
\eas
\end{enumerate}
The proof of Part (ii) is complete.

\end{description}
\end{proof}

A direct consequence of the preceding lemma is the following corollary.
\begin{coro}
\label{coro:20180914}
The value function $\tilde{V}_{0}(X_{0})$ defined in Eq. \eqref{aux_stochastic_control_problem} exhibits:
\ba
\label{aux_stochastic_control_problem-2}
\tilde{V}_{0}(X_{0})=\sup_{{\sf a} \in \mathcal{A}} 
\e \left[\sum_{t=0}^{T-1} \varphi^{t} f_{t}\left(X_{t}^{R},a_{t}\right) + \varphi^{T} G\left(X_{T}^{R}\right)\right].
\ea
\end{coro}
It is worth noting that the optimization problems in Eq. \eqref{aux_stochastic_control_problem-2} 
and Eq. \eqref{aux_stochastic_control_problem}
are taken over the set $\mathcal{A}$ and $\mathcal{A}^{R}$, respectively.
The above corollary states that the optimal values of these two optimization problems are exactly the same as given by $\tilde{V}_{0}(X_{0})$.

\subsubsection{Proof of the Main Result}
\begin{proof}[Proof of Theorem \ref{thm:truncation_error_bound}]
In view of Eqs. \eqref{stochastic_control_problem} and \eqref{aux_stochastic_control_problem-2}, we obtain
\ba
\nonumber
\left|\tilde{V}_{0}(X_0)-V_{0}(X_0)\right|
&=& \sup_{{\sf a}\in \mathcal{A}} 
\e \left[
\sum_{t=0}^{T-1} 
\left|f_{t}\left(X_{t}^{R},a_{t}\right)-f_{t}(X_{t},a_{t})\right|
\I_{\{X_{t}^{R} \neq X_{t}\}}\right]\\
\nonumber
&&+\sup_{{\sf a}\in \mathcal{A}} 
\e \left[\left|G\left(X_{T}^{R}\right)-G(X_{T})\right|\I_{\{X_{T}^{R} \neq X_{T}\}}\right]\\
\label{201808121709}
&:=&I_1 + I_2.
\ea
Below we establish upper bounds for the $I_1$ and $I_2$ defined in the above display, respectively. 
Let $E:=\left\{X_{t}=X_{t}^{R} \ \ \textup{for all}\ \ 1\leq t \leq T\right\}$.
Note that 
\bas
E \subseteq \left\{X_{t} = X_{t}^{R}\right\}  \Longrightarrow
\left\{X_{t}\neq X_{t}^{R}\right\}=\left\{X_{t} = X_{t}^{R}\right\}^{c} \subseteq E^{c}.
\eas
Accordingly, we get
\ba
\nonumber
I_1&=&\sup_{{\sf a}\in \mathcal{A}} 
\e \left[
\sum_{t=0}^{T-1} 
\left|f_{t}\left(X_{t}^{R},a_{t}\right)-f_{t}(X_{t},a_{t})\right|
\I_{\left\{X_{t}\neq X_{t}^{R}\right\}}\right]\\
\nonumber
&\leq& \sup_{{\sf a}\in \mathcal{A}} 
\e \left[
\sum_{t=0}^{T-1} 
\left(\left|f_{t}\left(X_{t}^{R},a_{t}\right)\right|+\left|f_{t}(X_{t},a_{t})\right|\right)
\I_{\left\{X_{t}\neq X_{t}^{R}\right\}}\right]\\
\nonumber
&\leq& \sup_{{\sf a}\in \mathcal{A}} 
\e \left[
\sum_{t=0}^{T-1} 
\left(\left|f_{t}\left(X_{t}^{R},a_{t}\right)\right|+\left|f_{t}(X_{t},a_{t})\right|\right)
\I_{E^{c}}\right]\\
&\leq& \sup_{{\sf a}\in \mathcal{A}} 
\e \left[
\sum_{t=0}^{T-1} 
\left(\xi^{\frac{1}{2}}(R)+B^{\frac{1}{2}}(X_{t})\right)
\I_{E^{c}}\right]=
\sup_{{\sf a}\in \mathcal{A}} 
\e \left[\left(\sum_{t=0}^{T-1} Y_{t}\right)\I_{E^{c}}\right],
\ea
with $Y_{t}:=\xi^{\frac{1}{2}}(R)+B(X_{t})^{\frac{1}{2}}$,
where the first inequality is by triangular inequality and Assumption \ref{assum:growth_condition}
and the third inequality is due to Part (ii) of Assumption \ref{assum:growth_condition}.
Applying Cauchy--Schwarz inequality twice gives
\ba
I_1
\nonumber
&\leq& 
\sup_{{\sf a}\in \mathcal{A}}
\left\{
\e [\I_{E^{c}}] \cdot
\e \left[\left(\sum_{t=0}^{T-1} Y_{t}\right)^{2}\right]
\right\}^{\frac{1}{2}}\\
\nonumber
&\leq& (T-1)^{\frac{1}{2}}\cdot 
\sup_{{\sf a}\in \mathcal{A}}
\left\{
\e [\I_{E^{c}}] \cdot
\e \left[
\sum_{t=0}^{T-1} Y_{t}^2\right]
\right\}^{\frac{1}{2}}\\
\nonumber
&\leq& (T-1)^{\frac{1}{2}}\cdot 
\sup_{{\sf a}\in \mathcal{A}}
\left\{
\e [\I_{E^{c}}] \cdot
\e \left[
2\sum_{t=0}^{T-1} 
\big(\xi(R)+B(X_{t})\big)\right]
\right\}^{\frac{1}{2}}\\
\nonumber
&\leq& \sqrt{2}(T-1)^{\frac{1}{2}}\cdot 
\sup_{{\sf a}\in \mathcal{A}}
\left\{
\e [\I_{E^{c}}] \cdot
\sum_{t=0}^{T-1} 
\big(
\e \left[
\xi(R)\right]+
\e \left[
B(X_{t})\right]\big)
\right\}^{\frac{1}{2}},
\ea
where the third inequality follows because
$(a+b)^2 \leq 2a^2 + 2b^2$ for two real numbers $a$ and $b$.
In view of Assumption \ref{assum:growth_condition}, we get
\bas
\sum_{t=0}^{T-1} \big(\e \left[\xi(R)\right]+\e \left[B(X_{t})\right]\big)
\leq (T-1) \left(\xi(R) + \sup_{{\sf a}\in \mathcal{A}}\e \left[B(X_{t})\right]\right)
\leq (T-1) \big(\xi(R) + \zeta \big).
\eas
Combing the last two displays with Assumption \ref{assum:tail_prob} implies
\ba
\nonumber
I_1 &\leq& \sqrt{2}(T-1) \big(\xi(R) + \zeta \big)^{\frac{1}{2}} 
\left(\sup_{{\sf a}\in \mathcal{A}} \e [\I_{E^{c}}]\right)^{\frac{1}{2}}\\
\nonumber
&=&  \sqrt{2}(T-1) \big(\xi(R) + \zeta \big)^{\frac{1}{2}} 
\left(1- \inf_{{\sf a}\in \mathcal{A}} \e [\I_{E}]\right)^{\frac{1}{2}}\\
\label{201808121703}
&\leq& (T-1) \sqrt{2\big(\xi(R) + \zeta \big) \mathcal{E}(X_0,R)}.
\ea

A similar argument gives
\ba
\label{201808121708}
I_2 \leq \sqrt{2\big(\xi(R) + \zeta \big) \mathcal{E}(X_0,R)}.
\ea
Combining \eqref{201808121709}, \eqref{201808121703}, and \eqref{201808121708} together implies
\bas
\big|V_{0}(X_0)-\tilde{V}_{0}(X_0)\big| \leq T 
\sqrt{2 \big(\xi(R) + \zeta \big) \mathcal{E}(X_0,R)}.
\eas
The proof is complete.
\end{proof}

\subsection{Proof of Theorem \ref{thm:LSMC_error}}
\label{app:proof_thm2}
\subsubsection{Preliminary lemmas}
We first give the definitions of ``Big O p" and ``Small O p" notations which are commonplaces in statistical literature. 
\begin{defn}
\label{def:big_O_P}
\begin{description}
\item[(i)]
For two sequences of random variables $\{a_{M}\}_{M\in \N}$ and $\{b_{M}\}_{M\in \N}$ indexed by $M$,
we say $a_{M}=O_{\p}(b_{M})$ if 
$
\lim_{k\rightarrow\infty} \limsup_{M \rightarrow \infty} \p \left(|a_{M}|>kb_{M}\right)=0.
$

\item[(ii)]
Moreover, we say $a_{M}=o_{\p}(b_{M})$ if
$
\limsup_{M \rightarrow \infty} \p \left(|a_{M}|>kb_{M}\right)=0
$
for all $k>0$.
\end{description}
\end{defn}

\textit{Some Matrices}\quad Let
$
\mathbf{h}_{t}(x)=\left(
\sup\limits_{a \in A_{t}(x)}\phi_{1}\big(K(x,a)\big), \dots,
\sup\limits_{a \in A_{t}(x)}\phi_{J}\big(K(x,a)\big)
\right)^{\T},
$
for $x\in\col\left(\mathcal{X}_{R}\right)$,
and we suppress its dependency on $J$.
Define matrices
\bas
\Psi_{t}=\e\left[\mathbf{h}_{t}\left(X_{t}^{(m)}\right)\mathbf{h}_{t}^{\T}\left(X_{t}^{(m)}\right)\right]
\ \ \textup{and}\ \
\hat{\Psi}_{t}=\frac{1}{M} \sum_{m=1}^{M}
\mathbf{h}_{t}\left(X_{t}^{(m)}\right)\mathbf{h}_{t}^{\T}\left(X_{t}^{(m)}\right)
\eas
for $t=1,2,\dots,T-1$
with the superscript $\T$ denoting vector transpose.
It is palpable that $\hat{\Psi}_{t}$ is a finite-sample estimate for $\Psi_{t}$.
In the sequel, we denote $\lambda_{\max}(B)$ (resp. $\lambda_{\min}(B)$) 
as the largest (resp. smallest) eigenvalue of a square matrix $B$.
We impose the following Assumption on the eigenvalues of $\Psi_{t}$.
\begin{assumption}
\label{assum:eigen_psi}
\begin{description}
\item[(i)]
For any fixed $x$ and $t$, $A_{t}(x)$ is a compact set.
Moreover, $a \longmapsto K(x,a)$ and $\phi_{j}(\cdot): \R^{r} \longrightarrow \R$ 
are continuous functions for $1 \leq j \leq J$.

\item[(ii)]
There exists a positive constant $\bar{c}_{\Psi}$ independent of $t$ and $J$
such that $\lambda_{\max}\left(\Psi_{t}\right)\leq \bar{c}_{\Psi}<\infty$. 
\end{description}
\end{assumption}
Part (i) of the preceding assumption guarantees that the function $\mathbf{h}_{t}(\cdot)$
is well-defined for $t \in \mathcal{T}_{0}$.
The continuity requirement of $a \longmapsto K(x,a)$ can be removed
if $A_{t}(x)$ is a lattice (discrete set), which is particularly the case
when the stochastic optimal control problem exhibits the Bang-bang solution, 
see, e.g., \cite{Azimzadeh2015} and \cite{Huang2016}.
Part (ii) requires the largest eigenvalue of the matrix $\hat{\Psi}_{t}$
does not blow up as $M$ and $J$ approach infinity.
This condition ensures the sample eigenvalue converges to the non-sample counterpart 
as $M$ approaches infinity as shown in the sequel Lemma \ref{lemma:201810021158}.

Moreover, we define matrices
\bas
\Phi_{t}=\e\left[\bm{\phi}\left(X_{t^{+}}^{(m)}\right)\bm{\phi}^{\T}\left(X_{t^{+}}^{(m)}\right)\right]
\ \ \textup{and}\ \
\hat{\Phi}_{t}=\frac{1}{M} \sum_{m=1}^{M}
\bm{\phi}\left(X_{t^{+}}^{(m)}\right)\bm{\phi}^{\T}\left(X_{t^{+}}^{(m)}\right).
\eas
The following lemma relates the eigenvalues of $\hat{\Phi}_{t}$ and $\hat{\Psi}_{t}$
to those of $\Phi_{t}$ and $\Psi_{t}$.
\begin{lemma}
\label{lemma:201810021158}
\begin{description}
\item[(i)]
Suppose Condition (ii) of Theorem \ref{thm:LSMC_error} is satisfied.
Then, 
\bas
\left|\lambda_{\max}\left(\Phi_{t}\right)-\lambda_{\max}\left(\hat{\Phi}_{t}\right)\right|
=O_{\p}\left(\Upsilon(J)\sqrt{J/M}\right),
\eas
and
\bas
\left|\lambda_{\min}\left(\Phi_{t}\right)-\lambda_{\min}\left(\hat{\Phi}_{t}\right)\right|
=O_{\p}\left(\Upsilon(J)\sqrt{J/M}\right),
\eas
for $t \in \mathcal{T}_{0}$.

\item[(ii)]
Suppose Assumption \ref{assum:eigen_psi} holds.
In addition, Condition (v) of Assumption \ref{assum:sieve_estimation} is satisfied.
Then, $\lambda_{\max} \left(\hat{\Psi}_{t}\right)=O_{\p}(1)$
for $t=1,2,\dots,T-1$.
\end{description}
\end{lemma}
\begin{proof}[Proof of Lemma \ref{lemma:201810021158}]
Lemma \ref{lemma:201810021158} can be proved by a similar argument as 
that used in the proof of Eq. (A.1) in \cite{Newey1997}.
\end{proof}
The above lemma shows the sample eigenvalues converge to the non-sample counterparts
as $M$ approaches infinity.
In view of Condition (iv) of Assumption \ref{assum:sieve_estimation},
Lemma \ref{lemma:201810021158} also implies the largest (resp., smallest) eigenvalue of $\hat{\Phi}_{t}$
is bounded from above (resp., below) with probability approaching 1 as $M \longrightarrow \infty$.
This fact is exploited in the proofs of sequel Lemmas \ref{lemma:201809281926} and \ref{lemma:201809282016}.

\textit{Pseudo Estimate, Oracle, and True Estimate}\quad 
Next, we introduce the concept of \textit{pseudo estimate}.
Let $\bar{\bm{\beta}}_{t}$ (resp. $\hat{\bm{\beta}}_{t}$) 
be the solution to the optimization problem in Eq. \eqref{sieve_estimator}
with $U^{(m)}=\tilde{V}_{t+1}\left(X_{t+1}^{(m)}\right)$ 
(resp. $\tilde{V}_{t+1}^{\textup{E}}\left(X_{t+1}^{(m)}\right)$) and $Z^{(m)}=X_{t^{+}}^{(m)}$.
Given $\bar{\bm{\beta}}_{t}$ and $\hat{\bm{\beta}}_{t}$, 
denote the associated regression estimates by 
$\tilde{C}_{t}^{\textup{PE}}(\cdot)=\bar{\bm{\beta}}_{t}^{\T} \bm{\phi}(\cdot)$
and $\tilde{C}_{t}^{\textup{E}}(\cdot)=\hat{\bm{\beta}}_{t}^{\T} \bm{\phi}(\cdot)$, respectively.
$\tilde{C}_{t}^{\textup{PE}}(\cdot)$ is essentially the sieve estimate for the continuation function $\tilde{C}_{t}(\cdot)$
when the true value function $\tilde{V}_{t+1}(\cdot)$ is employed in the regression.
We further define function $\tilde{V}_{t}^{\textup{PE}}(x)$ for $x \in \mathring{\mathcal{X}}_{R}$ 
by substituting $\tilde{C}_{t}^{\textup{E}}(\cdot)$ in Eq. \eqref{Bellman_eq-4} with $\tilde{C}_{t}^{\textup{PE}}(\cdot)$.
For $x \in \partial \mathcal{X}_{R}$, we set $\tilde{V}_{t}^{\textup{PE}}\left(x\right)=\tilde{V}_{t}(x)$ 
with $\tilde{V}_{t}(\cdot)$ given by Eq. \eqref{aux_value_fun_boundary}.

Admittedly, in the implementation of the BSBU algorithm, 
$\bar{\bm{\beta}}_{t}$ is not tractable because
the true value function is unknown and should be replaced by the numerical estimate $\tilde{V}_{t+1}^{\textup{E}}(\cdot)$ obtained inductively.
For this reason, following \cite{Belomestny2010}, we call $\bar{\bm{\beta}}_{t}$ the \textit{pseudo estimate}.
Despite this, the pseudo estimate plays an indispensable role in establishing the convergence result of Theorem \ref{thm:LSMC_error}.
In addition to the two estimates $\bar{\bm{\beta}}_{t}$ and $\hat{\bm{\beta}}_{t}$ defined in the above,
we further define the \textit{oracle} $\tilde{\bm{\beta}}_{t}$ 
as the solution to the optimization problem \eqref{oracle}
with $g(\cdot)$ replaced by $\tilde{C}_{t}(\cdot)$.

The following lemma discloses that the gap between pseudo estimate and the oracle vanishes 
when both $M$ and $J$ increase at a certain rate.
\begin{lemma}
\label{lemma:201809281926}
Suppose the conditions of Theorem \ref{thm:LSMC_error} are satisfied.
Then,
\bas
\norm{\bar{\bm{\beta}}_{t}-\tilde{\bm{\beta}}_{t}}=O_{\p}\left(\sqrt{J/M}+\rho_{J}\right),
\ \ \textup{for}\ \ t\in \mathcal{T}_{0}.
\eas
\end{lemma}
\begin{proof}[Proof of Lemma \ref{lemma:201809281926}]
Recall that $\bar{\bm{\beta}}_{t}$ solves the optimization problem:
\bas
\min\limits_{\bm{\beta} \in \R^{J}}
\frac{1}{M} \sum_{m=1}^{M}\left[\tilde{V}_{t+1}\left(X_{t+1}^{(m)}\right)
- \bm{\beta}^{\T} \bm{\phi} \left(X_{t^{+}}^{(m)}\right)\right]^2,
\ \ \textup{subject to}\ \ \bm{\beta}^{\T} \bm{\phi}(\cdot) \in \mathcal{H}_{J}.
\eas
On the other hand, $\tilde{\bm{\beta}}_{t}$ is a suboptimal solution to the above optimization problem.
Therefore, we get
\bas
\norm{\mathbf{V}_{t+1}-P\bar{\bm{\beta}}_{t}}^2 
\leq \norm{\mathbf{V}_{t+1}-P\tilde{\bm{\beta}}_{t}}^2,
\eas
where $P$ is a $M$-by-$J$ matrix with $m$-th row being $\bm{\phi}^{\T} \left(X_{t^{+}}^{(m)}\right)$
and $\mathbf{V}_{t+1}$ is a $M$-by-1 vector with $m$-th element given by $\tilde{V}_{t+1}\left(X_{t+1}^{(m)}\right)$.

By adding and subtracting the term $P\tilde{\bm{\beta}}_{t}$ in the L.H.S. of the above inequality, 
we get
\bas
\norm{\bar{\mathbf{U}}-P\bar{\bm{\delta}}}^2 
\leq \norm{\bar{\mathbf{U}}}^2,
\eas
where we use the shorthand notations $\bar{\bm{\delta}}:=\bar{\bm{\beta}}_{t}-\tilde{\bm{\beta}}_{t}$
and $\bar{\mathbf{U}}:=\mathbf{V}_{t+1}-P\tilde{\bm{\beta}}_{t}$.
Expanding both sides of the above inequality gives
\bas
\frac{\norm{P\bar{\bm{\delta}}}^2}{2M}
\leq \frac{\left|\bar{\mathbf{U}}^{\T}P\bar{\bm{\delta}}\right|}{M}
\leq \frac{\norm{P^{\T}\bar{\mathbf{U}}}\norm{\bar{\bm{\delta}}}}{M},
\eas
where the second inequality is by H\"older's inequality.
For the L.H.S. of the above inequality, it follows from the definition of 
the smallest eigenvalue that
\bas
    \frac{\norm{P\bar{\bm{\delta}}}^2}{2M}
=   \frac{\bar{\bm{\delta}}^{\T}P^{\T}P\bar{\bm{\delta}}}{2M}
\geq \frac{\norm{\bar{\bm{\delta}}}^2}{2}\lambda_{\min}\left(\hat{\Phi}_{t}\right).
\eas
Combing the last two inequalities together implies
\bas
\norm{\bar{\bm{\delta}}}\lambda_{\min}\left(\hat{\Phi}_{t}\right)
\leq \frac{2}{M}\norm{P^{\T}\bar{\mathbf{U}}}.
\eas
It follows from Lemma \ref{lemma:201810021158} that 
the event $\left\{\underline{c}_{\Phi}/2\leq \lambda_{\min}\left(\hat{\Phi}_{t}\right)\right\}$ holds 
with probability approaching $1$ as $M \longrightarrow \infty$.
And therefore, 
\ba
\label{201810151642}
\norm{\bar{\bm{\delta}}}
\leq \big(4/\underline{c}_{\Phi}\big)M^{-1}\norm{P^{\T}\bar{\mathbf{U}}}
\ea
holds with probability approaching $1$ as $M \longrightarrow \infty$.

It follows as in Eq. (A.2) of \citet[pp. 163]{Newey1997} that
$
M^{-1}\norm{P^{\T}\bar{\mathbf{U}}}=O_{\p}\left(\sqrt{J/M}+\rho_{J}\right).
$
This in conjunction with the last display proves the desired result.
The proof is complete.
\end{proof}

The next lemma relates the discrepancy between the pseudo estimate $\bar{\bm{\beta}}_{t}$
and the true estimate $\hat{\bm{\beta}}_{t}$ to the estimation error of the value function at the previous time step.
\begin{lemma}
\label{lemma:201809282016}
Suppose the conditions of Theorem \ref{thm:LSMC_error} are satisfied.
Then, for $t\in \mathcal{T}_{0}$,
there exists a constant $\psi>0$ independent of $t$, $R$ and $J$ such that
\bas
\norm{\bar{\bm{\beta}}_{t}-\hat{\bm{\beta}}_{t}} \leq \sqrt{\frac{\psi}{M}}
\norm{\mathbf{V}_{t+1}-\hat{\mathbf{V}}_{t+1}} +O_{\p}\left(\sqrt{\psi} \rho_{J}\right)
\eas
holds with probability approaching 1 as $M \longrightarrow \infty$,
where $\mathbf{V}_{t+1}$ and $\hat{\mathbf{V}}_{t+1}$ are two $M$-by-1 vectors
with $m$-th element given by $\tilde{V}_{t+1}\left(X_{t+1}^{(m)}\right)$ and 
$\tilde{V}_{t+1}^{\textup{E}}\left(X_{t+1}^{(m)}\right)$, respectively.
\end{lemma}
\begin{proof}[Proof of Lemma \ref{lemma:201809282016}]
Using the argument as in the proof of inequality \eqref{201810151642}, we obtain
\bas
\norm{\hat{\bm{\delta}}}\leq \big(4/\underline{c}_{\Phi}\big)M^{-1}\norm{P^{\T}\hat{\mathbf{U}}},
\eas
holds with probability approaching 1 as $M \longrightarrow \infty$,
where we adopt shorthand notations 
$
\hat{\bm{\delta}}:=\bar{\bm{\beta}}_{t}-\hat{\bm{\beta}}_{t}
$
and 
$
\hat{\mathbf{U}}:=\hat{\mathbf{V}}_{t+1}-P\bar{\bm{\beta}}_{t}.
$
On the other hand, it follows from Lemma \ref{lemma:201810021158} that
\bas
M^{-2}\norm{P^{\T}\hat{\mathbf{U}}}^2
=M^{-1}\hat{\mathbf{U}}^{\T}\left(M^{-1}PP^{\T}\right)\hat{\mathbf{U}}
\leq M^{-1} \lambda_{\max} \left(\hat{\Phi}_{t}\right)\norm{\hat{\mathbf{U}}}^2
\leq M^{-1} 2\bar{c}_{\Phi}\norm{\hat{\mathbf{U}}}^2
\eas
holds with probability approaching 1 as $M \longrightarrow \infty$.

Combing the above two inequalities implies
\bas
\norm{\hat{\bm{\delta}}} \leq 
\sqrt{\big(32\bar{c}_{\Phi}/\underline{c}_{\Phi}^2\big)M}\norm{\hat{\mathbf{U}}}^2:=
\sqrt{\frac{\psi}{M}}\norm{\hat{\mathbf{U}}}.
\eas
By adding and subtracting the term $\mathbf{V}_{t+1}$ in the R.H.S. of the above inequality, we get
\bas
\norm{\hat{\mathbf{U}}}
&=& \norm{\hat{\mathbf{V}}_{t+1}-\mathbf{V}_{t+1}+\mathbf{V}_{t+1} - P\bar{\bm{\beta}}_{t}}\\
&\leq&  \norm{\hat{\mathbf{V}}_{t+1}-\mathbf{V}_{t+1}} + 
\norm{\mathbf{V}_{t+1} - P\bar{\bm{\beta}}_{t}}\\
&=& \norm{\hat{\mathbf{V}}_{t+1}-\mathbf{V}_{t+1}} + 
O\left(\sqrt{M}\rho_{J}\right)
\eas
where the last equality is guaranteed by Part (ii) of Assumption \ref{assum:sieve_estimation}.
Combing the last two inequalities implies
\bas
\norm{\hat{\bm{\delta}}} \leq 
\sqrt{\frac{\psi}{M}}\norm{\hat{\mathbf{V}}_{t+1}-\mathbf{V}_{t+1}} 
+ O_{\p}\left(\sqrt{\psi}\rho_{J}\right).
\eas
holds with probability approaching 1 as $M \longrightarrow \infty$.
This proves Lemma \ref{lemma:201809282016}.
\end{proof}

The statement of the above lemma is not hard to expect because
the primary difference between the pseudo estimate and the true estimate
stems from the the estimation error of value function.

The final lemma quantifies the discrepancy between the value function and its numerical estimate
under the empirical $L^2$ norm.
\begin{lemma}
\label{lemma:L2_gap}
Let $F_{t}^{X}(\cdot)$ be the probability distribution function of $X_{t}^{(m)}$ for $t=1,2,\dots,T-1$.
Suppose the assumptions of Theorem \ref{thm:LSMC_error} hold. Then
\ba
\label{induction_hypothesis-2}
M^{-1}\norm{\mathbf{V}_{t}-\hat{\mathbf{V}}_{t}}^2=O_{\p}\left(\psi^{T-t-1}\left(J/M + \rho_{J}^2\right)\right),
\ \ \textup{for}\ \ 
t=1,2,\dots,T-1,
\ea
where
$\mathbf{V}_{t}$ and $\hat{\mathbf{V}}_{t}$ are two $M$-by-$1$ vectors
with $m$-th element being $\tilde{V}_{t}\left(X_{t}^{(m)}\right)$ and $\tilde{V}_{t}^{\textup{E}}\left(X_{t}^{(m)}\right)$, respectively.
\end{lemma}
\begin{proof}[Proof of Lemma \ref{lemma:L2_gap}]
We use a backward induction procedure to prove the statement of Lemma \ref{lemma:L2_gap}.
For $t=T-1$, we note that $\tilde{C}_{T-1}^{\textup{E}}(\cdot)$ is in agreement with $\tilde{C}_{T-1}^{\textup{PE}}(\cdot)$
because $\tilde{V}_{T}^{\textup{E}}(x)=\tilde{V}_{T}(x)=G(x)$ for $x \in \col \left(\mathcal{X}_{R}\right)$.
We get $\tilde{V}_{T-1}^{\textup{E}}(x)=\tilde{V}_{T-1}^{\textup{PE}}(x)$ for $x \in \col \left(\mathcal{X}_{R}\right)$, accordingly. Furthermore, we observe that
\ba
\nonumber
    \left|\tilde{V}_{T-1}^{\textup{E}}(x)-\tilde{V}_{T-1}(x)\right|
&=&   \left|\tilde{V}_{T-1}^{\textup{PE}}(x)-\tilde{V}_{T-1}(x)\right|\\
\nonumber
&\leq&  \sup_{a \in A_{T-1}(x)}\left|\tilde{C}_{T-1}^{\textup{PE}}\big(K(x,a)\big)-\tilde{C}_{T-1}\big(K(x,a)\big)\right|\\
\nonumber
&=&  \sup_{a \in A_{T-1}(x)} \left|\bar{\bm{\beta}}_{T-1}^{\T}\bm{\phi}\big(K(x,a)\big)-\tilde{C}_{T-1}\big(K(x,a)\big)\right|\\
\nonumber
&\leq& \sup_{a \in A_{T-1}(x)} \left|\left(\bar{\bm{\beta}}_{T-1}-\tilde{\bm{\beta}}_{T-1}\right)^{\T}\bm{\phi}\big(K(x,a)\big)\right|\\
\nonumber
&&+ \sup_{a \in A_{T-1}(x)} \left|\tilde{\bm{\beta}}_{T-1}^{\T}\bm{\phi}\big(K(x,a)\big)-\tilde{C}_{T-1}\big(K(x,a)\big)\right|\\
\nonumber
&\leq& \left|\left(\bar{\bm{\beta}}_{T-1}-\tilde{\bm{\beta}}_{T-1}\right)^{\T}\mathbf{h}_{T-1}(x)\right|
+\norm{\tilde{\bm{\beta}}_{T-1}^{\T}\bm{\phi}-\tilde{C}_{T-1}}_{\infty}\\
\label{201810021314}
&=& \left|\left(\bar{\bm{\beta}}_{T-1}-\tilde{\bm{\beta}}_{T-1}\right)^{\T}\mathbf{h}_{T-1}(x)\right| + O(\rho_{J}),
\ea
where the third inequality is by the definition of function $\mathbf{h}_{T-1}(\cdot)$ 
and the last equality is guaranteed by Assumption \ref{assum:sieve_estimation}.

Consequently, we obtain
\ba
\nonumber
    M^{-1}\norm{\mathbf{V}_{T-1}-\hat{\mathbf{V}}_{T-1}}^2
&=& \frac{1}{M} \sum_{m=1}^{M} 
\left|\tilde{V}_{T-1}^{\textup{E}}\left(X_{T-1}^{(m)}\right)-\tilde{V}_{T-1}\left(X_{T-1}^{(m)}\right)\right|^2\\
\nonumber
&\leq&\frac{1}{M} \sum_{m=1}^{M} 
\left|\left(\bar{\bm{\beta}}_{T-1}-\tilde{\bm{\beta}}_{T-1}\right)^{\T}\mathbf{h}_{T-1}\left(X_{T-1}^{(m)}\right)\right|^2
+ O\left(\rho_{J}^2\right)\\
\nonumber
&=& \left(\bar{\bm{\beta}}_{T-1}-\tilde{\bm{\beta}}_{T-1}\right)^{\T}
\hat{\Psi}_{T-1}
\left(\bar{\bm{\beta}}_{T-1}-\tilde{\bm{\beta}}_{T-1}\right) + O\left(\rho_{J}^2\right)\\
\nonumber
&\leq& 2 \lambda_{\max} \left(\hat{\Psi}_{T-1}\right)
\norm{\bar{\bm{\beta}}_{T-1}-\tilde{\bm{\beta}}_{T-1}}^2 + O\left(\rho_{J}^2\right)\\
\label{201810021346}
&=& O_{\p} \left(J/M + \rho_{J}^2\right),
\ea
where the second inequality follows from the definition of the largest eigenvalue of a matrix
and the last equality is guaranteed by Lemma \ref{lemma:201809281926} and 
Lemma \ref{lemma:201810021158}.
In view of the above display, Eq. \eqref{induction_hypothesis-2} holds for $t=T-1$.

As induction hypothesis, we assume \eqref{induction_hypothesis-2} holds for $t+1$.
Note that, for $x \in \mathring{\mathcal{X}_{R}}$,
\ba
\label{201810021335}
\left|\tilde{V}_{t}(x)-\tilde{V}_{t}^{\textup{E}}(x)\right|
\leq \left|\tilde{V}_{t}(x)-\tilde{V}_{t}^{\textup{PE}}(x)\right| + \left|\tilde{V}_{t}^{\textup{E}}(x)-\tilde{V}_{t}^{\textup{PE}}(x)\right|.
\ea
An argument similar to the one used in establishing \eqref{201810021346} shows that
\ba
\label{201810021342}
\frac{1}{M} \sum_{m=1}^{M} 
\left|\tilde{V}_{t}^{\textup{PE}}\left(X_{t}^{(m)}\right)-\tilde{V}_{t}\left(X_{t}^{(m)}\right)\right|^2
=O_{\p} \left(J/M + \rho_{J}^2\right).
\ea
Next, we investigate the term $\left|\tilde{V}_{t}^{\textup{E}}(x)-\tilde{V}_{t}^{\textup{PE}}(x)\right|$.
Observe that
\ba
\nonumber
\left|\tilde{V}_{t}^{\textup{E}}(x)-\tilde{V}_{t}^{\textup{PE}}(x)\right|
&\leq&  \sup_{a\in A_{t}(x)} \left|\tilde{C}_{t}^{\textup{E}}\big(K(x,a)\big)-\tilde{C}_{t}^{\textup{PE}}\big(K(x,a)\big)\right|\\
\nonumber
&=&  \sup_{a\in A_{t}(x)} \left|
\left(\hat{\bm{\beta}}_{t}-\bar{\bm{\beta}}_{t}\right)^{\T}
\bm{\phi}\big(K(x,a)\big)
\right|\\
\label{201810021520}
&\leq&  \left|
\left(\hat{\bm{\beta}}_{t}-\bar{\bm{\beta}}_{t}\right)^{\T}
\mathbf{h}_{t}\big(x\big)
\right|.
\ea
We adopt the same argument as in the proof of \eqref{201810021346} to get
\bas
\frac{1}{M} \sum_{m=1}^{M} 
\left|\tilde{V}_{t}^{\textup{E}}\left(X_{t}^{(m)}\right)-\tilde{V}_{t}^{\textup{PE}}\left(X_{t}^{(m)}\right)\right|^2
\leq\lambda_{\max} \left(\hat{\Psi}_{t}\right)\norm{\hat{\bm{\beta}}_{t}-\bar{\bm{\beta}}_{t}}^2.
\eas
Applying Lemma \ref{lemma:201809282016} yields
\bas
\frac{1}{M} \sum_{m=1}^{M} 
\left|\tilde{V}_{t}^{\textup{E}}\left(X_{t}^{(m)}\right)-\tilde{V}_{t}^{\textup{PE}}\left(X_{t}^{(m)}\right)\right|^2
&\leq& 2\lambda_{\max}
\left(\hat{\Psi}_{t}\right)
\left[ 
\frac{\psi}{M} \norm{\mathbf{V}_{t+1}-\hat{\mathbf{V}}_{t+1}}^2
+ O\left(\psi \rho_{J}^2\right)\right]\\
&=&O_{\p}\left(\psi^{T-t-1}\left(J/M + \rho_{J}^2\right)\right),
\eas
where the last equality is due to induction hypothesis \eqref{induction_hypothesis-2}
and $\lambda_{\max} \left(\hat{\Psi}_{t}\right)=O_{\p}(1)$ (see Lemma \ref{lemma:201810021158}).
The above display in conjunction with \eqref{201810021335} and \eqref{201810021342} implies 
\bas
M^{-1}\norm{\mathbf{V}_{t}-\hat{\mathbf{V}}_{t}}^2
&=&O_{\p}\left(J/M + \rho_{J}^2\right)+
O_{\p}\left(\psi^{T-t-1}\left(J/M + \rho_{J}^2\right)\right)\\
&=&O_{\p}\left(\psi^{T-t-1}\left(J/M + \rho_{J}^2\right)\right).
\eas
This completes the proof.
\end{proof}

\subsubsection{Proof of the Main Result}
\begin{proof}[Proof of Theorem \ref{thm:LSMC_error}]
Following the arguments used to prove \eqref{201810021314}, we get
\bas
\left|\tilde{V}_{0}^{\textup{PE}}(X_0)-\tilde{V}_{0}(X_0)\right|
\leq \left|\left(\bar{\bm{\beta}}_{0}-\tilde{\bm{\beta}}_{0}\right)^{\T}\mathbf{h}_{0}(X_{0})\right| + O(\rho_{J})
=O_{\p}\left(\sqrt{J/M} + \rho_{J}\right),
\eas
where the last equality is by Lemma \ref{lemma:201809281926} and Part (ii) of Assumption \ref{assum:eigen_psi}.

On the other hand, an argument similar to the one used in deriving \eqref{201810021520} shows 
\bas
\left|\tilde{V}_{0}^{\textup{PE}}(X_0)-\tilde{V}_{0}^{\textup{E}}(X_0)\right|
\leq \left|
\left(\hat{\bm{\beta}}_{0}-\bar{\bm{\beta}}_{0}\right)^{\T}
\mathbf{h}_{0}(X_{0})
\right|
\leq  M^{-1/2}\norm{\mathbf{V}_{1}-\hat{\mathbf{V}}_{1}}
\norm{\mathbf{h}_{0}(X_{0})}.
\eas
The above two displays in conjunction with \eqref{induction_hypothesis-2} implies
\bas
\left|\tilde{V}_{0}(X_0)-\tilde{V}_{0}^{\textup{E}}(X_0)\right|&\leq&
\left|\tilde{V}_{0}^{\textup{PE}}(X_0)-\tilde{V}_{0}(X_0)\right| + 
\left|\tilde{V}_{0}^{\textup{PE}}(X_0)-\tilde{V}_{0}^{\textup{E}}(X_0)\right|\\
&=&O_{\p} \left(\sqrt{\psi^{T-1}\left(J/M + \rho_J^2\right)}\right).
\eas
This shows \eqref{LSMC_error} and completes the proof of Theorem \ref{thm:LSMC_error}.
\end{proof}

\end{document}